\documentclass[letterpaper,twocolumn,10pt]{article}
\usepackage{usenix}

\usepackage{xspace}
\usepackage{xparse}
\usepackage{amsmath}
\usepackage{amsfonts}
\usepackage{amsthm}
\usepackage{enumitem}
\usepackage{amssymb}
\usepackage{color}
\usepackage{mathtools}
\usepackage{array}
\usepackage[linewidth=1pt]{mdframed}

\usepackage{tikz}
\usepackage{tikz-qtree}
\usepackage{algorithm}
\usepackage[noend]{algpseudocode}
\usepackage[caption=false]{subfig}
\usepackage{tikz-qtree-compat}
\usepackage{tabularx}
\usepackage[square,sort,comma,numbers]{natbib}
\usetikzlibrary{arrows}
\usetikzlibrary{patterns}
\usepackage{stackrel}
\usepackage{pstricks}
\PassOptionsToPackage{gray}{xcolor}
\usepackage{pgfplots, pgfplotstable}
\pgfplotsset{compat=1.9}
\usepackage{framed}
\usepackage{booktabs}
\usepackage{multirow}
\usepackage{expl3}
\usepackage{graphicx}
\usepackage{thmtools,thm-restate}
\usetikzlibrary{calc,positioning,shapes,decorations.pathreplacing,shadows}
\usetikzlibrary{shapes.geometric}
\usepackage{enumitem}
\usepackage{wrapfig}

\algblock{Input}{EndInput}
\algnotext{EndInput}

\algblock{Output}{EndOutput}
\algnotext{EndOutput}
\newtheorem{thm}{Theorem}
\newtheorem{lemma}{Lemma}
\newtheorem{fact}{Fact}

\newtheorem{defn}{Definition}
\newenvironment{proofsketch}{\par\noindent\textsl{Proof (sketch):}\kern1ex}%
{\hfil\hskip2em\penalty250\parfillskip=0pt\finalhyphendemerits=0$\qed$%
	\par\smallskip\par}
	\newcommand{\myparagraph}[1]{\noindent\textbf{#1}:\xspace}
\newcommand{\naive}{na\"ive\xspace}
{\list{}{\leftmargin=0.2in\rightmargin=0.2in}\item[] \em}%
{\endlist}

\newcounter{protocol}

\newcounter{note}[section]

\renewcommand{\thenote}{\thesection.\arabic{note}}

\newcommand{\ac}[1]{\refstepcounter{note}{\bf \textcolor{blue}{$\ll$AC~\thenote: {\sf #1}$\gg$}}}

\newcommand{\secref}[1]{\mbox{Sec.~\ref{#1}}\xspace}

\newcommand{\lineref}[1]{\mbox{line~\ref{#1}}\xspace}
\newcommand{\linesref}[2]{\mbox{lines~\ref{#1}--\ref{#2}}\xspace}

\newcommand{\figref}[1]{\mbox{Fig.~\ref{#1}}\xspace}
\newcommand{\figsref}[2]{\mbox{Figs.~\ref{#1}--\ref{#2}}\xspace}

\newcommand{\tblref}[1]{\mbox{Table~\ref{#1}}\xspace}

\newcommand{\appref}[1]{\mbox{App.~\ref{#1}}\xspace}

\newcommand{\propref}[1]{\mbox{Prop.~\ref{#1}}\xspace}

\newcommand{\lemmaref}[1]{\mbox{Lemma~\ref{#1}}\xspace}
\newcommand{\thmref}[1]{\mbox{Theorem~\ref{#1}}\xspace}

\newcommand{\kilo}{\ensuremath{\mathrm{K}}\xspace}
\newcommand{\mega}{\ensuremath{\mathrm{M}}\xspace}

\newcommand{\kilobytes}{\ensuremath{\mathrm{KB}}\xspace}

\newcommand{\megabytes}{\ensuremath{\mathrm{MB}}\xspace}
\newcommand{\gigabytes}{\ensuremath{\mathrm{GB}}\xspace}

\newcommand{\assign}{\ensuremath{:=}\xspace}

\newcommand{\prob}[1]{\ensuremath{\mathbb{P}\left [#1 \right]}\xspace}

\newcommand{\tpr}{\textsc{tpr}\xspace}
\newcommand{\tnr}{\textsc{tnr}\xspace}
\newcommand{\fpr}{\textsc{fpr}\xspace}
\newcommand{\fnr}{\textsc{fnr}\xspace}

\newcommand{\genericSetSize}{\ensuremath{n}\xspace}

\newcommand{\setSize}[1]{\ensuremath{\left |{#1} \right |}\xspace}

\newcommand{\getsr}{\ensuremath{\;\stackrel{\$}{:=}\;}}

\newcommand{\cset}[3]{\ensuremath{#1\{}{#2}\ensuremath{\;#1|} \ifmmode{\;}\fi {#3}\ensuremath{#1\}}\xspace}
\newcommand{\cprob}[3]{\ensuremath{\mathbb{P}#1(}{#2}\ensuremath{\;#1|} \ifmmode{\;}\fi {#3}\ensuremath{#1)}\xspace}
\newcommand{\cexpv}[3]{\ensuremath{\mathbb{E}#1(}{#2}\ensuremath{\;#1|} \ifmmode{\;}\fi {#3}\ensuremath{#1)}\xspace}

\newcommand{\floor}[1]{\lfloor #1 \rfloor}
\newcommand{\ceil}[1]{\left\lceil #1 \right\rceil}
\DeclareRobustCommand{\stirling}{\genfrac\{\}{0pt}{}}
\newcommand{\polynomial}[1]{\ensuremath{#1(x)}\xspace}

\newcommand{\keygen}{\ensuremath{\mathsf{Gen}}\xspace}
\newcommand{\encrypt}[1]{\ensuremath{\mathsf{Enc}(#1)}\xspace}
\newcommand{\decrypt}[1]{\ensuremath{\mathsf{Dec}(#1)}\xspace}
\NewDocumentCommand{\plaintext}{ g g }{\ensuremath{m\IfNoValueF{#1}{\IfNoValueTF{#2}{_{#1}}{_{{#1},{#2}}}}}\xspace}

\newcommand{\secParam}{\ensuremath{\lambda}\xspace}
\newcommand{\privKey}{\ensuremath{\mathit{sk}}\xspace}
\newcommand{\pubKey}{\ensuremath{\mathit{pk}}\xspace}
\newcommand{\encAdd}{\ensuremath{+_{\pubKey}}\xspace}
\newcommand{\encSub}{\ensuremath{-_{\pubKey}}\xspace}
\newcommand{\encMult}{\ensuremath{\times_{\pubKey}}\xspace}

\NewDocumentCommand{\encSum}{ g g g }{\ensuremath{\displaystyle\operatorname*{\textstyle\sum_{\mathrlap{#1}}}\IfNoValueF{#2}{_{#2}}\IfNoValueF{#3}{^{#3}}\hphantom{_{#1}}}\xspace}

\newcommand{\genericPRF}{\ensuremath{\phi}\xspace}
\newcommand{\AHESchemeDefn}{\ensuremath{\mathsf{AHE} = (\keygen, \mathsf{Enc}, \mathsf{Dec})}\xspace}

\newcommand{\fieldOrder}{\ensuremath{p}\xspace}

\newcommand{\finiteFieldVec}[1]{\ensuremath{{\mathbb F_\fieldOrder^#1}\xspace}}
\def\finiteField{\ensuremath{{\mathbb F_{\fieldOrder}}}\xspace}
\def\IntegersPositive{{\mathbb Z^{+}}}
\def\NaturalNumbers{{\mathbb N}}       
 
\newcommand{\alice}{Alice\xspace}
\newcommand{\bob}{Bob\xspace}
\newcommand{\maxHamDist}{\ensuremath{d_\mathsf{Max}}\xspace}
\newcommand{\intAwarePSIFunc}{\ensuremath{\mathcal{F}_{\intDistThreshold}^{\lowercase{i}{\text -}\mathit{PSI}}}\xspace}
\newcommand{\tHamQueryFunc}{\ensuremath{\mathcal{F}_{\genericVecLen, \hamDistThreshold}^{\lowercase{t}{\text -}\mathit{HQ}}}\xspace}

\newcommand{\hamAwarePSIFunc}{\ensuremath{\mathcal{F}_{\ell, \hamDistThreshold}^{\lowercase{h}{\text -}\mathit{PSI}}}\xspace}
\newcommand{\voleFunc}{\ensuremath{\mathcal{F}^\fieldOrder_{\mathit{vole}}}\xspace}
\newcommand{\oleFunc}{\ensuremath{\mathcal{F}^\fieldOrder_{\mathit{ole}}}\xspace}

\newcommand{\tHamQueryLite}{\ensuremath{t\mathsf{HamQueryLite}}\xspace}
\newcommand{\tHamQuery}{\ensuremath{t\mathsf{HamQuery}}\xspace}
\newcommand{\tHamQueryExp}{\ensuremath{t\mathsf{HamQueryExp}}\xspace}
\newcommand{\tHamQueryRestricted}{\ensuremath{t\mathsf{HamQueryRestricted}}\xspace}

\newcommand{\intAwarePSIProtocol}{\ensuremath{\mathsf{IntPSI}}\xspace}
\newcommand{\oneSidedSetRecon}{\ensuremath{\mathsf{OneSidedSetRecon}}\xspace}
\newcommand{\oneSidedSetReconExp}{\ensuremath{\mathsf{OneSidedSetReconExp}}\xspace}
\newcommand{\oneSidedSetReconBlind}{\ensuremath{\mathsf{OneSidedSetReconBlind}}\xspace}
\newcommand{\oneSidedSetReconMultiBlind}{\ensuremath{\mathsf{OneSidedSetReconMultiBlind}}\xspace}
\newcommand{\permAndPart}{\ensuremath{\mathsf{PermuteAndPartition}}\xspace}
\newcommand{\subSample}{\ensuremath{\mathsf{SubSample}}\xspace}
\newcommand{\hamCompute}{\ensuremath{\mathsf{HamCompute}}\xspace}
\newcommand{\hamAwarePSIProtocol}{\ensuremath{\mathsf{HamPSI}}\xspace}
\newcommand{\hamAwarePSIProtocolExp}{\ensuremath{\mathsf{HamPSISample}}\xspace}
\newcommand{\hamQuerySample}{\ensuremath{t\mathsf{HamQuerySample}}\xspace}
\newcommand{\HamContainQuery}{\ensuremath{\mathsf{HamContainQuery}}\xspace}

\newcommand{\indcpaAdversary}{\ensuremath{\mathcal{A}}\xspace}
\newcommand{\indcpaChallenger}{\ensuremath{\mathcal{C}}\xspace}

\newcommand{\indcpaBit}{\ensuremath{b}\xspace}
\newcommand{\indcpaAdversaryBit}{\ensuremath{b'}\xspace}

\newcommand{\indcpaChallengePolyZero}{\ensuremath{\polynomial{Q_0}}\xspace}

\newcommand{\indcpaChallengePolyOne}{\ensuremath{\polynomial{Q_1}}\xspace}

\newcommand{\indcpaChallengePolyBit}{\ensuremath{\polynomial{Q_b}}\xspace}
\newcommand{\indcpaChallengePolyEval}[1]{\ensuremath{Q_b(#1)}\xspace}
\newcommand{\indcpaFixedPoly}{\ensuremath{\polynomial{P_c}}\xspace}
\newcommand{\indcpaFixedPolyEval}[1]{\ensuremath{P_c(#1)}\xspace}

\newcommand{\indcpaOutputPolyEval}[1]{\ensuremath{\alicesRandomPolyEval{#1} \indcpaFixedPolyEval{#1} + \bobsRandomPolyEval{#1} \indcpaChallengePolyEval{#1}}}
\newcommand{\degreeOfGCD}{\ensuremath{d_{\mathsf{GCD}}}\xspace}

\newcommand{\gcdPoly}{\ensuremath{C(x)}\xspace}
\newcommand{\gcdPolyEval}[1]{\ensuremath{C(#1)}\xspace}
\newcommand{\indcpaChallengeSetZero}{\ensuremath{S_{Q0}}\xspace}
\newcommand{\indcpaChallengeSetOne}{\ensuremath{S_{Q1}}\xspace}
\newcommand{\indcpaChallengeSetBit}{\ensuremath{S_{Q\indcpaBit}}\xspace}
\newcommand{\indcpaFixedSet}{\ensuremath{S_{P}}\xspace}

\newcommand{\norm}[1]{\left\lVert#1\right\rVert}
\newcommand{\hamWeight}[1]{\ensuremath{\norm{#1}_{w}}\xspace}

\newcommand{\set}[1]{\ensuremath{#1}\xspace}
\newcommand{\setDefn}[1]{\ensuremath{\{x_\ptIdx\}_{\ptIdx = 1}^{#1}}}
\newcommand{\intSet}{\ensuremath{\set{S}}\xspace}
\newcommand{\dist}{\ensuremath{\delta}\xspace}

\newcommand{\elemInSet}[2]{\ensuremath{#1 \in \set{#2}}\xspace}

\newcommand{\setDiff}[2]{\ensuremath{#1 \setminus #2}\xspace}
\newcommand{\setDiffCard}[2]{\ensuremath{| #1 \setminus #2|}\xspace}

\newcommand{\setIntCard}[2]{\ensuremath{| #1 \cap #2|}\xspace}

\newcommand{\prfDomain}{\ensuremath{\genericPRF_i: \finiteField \rightarrow \finiteField}\xspace}
\newcommand{\mappingFn}{\ensuremath{M}}
\newcommand{\mapSet}{\ensuremath{\mappingFn_1, \ldots \mappingFn_{\genericVecLen}}\xspace}
\newcommand{\mapDomain}{\ensuremath{\mappingFn_\vecBitIdx: \{0,1\} \rightarrow \finiteField, ~\vecBitIdx \in [1, \genericVecLen]}\xspace}
\newcommand{\mapSetIdx}[1]{\ensuremath{\mappingFn_{#1}}\xspace}
\newcommand{\comp}[1]{\ensuremath{1 - #1}\xspace}
\newcommand{\alicesSet}{\ensuremath{\set{A}}\xspace}
\newcommand{\bobsSet}{\ensuremath{\set{B}}\xspace}

\newcommand{\tPSI}{\ensuremath{\lowercase{t}\mathsf{PSI}}\xspace}
\newcommand{\tPSICard}{\ensuremath{\lowercase{t}\mathsf{PSI}{\text -}\mathsf{CA}}\xspace}
\newcommand{\tPSIRec}{\ensuremath{\lowercase{t}\mathsf{PSI}{\text -}\mathsf{SR}}\xspace}

\newcommand{\tPSIOutputPoly}{\ensuremath{\alicesRandomPoly \alicesPoly + \bobsRandomPoly \bobsPoly}\xspace}
\newcommand{\alicesReducedPoly}{\ensuremath{\polyNotGCDRoots{P}{q}}\xspace}
\newcommand{\bobsReducedPoly}{\ensuremath{\polyNotGCDRoots{Q}{p}}\xspace}
\newcommand{\tPSIReducedOutputPoly}{\ensuremath{\alicesRandomPoly \alicesReducedPoly + \bobsRandomPoly \bobsReducedPoly}\xspace}
\newcommand{\tPSIOutputPolyEval}[1]{\ensuremath{\alicesRandomPolyEval{#1} \alicesPolyFromVecEval{#1} + \bobsRandomPolyEval{#1} \bobsPolyFromVecEval{#1}}\xspace}
\newcommand{\alicesPoly}{\ensuremath{\polynomial{P}}\xspace}
\newcommand{\bobsPoly}{\ensuremath{\polynomial{Q}}\xspace}
\newcommand{\rationalfn}[2]{\ensuremath{\frac{#1}{#2}}\xspace}

\newcommand{\alicesSetReconInput}{\ensuremath{S_A}\xspace}

\newcommand{\setOfBobsSetReconInputs}{\ensuremath{\{S_{B1}, \ldots, S_{B\genericSetSize}\}}\xspace}

\newcommand{\setOfAlicesRandomPolys}{\ensuremath{\{R_{11}(x), \ldots, R_{1\genericSetSize}(x)\}}\xspace}
\newcommand{\alicesRandomPolysIdx}[1]{\ensuremath{R_{1#1}(x)}\xspace}
\newcommand{\alicesRandomPolysIdxEval}[2]{\ensuremath{R_{1#1}(#2)}\xspace}
\newcommand{\setOfBobsRandomPolys}{\ensuremath{\{R_{21}(x), \ldots, R_{2\genericSetSize}(x)\}}\xspace}
\newcommand{\bobsRandomPolysIdx}[1]{\ensuremath{R_{2#1}(x)}\xspace}
\newcommand{\bobsRandomPolysIdxEval}[2]{\ensuremath{R_{2#1}(#2)}\xspace}

\newcommand{\setOfAlicesVoleOutputPolys}{\ensuremath{\{W_{A1}( x_{\ptIdx}) ,\ldots, W_{An}( x_{\ptIdx})\}}\xspace}

\newcommand{\alicesVoleOutputPolysIdxEval}[2]{\ensuremath{W_{A#1}(#2)}\xspace}

\newcommand{\blindingPolySym}{\ensuremath{U}}
\newcommand{\blindingPoly}{\ensuremath{\blindingPolySym(x)}\xspace}
\newcommand{\blindingPolyEval}[1]{\ensuremath{U(#1)}\xspace}
\newcommand{\alicesRandomPolySym}{\ensuremath{R_1}}
\newcommand{\alicesRandomPoly}{\ensuremath{\alicesRandomPolySym(x)}\xspace} 	
\newcommand{\bobsRandomPolySym}{\ensuremath{R_2}\xspace}		
\newcommand{\numeratorPoly}{\ensuremath{\mathsf{Num}(x)}\xspace}
\newcommand{\denominatorPoly}{\ensuremath{\mathsf{Den}(x)}\xspace}

\newcommand{\bobsRandomPoly}{\ensuremath{\bobsRandomPolySym(x)}\xspace}		
\newcommand{\alicesRandomPolyEval}[1]{\ensuremath{R_1(#1)}\xspace}
\newcommand{\bobsRandomPolyEval}[1]{\ensuremath{R_2(#1)}\xspace}

\newcommand{\alicesVoleOutputPolySym}{\ensuremath{W_A}}
\newcommand{\alicesVoleOutputPoly}{\ensuremath{\alicesVoleOutputPolySym(x)}\xspace}
\newcommand{\alicesVoleOutputPolyEval}[1]{\ensuremath{W_A(#1)}\xspace}

\newcommand{\bobsVoleOutputPolySym}{\ensuremath{W_B}\xspace} 
\newcommand{\bobsVoleOutputPoly}{\ensuremath{\bobsVoleOutputPolySym(x)}\xspace}
\newcommand{\bobsVoleOutputPolyEval}[1]{\ensuremath{W_B(#1)}\xspace}
\newcommand{\setOfPtsForInterpolation}{\ensuremath{V}\xspace}

\newcommand{\ptIdx}{\ensuremath{k}\xspace}
\newcommand{\setIdx}{\ensuremath{i}\xspace}
\newcommand{\genericPtX}{\ensuremath{x_\ptIdx}\xspace}
\newcommand{\genericPtY}{\ensuremath{y_\ptIdx}\xspace}
\newcommand{\genericPt}{\ensuremath{(\genericPtX, \genericPtY)}\xspace}

\newcommand{\vecBitIdx}{\ensuremath{m}\xspace}

\newcommand{\genericVecLen}{\ensuremath{\mathcal{\ell}}\xspace}
\newcommand{\hamDist}[2]{\ensuremath{\dist_{H}(#1, #2)}}
\newcommand{\prfKey}{\ensuremath{\kappa}\xspace}
\newcommand{\prfKeySet}{\ensuremath{\{\kappa_1, \ldots, \kappa_{\genericSetSize}\}}\xspace}
\newcommand{\prfKeyIdx}[1]{\ensuremath{\kappa_{#1}}\xspace}
\newcommand{\finiteFieldPRF}{\ensuremath{\phi}\xspace}
\newcommand{\allPossibleComb}[2]{\ensuremath{\binom{#1}{#2}}\xspace}
\newcommand{\hamDistThreshold}{\ensuremath{d_\mathsf{H}}\xspace}

\newcommand{\numPointsForInterpol}{\ensuremath{\genericVecLen + 2\hamDistThreshold + 1}\xspace}
\newcommand{\numPointsForInterpolPlus}{\ensuremath{\genericVecLen + \hamDistThreshold + 2}\xspace}
\newcommand{\numPointsForInterpolAfterBins}{\ensuremath{\numOfBinsRef + 2\hamDistThreshold + 1}\xspace}
\newcommand{\polyInField}[1]{\ensuremath{#1(x) \in \finiteField[x]}\xspace}
\newcommand{\polyField}{\ensuremath{\finiteField[x]}\xspace}

\newcommand{\degree}[1]{\ensuremath{\mathbf{deg}\mathopen{}\left(#1 \right)\mathclose{}}\xspace}		
\renewcommand{\gcd}[2]{\ensuremath{\mathbf{gcd}\mathopen{}\left(#1, #2 \right)\mathclose{}}\xspace}		

\newcommand{\polyFromSet}[1]{\ensuremath{\prod\limits_{r \in \set{#1}} (x - r)}\xspace}
\newcommand{\setOfXCoords}{\ensuremath{\set{X}}\xspace}
\newcommand{\setOfXCoordsDefn}{\ensuremath{\{x_\ptIdx\}_{\ptIdx = 1}^{\numPointsForInterpol}}\xspace}
\newcommand{\setOfXCoordsDefnPlus}{\ensuremath{\{x_\ptIdx\}_{\ptIdx = 1}^{\numPointsForInterpolPlus}}\xspace}
\newcommand{\setOfXCoordsDefnBins}{\ensuremath{\{x_\ptIdx\}_{\ptIdx = 1}^{\numPointsForInterpolAfterBins}}\xspace}
\newcommand{\polyNotGCDRoots}[2]{\ensuremath{#1_{{\scriptscriptstyle / #2}}(x)}\xspace}
\newcommand{\hamDistLessThanThreshold}[2]{\ensuremath{\hamDist{#1}{#2} \leq \hamDistThreshold}\xspace}
\newcommand{\hamDistGreaterThanThreshold}[2]{\ensuremath{\hamDist{#1}{#2} > \hamDistThreshold}\xspace}

\newcommand{\alicesSetFromVec}{\ensuremath{S_{\alicesVector}}\xspace}
\newcommand{\alicesSetFromVecDefn}{\ensuremath{ \alicesSetFromVec = \{\mappingFn_1(\alicesVector[1]), \ldots, 
\mappingFn_\genericVecLen(\alicesVector[\genericVecLen])\}}\xspace}
\newcommand{\bobsSetFromVec}{\ensuremath{S_{\bobsVector}}\xspace}
\newcommand{\bobsSetFromVecDefn}{\ensuremath{\bobsSetFromVec = \{\mappingFn_1(\bobsVector[1]), \ldots, 
	\mappingFn_\genericVecLen(\bobsVector[\genericVecLen])\}}\xspace}
\newcommand{\alicesVector}{\ensuremath{\vec{a}}\xspace}
\newcommand{\bobsVector}{\ensuremath{\vec{b}}\xspace}
\newcommand{\alicesVecIdx}[1]{\ensuremath{\vec{a}_{#1}}\xspace}
\newcommand{\bobsVecIdx}[1]{\ensuremath{\vec{b}_{#1}}\xspace}
\newcommand{\alicesPolyFromVecSym}{\ensuremath{P}\xspace}
\newcommand{\alicesPolyFromVec}{\ensuremath{\alicesPolyFromVecSym(x)}\xspace}
\newcommand{\alicesPolyFromVecEval}[1]{\ensuremath{P(#1)}\xspace}
\newcommand{\bobsPolyFromVecSym}{\ensuremath{Q}\xspace}
\newcommand{\bobsPolyFromVec}{\ensuremath{\bobsPolyFromVecSym(x)}\xspace}
\newcommand{\bobsPolyFromVecEval}[1]{\ensuremath{Q(#1)}\xspace}

\newcommand{\bobsPolyFromVecIdxEval}[2]{\ensuremath{Q_{#1}(#2)}\xspace}
\newcommand{\alicesSetFromVecBins}{\ensuremath{S'_{\alicesVector}}\xspace}
\newcommand{\bobsSetFromVecBins}{\ensuremath{S'_{\bobsVector}}\xspace}
\newcommand{\bobsSetFromVecBinsIdx}[1]{\ensuremath{S'_{\bobsVector_{#1}}}\xspace}

\newcommand{\keySet}{\ensuremath{\mathsf{KeySet}}\xspace}
\newcommand{\randPerm}{\ensuremath{\pi}\xspace}

\newcommand{\numOfBins}{\ensuremath{\frac{2\hamDistThreshold^2}{\fpr}}\xspace}
\newcommand{\numOfBinsRef}{\ensuremath{N_\mathsf{bins}}\xspace}
\newcommand{\parity}{\ensuremath{\mathbf{parity}}\xspace}
\newcommand{\alicesSetOfSubVecs}{\ensuremath{\vec{x_1}, \ldots, \vec{x}_{\numOfBinsRef}}\xspace}
\newcommand{\alicesSubVecIdx}[1]{\ensuremath{\vec{x}_{#1}}\xspace}
\newcommand{\bobsSubVecIdx}[1]{\ensuremath{\vec{y}_{#1}}\xspace}
\newcommand{\bobsSetOfSubVecs}{\ensuremath{\vec{y_1}, \ldots, \vec{y}_{\numOfBinsRef}}\xspace}

\newcommand{\alicesParityVec}{\ensuremath{\vec{X}}\xspace}
\newcommand{\bobsParityVec}{\ensuremath{\vec{Y}}\xspace}
\newcommand{\alicesPermutedVec}{\ensuremath{\alicesVector_\mathsf{perm}}\xspace}
\newcommand{\bobsPermutedVec}{\ensuremath{\bobsVector_\mathsf{perm}}\xspace}

\newcommand{\bobsPolySet}{\ensuremath{\{Q_1(x), \ldots, Q_\genericSetSize(x)\}\xspace}}

\newcommand{\bobsPolyIdx}[1]{\ensuremath{Q_{#1}(x)}\xspace}

\newcommand{\alicesInputVecSet}{\ensuremath{A}\xspace}
\newcommand{\alicesInputVecSetDefn}{\ensuremath{\alicesInputVecSet \assign \{\vec{a}_1, \ldots, \vec{a}_\genericSetSize\}}\xspace}
\newcommand{\bobsInputVecSet}{\ensuremath{B}\xspace}
\newcommand{\bobsInputVecSetDefn}{\ensuremath{\bobsInputVecSet \assign \{\vec{b}_1, \ldots, \vec{b}_\genericSetSize\}}\xspace}
\newcommand{\alicesVoleOutputPolyIdxEval}[2]{\ensuremath{W_{A#1}(#2)}\xspace}

\newcommand{\intDistThreshold}{\ensuremath{d_{\mathsf{int}}}\xspace}
\newcommand{\integerFromHigherOrderString}[1]{\ensuremath{\mathsf{High}(#1)\xspace}}
\newcommand{\integerFromLowerOrderString}[1]{\ensuremath{\mathsf{Low}(#1)\xspace}}

\newcommand{\numOfMaximalSubtriesinTree}[1]{\ensuremath{\phi({#1})}\xspace}
\newcommand{\trie}[2]{\ensuremath{\mathsf{T}({#1}, {#2})}\xspace}
\newcommand{\flooredLog}[1]{\ensuremath{\floor{\log #1}}\xspace}

\newcommand{\poweredfloorlog}[1]{\ensuremath{2^{\floor{\log #1}}\xspace}}
\newcommand{\mec}{{\textit{maximal enclosing complete}}\xspace}

\newcommand{\maxBitLength}{\ensuremath{\mathsf{MaxBitLen}}\xspace}
\newcommand{\Desc}[2]{\State \makebox[3em][l]{#1}#2}
\newcommand{\universe}{\ensuremath{\mathcal{U}}\xspace}

\newcommand{\integersInNeighborhood}{\ensuremath{a', |a' - a| \leq \intDistThreshold}\xspace}
\newcommand{\intDistanceLessThanThreshold}[2]{\ensuremath{|#1 - #2 | < \intDistThreshold}\xspace}
\newcommand{\alicesInt}{\ensuremath{a}\xspace}
\newcommand{\bobsInt}{\ensuremath{b}\xspace}
\newcommand{\alicesIthInt}{\ensuremath{a_i}\xspace}
\newcommand{\bobsJthInt}{\ensuremath{b_j}\xspace}
\newcommand{\augAlicesSet}{\ensuremath{\set{\hat{\alicesSet}}}\xspace}
\newcommand{\augBobsSet}{\ensuremath{\set{\hat{\bobsSet}}}\xspace}
\newcommand{\genericPrefLen}{\ensuremath{\ell_p}\xspace}
\newcommand{\integersInRange}{\ensuremath{(a - \intDistThreshold, a + \intDistThreshold)}\xspace}

\newcommand{\bigO}[1]{\ensuremath{\operatorname{O}\mathopen{}\left(#1\right)\mathclose{}}}
\newcommand{\bigTheta}[1]{\ensuremath{\operatorname{\Theta}\mathopen{}\left(#1\right)\mathclose{}}}
\newcommand{\neglfn}[1]{\ensuremath{\mathbf{negl}\mathopen{}\left(#1\right)\mathclose{}}\xspace}
\newcommand{\polyfn}{\ensuremath{\mathbf{poly}}\xspace}

\newcommand{\bigOtilde}[1]{\ensuremath{\operatorname{\widetilde{O}}\mathopen{}\left(#1\right)\mathclose{}}}
\newcommand{\hamPSICommComplexity}{\ensuremath{\bigO{\genericSetSize^2 \cdot \frac{\hamDistThreshold^2}{\fpr} \cdot \secParam}}\xspace}

\newcommand{\intAwarePSICommComplexity}{\ensuremath{\bigO{\genericSetSize \secParam \log \intDistThreshold}}\xspace}

\newcommand{\simm}{\ensuremath{\mathsf{Sim}}\xspace}

\newcommand{\alicesVoleOutputPolySim}{\ensuremath{\alicesVoleOutputPolySym^{\mathsf{Sim}}(x)}\xspace}
\newcommand{\alicesVoleOutputPolySimEval}[1]{\ensuremath{\alicesVoleOutputPolySym^{\mathsf{Sim}}(#1)}\xspace}
\newcommand{\bobsPolyFromVecSim}{\ensuremath{\bobsPolyFromVecSym^{\mathsf{Sim}}(x)}\xspace}
\newcommand{\bobsPolyFromVecSimEval}[1]{\ensuremath{\bobsPolyFromVecSym^{\mathsf{Sim}}(#1)}\xspace}
\newcommand{\bobsVoleOutputPolySim}{\ensuremath{\bobsVoleOutputPolySym^{\mathsf{Sim}}(x)}\xspace}
\newcommand{\bobsVoleOutputPolySimEval}[1]{\ensuremath{\bobsVoleOutputPolySym^{\mathsf{Sim}}(#1)}\xspace}
\newcommand{\tPSIOutputPolySimEval}[1]{\ensuremath{\alicesVoleOutputPolySimEval{#1} \assign \alicesRandomPolyEval{#1} \alicesPolyFromVecEval{#1} + \bobsRandomPolyEval{#1} \bobsPolyFromVecSimEval{#1}}}

\newcommand{\alicesVoleOutputPolyPtSetReal}{\ensuremath{V_{A}}\xspace}
\newcommand{\alicesVoleOutputPolyPtSetSim}{\ensuremath{V^{\mathsf{Sim}}_{A}}\xspace}
\newcommand{\bobsVoleOutputPolyPtSetReal}{\ensuremath{V_{B}}\xspace}
\newcommand{\bobsVoleOutputPolyPtSetSim}{\ensuremath{V^{\mathsf{Sim}}_{B}}\xspace}
\newcommand{\randomPtSetA}{\ensuremath{X_A}\xspace}
\newcommand{\randomPtSetB}{\ensuremath{P_B}\xspace}
\newcommand{\keySetSim}{\ensuremath{\keySet_{\mathsf{Sim}}}\xspace}

\begin{document}
\title{Distance-Aware Private Set Intersection} 

\author{
	{\rm Anrin Chakraborti}\\
	Duke University
	\and
	{\rm Giulia Fanti}\\
	Carnegie Mellon University
	\and
	{\rm Michael K.\ Reiter}\\
	Duke University
}
\maketitle

\maketitle

\begin{abstract}
  Private set intersection (PSI) allows two mutually untrusting parties to compute an intersection of their sets, without revealing information about items that are not in the intersection. This work introduces a PSI variant called {\em distance-aware} PSI (DA-PSI) for sets whose elements lie in a metric space. DA-PSI returns pairs of items that are within a specified distance threshold of each other. This paper puts forward DA-PSI constructions for two metric spaces: (i) Minkowski distance of order 1 over the set of integers (i.e., for integers $a$ and $b$, their distance is $|a-b|$); and (ii) Hamming distance  over the set of binary strings of length $\ell$.  In the Minkowski DA-PSI protocol, the communication complexity scales logarithmically in the distance threshold and linearly in the set size.  In the Hamming DA-PSI protocol, the communication volume scales quadratically in the distance threshold and is independent of the dimensionality of string length $\ell$.  Experimental results with real applications confirm that DA-PSI provides more effective matching at lower cost than \naive solutions.
\end{abstract}

\section{Introduction}

Private set intersection (PSI) is a widely-used multiparty cryptographic protocol, with  
applications across domains including contact discovery and tracing, private profile matching, privacy-preserving genomics, and collaborative learning. PSI protocols are used to compute the intersection (or common elements) of two or more sets held by mutually-untrusting parties. Critically, the parties learn no information about the elements that are \emph{not} in the set intersection. 

There is a long line of work on communication-efficient PSI protocols, with variants including different adversarial models~\cite{kacsmar2020dpPSI,groce2019CheaperPS,pinkas2020PSIFP}, threshold parameters~\cite{ghosh2019thresholdPSI,badrinarayanan2021thresholdPSI} and compute capabilities~\cite{ion2019GoogleJoin}. 
However, these solutions are designed to return only \emph{exact} matches. 
That is, an element appears in the intersection if and only if it matches (exactly) an element in each of the other parties' sets.

When exact matches may be rare, parties may want to privately compute \emph{approximate} matches. For instance, given sets of points in Euclidean space, the intersection may contain all pairs within a certain Euclidean distance of each other. This notion has applications in domains where replacing an exact-matched set intersection with a distance-based intersection yields more effective systems:

\begin{itemize}[nosep,leftmargin=1em,labelwidth=*,align=left]
	\item \textbf{Private collaborative blacklisting} enables mutually-untrusting parties to identify malicious network traffic and coordinated attacks. Typically, an intersection is computed privately over sets containing network identifiers, e.g., source IPs observed~\cite{melis2019collabblacklist}. However, botnets usually span multiple (often contiguous) subnets~\cite{west2010spam,collins2007uncleanliness}, 	and it is useful to compare ranges of IP addresses and detect overlaps.
	
	\item \textbf{Biometric identification} systems 
	leverage Hamming distance/edit distance \cite{wang2015editDist,kambs2017genomehamming, daugman2009iris} and need to account for inexact/fuzzy matches due variations in sampling technologies. Fuzzy PSI functionalities have been used before in privacy-preserving biometric identification systems \cite{uzun2021fuzzy}.
		
	\item \textbf{Credential stuffing identification} systems use PSI-like functionalities to detect password reuse across websites without revealing sensitive user information~\cite{wang2019reuse,wang2020stuffing}. These protocols consider exact password matches. However, it is useful to expand this idea for inexact/similar password matches, e.g., edit distance matches.

\end{itemize}

\myparagraph{Distance-Aware Private Set Intersection (DA-PSI)}
In this paper we initiate the study of \textit{distance-aware} private set intersection (DA-PSI). A DA-PSI protocol defined over a metric space allows two parties to compute an intersection of their respective sets containing all pairs of items that are within a predefined threshold distance in the metric space. Specifically, 
consider a metric space $(\universe, \delta)$ with metric $\delta$. Let the parties hold sets, \alicesSet and \bobsSet with \genericSetSize items each, where each item is a length-\genericVecLen vector drawn from the space.  The problem definition specifies a distance threshold $d$ and requires the protocol to return $S \subseteq \alicesSet \times \bobsSet$ where $(\vec{a}, \vec{b}) \in S \Leftrightarrow \delta(\vec{a}, \vec{b}) \leq d$. A party learns no information about items in the counterparty's set that are not close to (within threshold of) one of its own elements.


Traditional PSI tools are not optimized for DA-PSI. A \naive application will need to check for all $\vec{a} \in \alicesSet$ and for all $\vec{b} \in \universe$ such that $\delta(\vec{a}, \vec{b}) \leq d$ if $\vec{b} \in \bobsSet$. This is problematic since in many cases the 
search space is exponentially large. For example, Hamming distance with a distance threshold $d$ will require searching over the Hamming ball of radius $d$ around each $\vec{a} \in \alicesSet$. There are over $\binom{\genericVecLen}{d} = \Omega((\frac{\genericVecLen}{d})^{d})$ vectors around $\vec{a}$ in this Hamming ball. Thus, the communication cost of this protocol scales exponentially with the threshold and is impractical. This work poses the following question: \textit{Can we design DA-PSI protocols where communication and compute costs scale polynomially in the distance threshold?}



We answer this question affirmatively by putting forward constructions for two important metric spaces: i) Hamming distance  over the set $\{0,1\}^\ell$ for some fixed $\ell \in \mathbb Z^{+}$, and ii)
Minkowski distance of order 1 over the set of integers (i.e., for integers $a$ and $b$, their distance is $|a-b|$). 
In the following, we discuss the intuitions behind these protocols.

\myparagraph{Hamming Distance-Aware PSI}
Hamming distance is a good starting point since several other distances can be computed or approximated by Hamming distance \cite{raginsky2009kernelLSH,datar2004LSH,aristides1999LSH}. We provide a construction building on the idea of sub-sampling each of a user's input vectors and mapping it to a unique set of sub-vectors such that the cardinality of the set difference between the sets corresponding to two vectors is exactly equal to the Hamming distance between the original vectors.


Our construction leverages additively homomorphic encryption and vector oblivious linear evaluation (VOLE). A key building block in the protocol is a novel sub-sampling mechanism which trades off accuracy (by allowing some false-positives) for better communication complexity: the communication cost scales polynomially in the distance threshold $d$ and is \textit{independent of the vector length \genericVecLen}. Typically, for applications relying on Hamming distance comparisons, $d \ll \genericVecLen$ \cite{osadchy2010scifi,manku2007webcrawl}, and thus the reduced set sizes after sub-sampling concretely improves communication costs over existing work \cite{huang2006hamming, osadchy2010scifi} . To compute the set differences, we propose a modified (and significantly simpler) version of the private set reconciliation protocol due to \citet{ghosh2019thresholdPSI}. 



\myparagraph{Integer distance-Aware PSI}
We propose a DA-PSI protocol for Minkowski distance of first order over integers, loosely termed as the integer distance-aware PSI protocol. 
The communication cost scales 
linearly in the set size and logarithmically in the distance threshold; this is optimal with regards to the set sizes since linear communication is both necessary and sufficient for exact PSI \cite{kalyanasundaram1992probabilistic,ghosh2019thresholdPSI}. The key observation behind the protocol is that integers in a range $(a - d, a + d)$, where $d$ is a specified distance threshold,  can be succinctly represented by a collection of bit-strings corresponding to their binary representations. The total number of strings required is sublinear in $d$ since multiple integers within a sequence will share prefixes, and the same common prefix will represent multiple consecutive integers. We design an algorithm to augment the inputs sets with $\bigO{\log d}$ strings representing all integers in $(a-d, a+d)$. This mechanism is agnostic to the underlying cryptographic tools since any state-of-the-art PSI protocol can be augmented to provide an integer DA-PSI protocol. 




\myparagraph{Evaluation}
We have implemented both protocols and benchmarked them on a public cloud. As an application of the integer distance-aware PSI, we have deployed it for collaborative blacklisting of IPs seen by real-world honeypots; for our parameter settings, the distance-aware PSI almost doubles the number of identified malicious IPs. For computing this intersection over sets containing roughly 25K IP addresses collected across all the honeypots, our protocol only requires 64 MB of communication and 1.5 seconds. 

We have implemented our Hamming DA-PSI constructions. Micro-benchmarks show that it imposes $2$-$400\times$ less communication than
a generic garbled-circuit solution for distance thresholds up to $30$. As an application, we have evaluated our protocol for the task of privately comparing vectors derived from iris images, and for a distance threshold sufficient to retrieve all of the matches in our dataset, it achieves $2.5\times$ lower communication volume (with a false positive rate $\le 10\%$ and no false negatives) than a generic secure 2PC baseline. When compared with the state-of-the-art Hamming containment query protocol by \citet{uzun2021fuzzy}, our protocol features $33$-$63\%$ less communication and $33\%$ less computation. 

\section{Related Work}
\label{sec:related}
\begin{table}[]
\centering
{\footnotesize
\begin{tabular}{@{}l@{\hspace{0.55em}}c@{\hspace{0.25em}}c@{\hspace{0.25em}}c@{\hspace{0.25em}}c@{\hspace{0.5em}}c@{\hspace{0.5em}}c@{}}
\toprule
\multicolumn{7}{c}{\textbf{Hamming DA-PSI}} \\
& \multirow{2}{*}{Comm} & \multicolumn{3}{c}{Computation} & \multirow{2}{*}{Dep} & \multirow{2}{*}{\fpr, \fnr} \\
& & \alice & \bob & Offline & & \\ 
\cmidrule{3-5}
\naive & \bigO{\genericSetSize}  & \bigO{\genericSetSize \binom{\genericVecLen}{d}} & \bigO{\genericSetSize} & \bigO{1}  & -- & -- , --  \\

\parbox[c]{4em}{\citet{osadchy2010scifi}} & \bigO{\genericSetSize^2 \genericVecLen \secParam}  & 
\bigO{\genericSetSize^2} &  \bigO{\genericSetSize^2} & \bigO{1}  &  \textsc{ahe} & -- , --         \\[8pt]

\parbox[c]{4em}{\citet{huang2011garbled}} & \bigO{\genericSetSize^2 \genericVecLen \secParam} & 
\bigO{\genericSetSize^2} & \bigO{\genericSetSize^2} & \bigO{1}  &  \textsc{ot} & -- , --         \\[8pt]

\multirow{2}{*}{\parbox[c]{4em}{\citet{uzun2021fuzzy}}}      & \multirow{2}{*}{\bigO{\frac{\genericSetSize^2 T}{mB} \lambda}} & \multirow{2}{*}{\bigO{\!\binom{T}{t}\!\left(\frac{\genericSetSize ma}{T}\right)\!}} &  \multirow{2}{*}{\bigO{\frac{\genericSetSize^2 T}{m}}}  &  \multirow{2}{*}{\bigO{\frac{\genericSetSize^3 T^2}{ma}}}  & \multirow{2}{*}{\textsc{fhe}} & $0\!\!<\!\!\fpr\!\!<\!\!1$           \\
&  & & & &  & $0\!\!<\!\!\fnr\!\!<\!\!1$ \\[8pt]

\multirow{2}{*}{\parbox[c]{4em}{\hamAwarePSIProtocol (\secref{sec:hammingawarePSI})}} & \multirow{2}{*}{\bigO{\genericSetSize^2 d^2 \secParam}}  & \multirow{2}{*}{\bigO{\genericSetSize^2}} &  \multirow{2}{*}{\bigO{\genericSetSize^2}} & \multirow{2}{*}{\bigO{1}}  & \textsc{ahe} & $0\!\!<\!\!\fpr\!\!<\!\!1$         \\
&  & & & &  \textsc{ole} & -- \\[8pt]



\multicolumn{7}{c}{\textbf{Integer DA-PSI}} \\

\naive & \bigO{\genericSetSize \secParam d}  & \bigO{\genericSetSize} &  \bigO{\genericSetSize} & \bigO{1}  &  -- & -- , --         \\[8pt]

\parbox[c]{4em}{\intAwarePSIProtocol (\secref{sec:intergerAwarePSI})} & \bigO{\genericSetSize\secParam \log{d}}  & \bigO{\genericSetSize} &  \bigO{\genericSetSize} & \bigO{1}  &  -- & -- , --         \\ \bottomrule

\end{tabular}}
\vspace{-10pt}
\caption{\small Asymptotic performance for our protocols in comparison to existing work. 
\genericSetSize:  set size; $d$: distance threshold; \genericVecLen: length of vectors; $T, t$: subsampling parameters, $T = \bigO{\genericVecLen}$; $m,a, B$: FHE parameters, $m \gg T$. $\fpr, \fnr \in (0,1)$ are the false-positive and false-negative rates of the schemes; "--" indicates that the false positive or negative rate is negligible in \secParam. The schemes without dependencies are agnostic to the underlying primitives. The comm. cost of all existing Hamming DA-PSI protocol depends on the length of the  vectors \genericVecLen. The comm. cost of \hamAwarePSIProtocol is independent of \genericVecLen.  
\label{tab:comparison}}
\end{table}

Private set intersection is well-studied (e.g.,~\cite{freedman2004PSI,Kissner2005setops, chase2020PSIOPRF,pinkas2015Hashing,pinkas2018PSIOT,meadows86PSIDH,de2010:mal-psi,pinkas2019spotLight}). 
We refer to these for details on general PSI and focus on other distance-aware/fuzzy PSI primitives here.

\myparagraph{Private Hamming Distance Computation}
\tblref{tab:comparison} compares our Hamming DA-PSI constructions with existing work on privately computing Hamming distance. 
\citet{osadchy2010scifi} built a protocol using additively homomorphic encryption which enables a party to check when her vector is within a threshold Hamming distance of any of the vectors in a set held by the other party. A similar functionality is implemented by \citet{huang2006hamming} using garbled circuits. For both protocols, the communication cost scales linearly in the vector sizes. In contrast, the cost of our Hamming distance aware protocol scales sublinearly in the vector size. 

\citet{uzun2021fuzzy} propose a protocol for Hamming distance comparisons over vectors derived from biometric identifiers. The protocol reduces the input vectors to sets of sub-vectors after a sub-sampling process. The sub-sampling protocol ensures that when two vectors are close, their corresponding sets have a certain number of matching elements. With these sets as inputs, the protocol implements a $t$-out-of-$T$ matching protocol using fully homomorphic encryption (FHE), and leverages the ability of the FHE scheme to pack multiple ciphertexts using SIMD-style operations. In contrast,  our constructions are based on computationally less-expensive primitives, namely additively homomorphic encryption (AHE) and oblivious linear evaluations (OLE). 

\section{Security Definitions \& Background}
\label{sec:definitions}

\myparagraph{Notation}
\finiteField is a finite of field of prime order where $\fieldOrder$ is a \bigO{\polyfn(\secParam)}-bit prime, and \secParam is a security parameter. \neglfn{\cdot} is a function that is negligible in the input parameter; e.g., \neglfn{\secParam} = \bigO{2^{-\secParam}}. \polyInField{P} is a polynomial with coefficients drawn from \finiteField. The degree of polynomial \polynomial{P} is represented by  \degree{P(x)}. The greatest common divisor of two (or more) polynomials is represented by \gcd{P(x)}{Q(x)}. For polynomials $P(x)$ and $Q(x)$, $\polyNotGCDRoots{P}{q} \assign \frac{P(x)}{\gcd{P(x)}{Q(x)}}$.  

\myparagraph{Rational Function}
A rational function $r = \rationalfn{P(x)}{Q(x)}$ has degree at most equal to the degree of the numerator + degree of the denominator. Let $V = \{\genericPt\}_{\ptIdx = 1}^{2t}$ be a set of points in \finiteField. Then there exists a rational function with numerator and denominator in \polyField interpolating these points \cite{larkin1967some}. 


\myparagraph{Parties}
We assume that two \textit{semi-honest} (a.k.a.\ honest-but-curious) mutually untrusting parties \alice and \bob run the protocols. The parties may learn information from the intermediate results but do not deviate from the protocol. 

\myparagraph{Distance-Aware Private Set Intersection}
In this work, we are concerned with \textit{distance-aware private set intersection} protocols, which we define here. 
Parties \alice and \bob are assumed to each store a set $\alicesSet=\{a_1,\ldots,a_n\}$ and $\bobsSet=\{b_1,\ldots, b_n\}$, respectively, where the $a_i$'s and $b_i$'s are drawn from some universe $\universe$, and $\dist: \universe \times \universe \to \mathbb R$ denotes a distance metric defined over $\universe$.
A distance-aware PSI protocol over metric space  $(\universe,\dist)$ with threshold $d$ and input sets \alicesSet, \bobsSet returns a set $S \subseteq A \times B$ such that 
$S \triangleq \{(a,b) ~:~ a\in A,~b\in B,~\dist(a,b) \leq d\}$. We require this protocol to satisfy:

\begin{enumerate}[nosep,leftmargin=1.6em,labelwidth=*,align=left]

\item {\bf Correctness:} For any $(a, b) \in \alicesSet \times
  \bobsSet$,
  \begin{itemize}[nosep,leftmargin=1em,labelwidth=*,align=left]
    \item If $\dist(a, b) \leq d$, then $(a, b) \in S$
      with probability $\ge \tpr$.
    \item If $\dist(a,b) > d$, then $(a,b)
      \not\in S$ with probability $\ge \tnr$.
  \end{itemize}
\item {\bf Security:} \alice learns only $S$ and the cardinality of
  \bobsSet, and \bob learns only $S$ and the cardinality of \alicesSet.
\end{enumerate}

This definition allows for arbitrary false-positive and false-negative rates (i.e., $\fpr = 1-\tnr$ and $\fnr = 1-\tpr$, respectively).  This allows faster 
protocols (see \secref{sec:hammingawarePSI}) and accommodates protocols approximating one distance metric via another, e.g., with locality-sensitive hashing. 



%

%
%

\subsection{Background}
\label{sec:background}

\begin{figure}
	\footnotesize

	\begin{mdframed}
        \underline{$\oleFunc$: Ideal Functionality for Oblivious Linear Evaluation (OLE):}

		\smallskip\noindent
		\textbf{Parameters:~} Parties \alice and \bob, and finite field \finiteField from which inputs are drawn.  
		
		\smallskip\noindent
		\textbf{Inputs:~} \alice has input $x \in \mathbb{F} $ and \bob has as input a pair $(u, v) \in \finiteField$.

		\smallskip\noindent
		\textbf{Output:~} \alice learns $z = u x + v$. \bob learns $\bot$.

    \medskip \underline{$\voleFunc$: Ideal Functionality for Vector OLE (VOLE):}
    
    	\smallskip\noindent
		\textbf{Parameters:~} Parties \alice and \bob, and finite field \finiteField from which inputs are drawn.  
		
		\smallskip
		\noindent
		\textbf{Inputs:~} \alice has input $x \in \mathbb{F} $ and \bob has as input a pair of vectors $(\vec{u}, \vec{v}) \in \finiteFieldVec{\genericVecLen} \times \finiteFieldVec{\genericVecLen}$.

		\smallskip
		\noindent
		\textbf{Output:~} \alice learns $\vec{z} = \vec{u} x + \vec{v}$. \bob learns $\bot$.
\end{mdframed}
\vspace{-15pt}
\caption{\small Ideal functionalities for oblivious linear evaluation (OLE), and vector oblivious linear evaluation (VOLE) \label{fig:ole_func}}
\end{figure}

\noindent\textbf{Oblivious Linear Evaluation (OLE)} is a two-party cryptographic primitive wherein \alice inputs $x \in \finiteField$; \bob inputs $u, v \in \finiteField$; and \alice obtains $ux + v$ \textit{without learning $u$ and $v$}.

\noindent\textbf{Vector Oblivious Linear Evaluation (VOLE)} is an extension of the OLE functionality, where \bob's input is a pair of vectors, and \alice learns a linear combination of the vectors. \figref{fig:ole_func} describes the VOLE functionality. The state-of-the-art VOLE protocol \cite{weng2021vole} is based on the learning parity with noise (LPN) assumption. The communication complexity of the protocol is linear in the vector length \genericVecLen. Further technical details can be found in \cite{weng2021vole}. 




%


\noindent\textbf{Threshold Set Intersection (a.k.a., $t$-out-$T$ matching)} is a variant of the PSI problem where the intersection of two (or more) sets is revealed to the parties iff the number of items in the intersection are above a certain predefined threshold. More formally, given two sets \alicesSet and \bobsSet of size \genericSetSize, and a threshold $t$, the protocol outputs \intSet such that $\intSet \triangleq 	\alicesSet \cap \bobsSet$ iff $| \alicesSet \cap \bobsSet | \geq n - t$. Otherwise, the protocol outputs $\bot$.


%

The state-of-the-art threshold PSI protocol (henceforth referred to a \tPSI) is due to \citet{ghosh2019thresholdPSI}. The main observation underlying the protocol is that given \setDiff{\alicesSet}{\bobsSet}, the party holding \alicesSet can obtain $\alicesSet \cap \bobsSet = \alicesSet \setminus (\setDiff{\alicesSet}{\bobsSet})$. Thus, it suffices to build a threshold set reconciliation protocol where $\setDiff{\alicesSet}{\bobsSet}$ (respectively, \setDiff{\bobsSet}{\alicesSet}) is revealed to the parties iff $\setDiffCard{\alicesSet}{\bobsSet} \leq t$. The protocol is inspired in part by the set reconciliation protocol due to \citet{minsky2003setreconciliation}. The idea behind this protocol is as follows: \alice and \bob encode the items of their corresponding sets, \alicesSet and \bobsSet in roots of polynomials \alicesPoly and \bobsPoly, respectively. If $\setSize{\alicesSet \setminus \bobsSet} \leq t$, then 
$\degree{\gcd{\alicesPoly}{\bobsPoly}} \geq \genericSetSize - t$, and so $r(x) =\rationalfn{\alicesPoly}{\bobsPoly}$ is a rational function of degree at most $2t$ (after cancellation of common roots in the numerator and denominator). $r(x)$ can be uniquely interpolated with $2t + 1$ evaluation points.  The denominator of $r(x)$ gives $\setDiff{\alicesSet}{\bobsSet}$.




\tPSI builds on this idea and tweaks the protocol to ensure that the elements in $\bobsSet \setminus \alicesSet$ are never revealed to \alice.  \alice and \bob evaluate a polynomial \tPSIOutputPoly at $3t + 1$ points, where \alicesRandomPoly is a degree-$t$ random polynomial contributed by \bob and \bobsRandomPoly is a random degree-$t$ polynomial contributed by \bob . \alice (and \bob) then compute the values of the rational function $r(x) = \frac{\tPSIOutputPoly}{\alicesPoly}$ at the aforementioned $3t+1$ points. Clearly, if $\setDiffCard{\alicesSet}{\bobsSet} \leq t$, then $r(x)$ has a numerator of degree $\leq 2t$ and a denominator of degree $\leq t$ after cancellation of the common roots in \gcd{\alicesPoly}{\bobsPoly} 
and \alicesPoly. Since $r(x)$ is a rational functions of degree $\leq 3t$, it can be uniquely interpolated with the $3t+1$ evaluation points. 

The security of the scheme relies on showing that the numerator of $r(x)$ after cancellation is a uniformly random polynomial. Specifically, let \tPSIReducedOutputPoly be the numerator after canceling common roots in $r$. The following well-known result due to \citet{Kissner2005setops} shows that this polynomial is uniformly random.


\begin{lemma}[\cite{Kissner2005setops}]
	\label{lemma:kissener_uniformly_random}
	Given two polynomials $\alicesPoly, \bobsPoly \in \polyField$ with $\degree{\alicesPoly} = \degree{\bobsPoly} \leq D_p$ such that \gcd{\alicesPoly}{\bobsPoly} = 1, and two uniformly random polynomials, \alicesRandomPoly, \bobsRandomPoly of degree $D_r \geq D_p$, the polynomial \tPSIOutputPoly is a uniformly random polynomial of degree $\leq D_r + D_p$.
\end{lemma}

\section{Protocol for Hamming Distances}
\label{sec:hammingawarePSI}

We start with a protocol for privately computing a Hamming distance-aware set intersection between sets where elements are drawn from  the universe $\universe = \{0,1\}^\genericVecLen$. \figref{fig:ham_psi_func} defines the ideal functionality \hamAwarePSIFunc for Hamming DA-PSI between two parties with tunable true positive and true negative rates. We propose a protocol with \hamPSICommComplexity communication cost for set sizes \genericSetSize (i.e., the cost is \textit{independent of the vector length}), and compute time that scales polynomially in \hamDistThreshold.

\myparagraph{Remark on Ideal Functionality} We have defined \hamAwarePSIFunc such that for each $(\alicesVector, \bobsVector) \in \alicesSet \times \bobsSet$, both parties learn $(\alicesVector, \bobsVector)$ 
iff $\hamDist{\alicesVector}{\bobsVector} \leq \hamDistThreshold$. Another definition considering is where \alice only learns if there is a $\bobsVector \in \bobsSet$ such that $\hamDist{\alicesVector}{\bobsVector} \leq \hamDistThreshold$ \textit{but not \bobsVector itself}. However, this definition may not be meaningful in the context of distance aware applications. For instance, there are $\binom{\genericVecLen}{\hamDistThreshold}$ elements that are within Hamming distance \hamDistThreshold of an element $\alicesVector \in \alicesSet$; it is not straightforward for \alice to guess \bobsVector simply from the fact that $\hamDist{\alicesVector}{\bobsVector} \leq \hamDistThreshold$ .
This is unlike traditional PSI, where \alice can trivially guess \bob's element knowing that there is a match. 
Nonetheless, both our Hamming DA-PSI protocol have an additional $\mathsf{Recover}$ step where \alice obtains \bobsVector from \bob after learning some intermediate results which indicates $\hamDist{\alicesVector}{\bobsVector} \leq \hamDistThreshold$. The protocol may be aborted at this stage (to save one extra round of communication) to realize an ideal functionality which only enables \alice to learn \textit{if} $\hamDist{\alicesVector}{\bobsVector} \leq \hamDistThreshold$ without directly revealing \bobsVector.




 \begin{figure}[h]

	\begin{mdframed}
		\footnotesize
		
		\underline{$\hamAwarePSIFunc$: Ideal Functionality for Hamming Aware DA-PSI:}
		
		\smallskip \noindent
		\textbf{Parameters:~} Parties \alice and \bob. Universe $\universe = \{0,1\}^\genericVecLen$. Hamming distance threshold $\hamDistThreshold$. True positive rate \tpr and true negative rate \tnr.  
		
		\smallskip
		\noindent
		\textbf{Inputs:~} \alice has input $\alicesSet = \{\vec{a_i}\}_{i=1}^{\genericSetSize} \subseteq \universe$. \bob has input $\bobsSet = \{\vec{b_j}\}_{j=1}^{\genericSetSize} \subseteq \universe$.
		
		\smallskip
		\noindent
		\textbf{Output:~} \alice and \bob learn 
		$S \subseteq \alicesSet \times \bobsSet$ where for each $(\alicesVector, \bobsVector) \in \alicesSet \times \bobsSet$, if $\hamDistLessThanThreshold{\alicesVector}{\bobsVector}$ then $\prob{(\alicesVector, \bobsVector) \in S} \ge \tpr$ and if \hamDistGreaterThanThreshold{\alicesVector}{\bobsVector} then $\prob{(\alicesVector, \bobsVector) \not\in S} \ge \tnr$.

		\medskip \underline{$\tHamQueryFunc$: Ideal Functionality for Threshold Hamming Query:}
		
			\smallskip\noindent
			\textbf{Parameters:~} Parties \alice and \bob. Universe $\universe = \{0,1\}^\genericVecLen$. Hamming distance thresholds $\hamDistThreshold$, true positive rate \tpr, and true negative rate \tnr. 
			
			\smallskip
			\noindent
			\textbf{Inputs:~} \alice has vector $\alicesVector \in \universe$ and \bob has vector $\bobsVector \in \universe$. 
			
			\smallskip
			\noindent
			\textbf{Output:~} If \hamDistLessThanThreshold{\alicesVector}{\bobsVector}, then \alice and \bob learn 
			 $(\alicesVector, \bobsVector)$ with probability $\ge \tpr$. If \hamDistGreaterThanThreshold{\alicesVector}{\bobsVector}, 
			 then \alice and \bob learn $\bot$ with probability $\ge \tnr$. 

	\end{mdframed}
	\vspace{-15pt}
	\caption{\small Ideal functionalities for Hamming distance-aware PSI and threshold Hamming queries \label{fig:ham_psi_func}}
\end{figure}

 \subsection{Threshold Hamming Query}
 The key building block of our construction is a protocol to privately determine if two bit vectors are within 
 a certain Hamming distance of each other. We call this primitive a {\em threshold Hamming query}. \figref{fig:ham_psi_func} defines the ideal functionality \tHamQueryFunc for \genericVecLen-bit vectors and Hamming distance threshold \hamDistThreshold.

%
%

\begin{figure}
	\footnotesize
	\begin{mdframed}
		
		
		
		
		\noindent

		\textbf{\underline{Parameters:}} \alice and \bob have vectors $\alicesVector, \bobsVector \in \{0,1\}^{\genericVecLen}$, respectively, and a Hamming distance threshold $\hamDistThreshold \in [0, \genericVecLen/2)$.

		\smallskip
		\textbf{\underline{Procedure $\mathsf{Map}$}:}

		 \alice and \bob sample \genericVecLen injective mapping functions \mapSet, \mapDomain. \alice computes set $\alicesSetFromVec \assign \{ s_\vecBitIdx : s_\vecBitIdx \assign \mapSetIdx{\vecBitIdx}(\alicesVector[\vecBitIdx]) \}$ and 
		\bob computes 
		$\bobsSetFromVec \assign \{ s_\vecBitIdx : s_\vecBitIdx \assign \mapSetIdx{\vecBitIdx}(\bobsVector[\vecBitIdx]) \}$. \label{step:ham_query_lite:encoding}
	    
		\smallskip
		\textbf{\underline{Procedure \oneSidedSetRecon:}}
		
		\begin{enumerate}[nosep,leftmargin=1.5em,labelwidth=*,align=left,labelsep=0.25em]
			
			
			\item \alice and \bob select a set of \numPointsForInterpolPlus points in \finiteField $\setOfXCoords \assign \setOfXCoordsDefn$ such that none of the points  are in the ranges of any of the mapping functions.

			\item \alice encodes \alicesSetFromVec  in the polynomial \alicesPoly = \polyFromSet{\alicesSetFromVec} and \bob encodes \bobsSetFromVec in the polynomial \bobsPoly = \polyFromSet{\bobsSetFromVec}. \label{step:set_recon:poly_compute}

			\item \bob samples two degree-\genericVecLen random polynomials $\alicesRandomPoly, \bobsRandomPoly \getsr \polyField$. 
			
			\item For each $x_\ptIdx \in \setOfXCoords$, \alice sends \alicesPolyFromVecEval{x_\ptIdx} and \bob sends \alicesRandomPolyEval{x_\ptIdx} and $\bobsRandomPolyEval{x_\ptIdx} \times \bobsPolyFromVecEval{x_\ptIdx}$ to \oleFunc. \alice learns \alicesVoleOutputPolyEval{x_\ptIdx} \assign \tPSIOutputPolyEval{x_\ptIdx} as the output of \oleFunc. \label{step:set_recon:ole_step_1}
			
			
			\item For $\ptIdx \in [1, \numPointsForInterpol]$, \alice computes the set of points \label{step:set_recon:point_computation}
			
			\begin{center}
				$\setOfPtsForInterpolation \assign  \{(x_\ptIdx, y_\ptIdx) : y_\ptIdx \assign \frac{\alicesVoleOutputPolyEval{x_\ptIdx}}{\alicesPolyFromVecEval{x_\ptIdx}} =  \frac{\tPSIOutputPolyEval{x_\ptIdx}}{\alicesPolyFromVecEval{x_\ptIdx}} \}$.
			\end{center}
			
			\item \alice interpolates \setOfPtsForInterpolation with a rational function $r(x) \assign \frac{\numeratorPoly}{\denominatorPoly}$ and checks that \denominatorPoly is a factor of \alicesPoly. If so, \alice outputs $\alicesSetFromVec \setminus \bobsSetFromVec$ which contains the roots of \denominatorPoly, otherwise \alice outputs $\bot$. 	\label{step:set_recon:interpolation}  
			
		\end{enumerate}
		
		\smallskip
		\textbf{\underline{Procedure $\mathsf{Recover}$}:} 
         
         If \alice receives $\bot$ from \oneSidedSetRecon then output $\bot$. Otherwise, upon receiving $\alicesSetFromVec \setminus \bobsSetFromVec$: for each  $s \in \alicesSetFromVec \cap \bobsSetFromVec = \alicesSetFromVec \setminus (\alicesSetFromVec \setminus \bobsSetFromVec)$, if $s =  \mapSetIdx{\vecBitIdx}(\alicesVector[\vecBitIdx])$ 
	then $\bobsVector[\vecBitIdx] = \alicesVector[\vecBitIdx]$.   For all indices $\vecBitIdx' \in [1, \genericVecLen]$ that are left undetermined from \setDiff{\alicesSetFromVec}{\bobsSetFromVec}, $\bobsVector[\vecBitIdx'] = \comp(\alicesVector[\vecBitIdx'])$. Output (\alicesVector, \bobsVector).
		

		


\end{mdframed}
\vspace{-15pt}
\caption{\small \tHamQueryLite: A first pass Hamming query protocol \label{fig:ham_query_lite}}
\end{figure}

\subsubsection{\tHamQueryLite: Hamming Query First Pass}
We start with a simple and insecure version of our threshold Hamming query protocol dubbed \tHamQueryLite (\figref{fig:ham_query_lite}). The key observation is that \tHamQueryFunc can be realized as follows:

\begin{enumerate}[nosep,leftmargin=1.6em,labelwidth=*,align=left]
    \item \textbf{Map:} We use \genericVecLen deterministic, injective mapping functions \mapSet where \mapDomain, such that 
    the ranges of the functions do not overlap. These functions map the individual bits in the vectors to elements of \finiteField. The $\vecBitIdx$th bit of vector \alicesVector, denoted $\alicesVector[\vecBitIdx]$, is mapped to element $M_m(\alicesVector[\vecBitIdx])$. \alicesVector is then uniquely represented by \alicesSetFromVecDefn. Correspondingly,  \bobsVector is represented by \bobsSetFromVecDefn.

    \item \textbf{Threshold Set Reconciliation:} A protocol with inputs \alicesSetFromVec and \bobsSetFromVec allows \alice to learn $\setDiff{\alicesSetFromVec}{\bobsSetFromVec}$ iff
$\setDiffCard{\alicesSetFromVec}{\bobsSetFromVec} \leq \hamDistThreshold$. \alice learns \bobsVector from $\setDiff{\alicesSetFromVec}{\bobsSetFromVec}$ . E.g., let \alicesVector = $1001$, \bobsVector = $1011$, \alicesSetFromVec = $\{\mappingFn_1(1)$, $\mappingFn_2(0)$, $\mappingFn_3(0)$, $\mappingFn_4(1) \}$ and \bobsSetFromVec = $\{\mappingFn_1(1)$, $\mappingFn_2(0)$, $\mappingFn_3(1)$, $\mappingFn_4(1) \}$. Then,
\alicesSetFromVec $\setminus$ \bobsSetFromVec = $\{\mappingFn_3(0)\}$ and $\bobsVector[1] = \alicesVector[1]$, $\bobsVector[2] = \alicesVector[2]$, $\bobsVector[4] = \alicesVector[4]$ and $\bobsVector[3] = \comp{\alicesVector[3]}$.

\end{enumerate}



The mapping functions have no bearing on the security of the protocol, as long as the ranges do not overlap and the functions are injective. In our implementations we have used PRFs, but we do not rely on their \textit{randomness} guarantees.

\myparagraph{Threshold Set Reconciliation}
In \tHamQueryLite, we use \oneSidedSetRecon, a new private set reconciliation protocol which is based on the \tPSI protocol (see \secref{sec:background}). \oneSidedSetRecon allows one of the parties to learn \setDiff{\alicesSetFromVec}{\bobsSetFromVec} (say \alice) while the other party generates all the random coins. Both parties begin by encoding the items in their respective sets in the roots of polynomials \alicesPolyFromVec and \bobsPolyFromVec respectively (see \lineref{step:set_recon:poly_compute} of \oneSidedSetRecon in \figref{fig:ham_query_lite}). This is followed by the parties jointly computing the evaluations of the polynomial \tPSIOutputPoly at \numPointsForInterpol points, where \alicesRandomPoly and \bobsRandomPoly are random polynomials sampled by \bob. This is achieved using \numPointsForInterpol calls to \oleFunc where \alice sends evaluations of \alicesPoly at each of the points and \bob correspndingly sends evaluations of \alicesRandomPoly, and $\bobsPolyFromVec \bobsRandomPoly$ (see \figref{fig:ole_in_set_recon}).  Finally, \alice interpolates the rational function $\frac{\tPSIOutputPoly}{\alicesPolyFromVec}$ with the evaluations of \tPSIOutputPoly similar to \tPSI (\lineref{step:set_recon:interpolation} of \oneSidedSetRecon) and obtains \setDiff{\alicesSetFromVec}{\bobsSetFromVec} iff $\setDiffCard{\alicesSetFromVec}{\bobsSetFromVec} \leq \hamDistThreshold$. 
\oneSidedSetRecon is significantly simpler than \tPSI and has lower communication cost. This is because it requires only \numPointsForInterpol calls to \oleFunc compared to twice as many calls in \tPSI while also avoiding two extra rounds of communication. The improvement comes from the fact that in contrast to \tPSI where both parties learn the results, \oneSidedSetRecon enables only \alice to learn the final result (see \appref{app:tpsi_delta} for more details).




 \tikzset{every picture/.style={line width=0.75pt}} 
\begin{figure}
\centering
\tikzset{every picture/.style={line width=0.75pt}} 

\begin{tikzpicture}[x=0.75pt,y=0.75pt,yscale=-1,xscale=0.5]

\draw (88,20) node [anchor=north west][inner sep=0.75pt]    {\footnotesize \underline{\alice}};
\draw (488,20) node [anchor=north west][inner sep=0.75pt]    {\footnotesize \underline{\bob}};
\draw (388,35) node [anchor=north west][inner sep=0.75pt]    {\footnotesize $\alicesRandomPoly, \bobsRandomPoly \getsr \polyField $};

\draw   (279,76) -- (373.5,76) -- (373.5,100) -- (279,100) -- cycle ;

\draw    (137.5,80) -- (271,80) ;
\draw [shift={(273,80)}, rotate = 180] [color={rgb, 255:red, 0; green, 0; blue, 0 }  ][line width=0.75]    (10.93,-3.29) .. controls (6.95,-1.4) and (3.31,-0.3) .. (0,0) .. controls (3.31,0.3) and (6.95,1.4) .. (10.93,3.29)   ;

\draw    (495.5,80) -- (379.5,80) ;
\draw [shift={(377.5,80)}, rotate = 359.74] [color={rgb, 255:red, 0; green, 0; blue, 0 }  ][line width=0.75]    (10.93,-3.29) .. controls (6.95,-1.4) and (3.31,-0.3) .. (0,0) .. controls (3.31,0.3) and (6.95,1.4) .. (10.93,3.29)   ;

\draw    (496.5,98) -- (380.5,98) ;
\draw [shift={(378.5,98)}, rotate = 359.74] [color={rgb, 255:red, 0; green, 0; blue, 0 }  ][line width=0.75]    (10.93,-3.29) .. controls (6.95,-1.4) and (3.31,-0.3) .. (0,0) .. controls (3.31,0.3) and (6.95,1.4) .. (10.93,3.29)   ;

\draw    (275.5,98) -- (98.5,98) ;
\draw [shift={(96.5,98)}, rotate = 0.32] [color={rgb, 255:red, 0; green, 0; blue, 0 }  ][line width=0.75]    (10.93,-3.29) .. controls (6.95,-1.4) and (3.31,-0.3) .. (0,0) .. controls (3.31,0.3) and (6.95,1.4) .. (10.93,3.29)   ;

\draw (300,80.4) node [anchor=north west][inner sep=0.75pt]    {$\oleFunc$};
\draw (388,64) node [anchor=north west][inner sep=0.75pt]    {\footnotesize $\alicesRandomPolyEval{x_1}$};

\draw (388,103.4) node [anchor=north west][inner sep=0.75pt]    {\footnotesize $\bobsRandomPolyEval{x_1} \bobsPolyFromVecEval{x_1}$};

\draw (184,64) node [anchor=north west][inner sep=0.75pt]    {\footnotesize $\alicesRandomPolyEval{x_1}$};

\draw (30,103.4) node [anchor=north west][inner sep=0.75pt]    {\footnotesize $\tPSIOutputPolyEval{x_1}$};

\draw (316,120) node [anchor=north west][inner sep=0.75pt]   [align=left] {...};
\draw (316,130) node [anchor=north west][inner sep=0.75pt]   [align=left] {...};

\draw   (279,146) -- (373.5,146) -- (373.5,170) -- (279,170) -- cycle ;

\draw    (137.5,150) -- (271,150) ;
\draw [shift={(273,150)}, rotate = 180] [color={rgb, 255:red, 0; green, 0; blue, 0 }  ][line width=0.75]    (10.93,-3.29) .. controls (6.95,-1.4) and (3.31,-0.3) .. (0,0) .. controls (3.31,0.3) and (6.95,1.4) .. (10.93,3.29)   ;

\draw    (495.5,150) -- (379.5,150) ;
\draw [shift={(377.5,150)}, rotate = 359.74] [color={rgb, 255:red, 0; green, 0; blue, 0 }  ][line width=0.75]    (10.93,-3.29) .. controls (6.95,-1.4) and (3.31,-0.3) .. (0,0) .. controls (3.31,0.3) and (6.95,1.4) .. (10.93,3.29)   ;

\draw    (496.5,168) -- (380.5,168) ;
\draw [shift={(378.5,168)}, rotate = 359.74] [color={rgb, 255:red, 0; green, 0; blue, 0 }  ][line width=0.75]    (10.93,-3.29) .. controls (6.95,-1.4) and (3.31,-0.3) .. (0,0) .. controls (3.31,0.3) and (6.95,1.4) .. (10.93,3.29)   ;

\draw    (275.5,168) -- (98.5,168) ;
\draw [shift={(96.5,168)}, rotate = 0.32] [color={rgb, 255:red, 0; green, 0; blue, 0 }  ][line width=0.75]    (10.93,-3.29) .. controls (6.95,-1.4) and (3.31,-0.3) .. (0,0) .. controls (3.31,0.3) and (6.95,1.4) .. (10.93,3.29)   ;

\draw (300,150.4) node [anchor=north west][inner sep=0.75pt]    {$\oleFunc$};
\draw (388,134) node [anchor=north west][inner sep=0.75pt]    {\footnotesize $\alicesRandomPolyEval{x_\ptIdx} $};
\draw (388,173.4) node [anchor=north west][inner sep=0.75pt]    {\footnotesize $\bobsRandomPolyEval{x_\ptIdx} \bobsPolyFromVecEval{x_1}$};
\draw (184,134) node [anchor=north west][inner sep=0.75pt]    {\footnotesize $\alicesRandomPolyEval{x_\ptIdx}$};
\draw (30,173.4) node [anchor=north west][inner sep=0.75pt]    {\footnotesize $\tPSIOutputPolyEval{x_\ptIdx}$};




\end{tikzpicture}
\vspace{-10pt}
\caption{Using OLE for \oneSidedSetRecon \label{fig:ole_in_set_recon}}
\end{figure}
 \subsubsection{The (In)Security of \tHamQueryLite} 
 \label{sec:tHamQueryfromPSI}

\tHamQueryLite is not secure across all input parameters, and as we will show in this section, reveals $\bobsVector$ to \alice when $\hamDist{\alicesVector}{\bobsVector} \in (\hamDistThreshold, 2\hamDistThreshold)$. This is because \oneSidedSetRecon reveals information when $|\alicesSetFromVec \setminus \bobsSetFromVec| \in (\hamDistThreshold, 2\hamDistThreshold)$ . In fact, the protocol of \citet{ghosh2019thresholdPSI} on which \oneSidedSetRecon is based also has the same leakage, and while the authors caution against using it as a standalone protocol\footnote{To address the leakage, the paper proposes a significantly more expensive threshold cardinality of intersection protocol.}, they have not analyzed this. More formally, we prove the following result.


\begin{thm}
\label{thm:set_recon_claims}
Given sets \alicesSetFromVec and \bobsSetFromVec such that $\setSize{\alicesSetFromVec} = \setSize{\bobsSetFromVec}$ as inputs to \oneSidedSetRecon, the following results hold:

\begin{itemize}[nosep,leftmargin=1.6em,labelwidth=*,align=left]
    \item \textbf{Proposition 1:} If $\setDiffCard{\alicesSetFromVec}{\bobsSetFromVec} \ge 2\hamDistThreshold$, then there does not exist a PPT adversary that can determine any information regarding \bobsSetFromVec from \oneSidedSetRecon with more than negligible advantage (in \secParam) over guessing. \label{prop:set_diff_larger}
    
    \item \textbf{Proposition 2:} If $\setDiffCard{\alicesSetFromVec}{\bobsSetFromVec} \in (\hamDistThreshold, 2\hamDistThreshold) $, there exists an adversary that can determine \bobsSetFromVec from \oneSidedSetRecon with overwhelming probability (at least $1 - \neglfn{\secParam}$). \label{prop:set_diff_in_range}
\end{itemize}
\end{thm}

\begin{proofsketch}
We prove the result in  \appref{app:proof_set_recon_claims}. Here, we provide the key arguments behind the proof. 

\noindent\textbf{Proof of Proposition~1:}
Consider the evaluation points \alice computes in \lineref{step:set_recon:point_computation} corresponding to the rational function 
$\frac{\tPSIOutputPoly}{\alicesPolyFromVec} = \frac{\gcd{\alicesPolyFromVec}{\bobsPolyFromVec} \times \left(\tPSIReducedOutputPoly\right)}{\alicesPolyFromVec} = \frac{\tPSIReducedOutputPoly}{\alicesReducedPoly}$. Here, 
$\alicesReducedPoly \assign \frac{\alicesPolyFromVec}{\gcd{\alicesPolyFromVec}{\bobsPolyFromVec}}$. If the degree of $\gcd{\alicesPolyFromVec}{\bobsPolyFromVec} = \degreeOfGCD$, then \tPSIReducedOutputPoly is a random polynomial of degree $2\genericVecLen - \degreeOfGCD$. This is due to \lemmaref{lemma:kissener_uniformly_random} as  \alicesPolyFromVec, \bobsPolyFromVec, \alicesRandomPoly, \bobsRandomPoly are all degree-\genericVecLen polynomials, and $\alicesRandomPoly, \bobsRandomPoly \getsr \polyField$.

From the set of evaluation points, \setOfPtsForInterpolation, \alice may try to guess \bob's input polynomial \bobsPolyFromVec and check whether there is a polynomial \numeratorPoly of degree $2\genericVecLen - \degreeOfGCD + 1$ such that $\frac{\numeratorPoly}{\alicesReducedPoly}$ is consistent with \setOfPtsForInterpolation. We show in \appref{app:lemmaPolys}
that when $\setSize{\setOfPtsForInterpolation} \leq 2\genericVecLen - \degreeOfGCD$ which implies $\genericVecLen - \degreeOfGCD \ge 2\hamDistThreshold$ , for every possible \alicesReducedPoly there is at least one candidate polynomial for \numeratorPoly. Since \tPSIReducedOutputPoly is a random polynomial, any obtained value of \numeratorPoly is equally likely to be \tPSIReducedOutputPoly.  Moreover, if there are more than one candidates for \numeratorPoly, then they are all equally likely. Thus, Proposition~1 holds. 


\noindent\textbf{Proof of Proposition~2:}
When $\setSize{\setOfPtsForInterpolation} > 2\genericVecLen - \degreeOfGCD + 1$ which implies  $\genericVecLen - \degreeOfGCD < 2\hamDistThreshold$, the probability that \alice will find a candidate polynomial for \numeratorPoly such that $\frac{\numeratorPoly}{\alicesReducedPoly}$ is consistent with \setOfPtsForInterpolation \textit{when she has incorrectly guessed} \bobsPolyFromVec (and \alicesReducedPoly) is negligible in \secParam. And so, \alice may check all possible candidates for \bobsPolyFromVec and verify her guesses. The set of all possible values of \bobsPolyFromVec is smaller than the set of degree-\genericVecLen polynomials in \polyField since the roots of \bobsPolyFromVec are fixed by the mapping functions \mapSet. There are $2^{\genericVecLen}$ possible values of \bobsPolyFromVec, and for small \genericVecLen, the search is computationally feasible for a PPT adversary. Thus, Proposition~2 holds.
\end{proofsketch}



\begin{figure}
	\centering
	\includegraphics[scale=0.2]{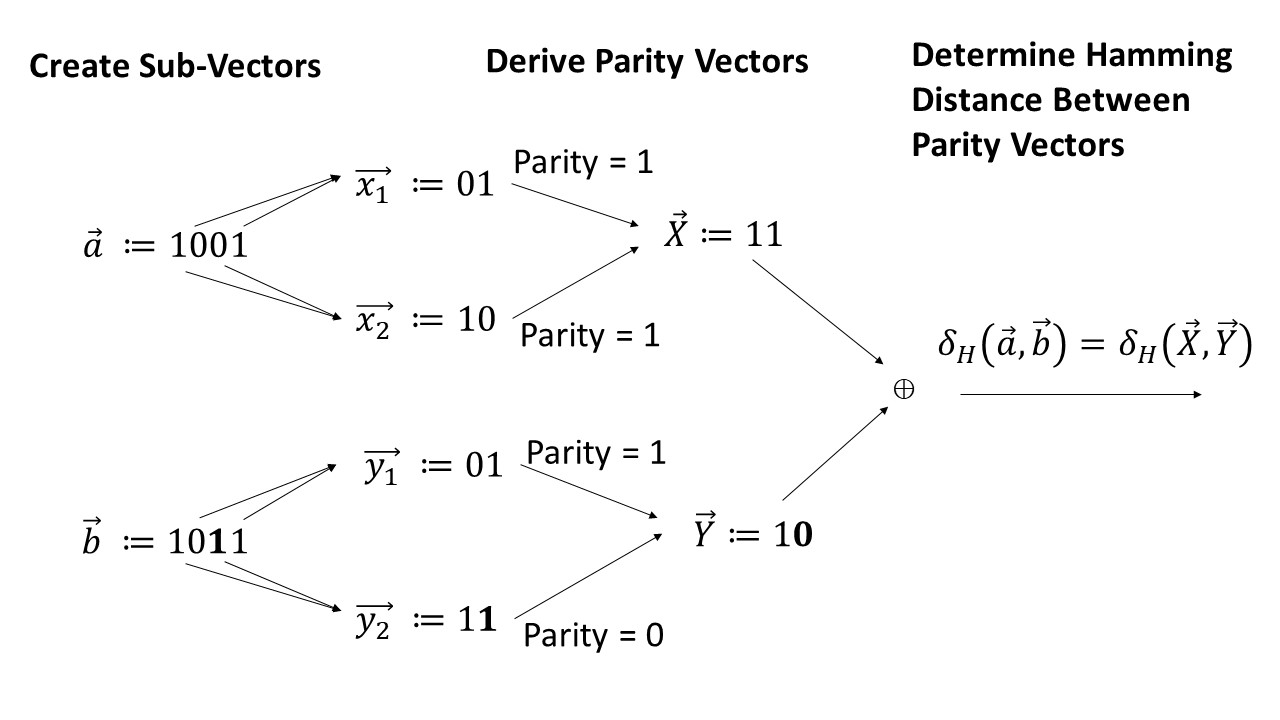}
	\vspace{-0.8cm}
	\caption{\small The process of computing \hamDist{\alicesVector}{\bobsVector} in \tHamQueryRestricted. \alice and \bob partition the bits in \alicesVector and \bobsVector respectively into sub-vectors $\alicesSubVecIdx{1}, \alicesSubVecIdx{2}$ and $\bobsSubVecIdx{1}, \bobsSubVecIdx{2}$. They compute \alicesParityVec and \bobsParityVec respectively using the parities of the sub-vectors. The bold bits show the locations where \alice's and \bob's inputs differ.}
	\label{fig:sub_sample_example}
\end{figure}


\begin{figure}
	\footnotesize
	\begin{mdframed}
	
	\noindent
	\textbf{\underline{Parameters}:} \alice and \bob have vectors \alicesVector and \bobsVector respectively, where 
	$\alicesVector, \bobsVector \in \{0,1\}^\genericVecLen, ~\hamDist{\alicesVector}{\bobsVector} \leq 2\hamDistThreshold$ where 
	$\hamDistThreshold$ is the Hamming distance threshold. False positive rate $\fpr \in (0, 0.5)$.

	\smallskip\noindent
	\textbf{\underline{Procedure \permAndPart}:}

	\begin{enumerate}[nosep,leftmargin=1.5em,labelwidth=*,align=left,labelsep=0.25em]
	    \item \alice and \bob sample a permutation $\randPerm : [1, \genericVecLen] \rightarrow [1, \genericVecLen]$ uniformly randomly from the set of all such permutations.\label{step:hamQueryRes:samplePerm}
	
		\item \alice computes $\alicesPermutedVec \in \{0,1\}^{\genericVecLen}$ after permuting the bits of \alicesVector as follows: 
		for $\vecBitIdx \in [1, \genericVecLen]$, $\alicesPermutedVec[\randPerm(\vecBitIdx)] \assign \alicesVector[\vecBitIdx]$.  Similarly, \bob computes $\bobsPermutedVec \in \{0,1\}^{\genericVecLen}$ such that for $\vecBitIdx \in [1, \genericVecLen]$ $\bobsPermutedVec[\randPerm(\vecBitIdx)] \assign \bobsVector[\vecBitIdx]$

		\item  \alice creates \numOfBinsRef = \numOfBins sub-vectors \alicesSetOfSubVecs where each sub-vector is created by a contiguous sequence of $\frac{\genericVecLen}{\numOfBinsRef}$ bits of \alicesPermutedVec. Specifically, 
		$\alicesSubVecIdx{1} \assign \alicesPermutedVec[1:\frac{\genericVecLen}{\numOfBinsRef}], ~\alicesSubVecIdx{2} \assign \alicesPermutedVec[\frac{\genericVecLen}{\numOfBinsRef} + 1: 2\times \frac{\genericVecLen}{\numOfBinsRef}], \ldots, \alicesSubVecIdx{\numOfBinsRef} \assign \alicesPermutedVec[ (\numOfBinsRef - 1) \times \frac{\genericVecLen}{\numOfBinsRef} + 1: \genericVecLen]$. Here $\alicesPermutedVec[\vecBitIdx : \vecBitIdx ']$ is the contiguous sequence of bits starting from index \vecBitIdx and up to (including) index $\vecBitIdx '$ in \alicesPermutedVec. \bob similarly creates \bobsSetOfSubVecs from \bobsPermutedVec. \label{step:hamQueryRes:createVecs}
	
		
		\item \alice computes vector $\vec{X} \in \{0,1\}^{\numOfBinsRef}$ such that for $\vecBitIdx \in [1, \numOfBinsRef]$,  $\vec{X}[\vecBitIdx] \assign \parity(\vec{x}_\vecBitIdx)$. Similarly, \bob computes vector $\vec{Y} \in \{0,1\}^{\numOfBinsRef}$ such that $\vec{Y}[\vecBitIdx] \assign \parity(\vec{y}_\vecBitIdx)$. \label{step:hamQueryRes:computeParityVecs}

	\end{enumerate}

	\smallskip\noindent
	\textbf{\underline{Protocol \hamCompute}:}

    \begin{enumerate}[resume,nosep,leftmargin=1.5em,labelwidth=*,align=left,labelsep=0.25em]
		
		\item \alice generates a key-pair $(\pubKey, \privKey)$ for a semantically-secure additively homomorphic encryption \AHESchemeDefn, and provides \pubKey to \bob.

		\item For $m \in [1, \numOfBinsRef]$, \alice computes $ct_m \assign \encrypt{\alicesParityVec[\vecBitIdx]}$ and $ct_w \assign \norm{\alicesParityVec}$  where $\norm{.}$ is the Hamming weight of the input vector.  \alice sends $\{ct_w, ct_1, \ldots, ct_{\numOfBinsRef}\}$ to \bob. \label{step:hamQueryRes:sendEnc}


		\item \bob computes $\encrypt{\hamDist{\alicesParityVec}{\bobsParityVec}} \assign ct_w \encAdd ct'_w \encSub 2 \encMult 
		\left((\bobsParityVec[1] \encMult ct'_1) \encAdd \ldots \encAdd (\bobsParityVec[\numOfBinsRef]\encMult ct'_{\numOfBinsRef})\right)$
		 where  $ct'_w \assign \norm{\bobsParityVec}$,  \encAdd (\encSub) is homomorphic addition (subtraction) of ciphertexts, and \encMult is scalar multiplication.  \label{step:hamQueryRes:computeHamDist} 
		 
	       \item \bob samples $\prfKey \getsr \finiteField$ and computes $\keySet \assign \{\prfKey_\setIdx : \prfKey_\setIdx \assign r_\setIdx \encMult \left(\encrypt{\hamDist{\alicesParityVec}{\bobsParityVec}} \encSub \encrypt{\setIdx}\right) \encAdd \encrypt{\prfKey}, ~\setIdx \in [0,\hamDistThreshold], ~r_\setIdx \getsr \finiteField\}$.
           \bob returns \keySet to \alice.  \label{step:hamQueryRes:computeSet}

		\item \alice computes $\keySet ' \assign \{\prfKey'_i ~:~ \prfKey'_i \assign \decrypt{\prfKey_i}, ~\prfKey_i \in \keySet \}$. If $\hamDist{\alicesVector}{\bobsVector} = \hamDist{\alicesParityVec}{\bobsParityVec} \leq \hamDistThreshold$, $\prfKey \in \keySet'$. \alice outputs $\keySet'$. \label{step:hamQueryRes:output}
		
	\end{enumerate}
\end{mdframed}
\vspace{-15pt}
\caption{\tHamQueryRestricted: Threshold Hamming query over restricted domain \label{fig:hamQueryRestricted}}
\end{figure}

\subsubsection{\tHamQuery: Hamming Queries with Polynomial Computation}
\label{sec:hamming_psi:poly_compute}
One way to fix \tHamQueryLite is by checking if  $\hamDist{\alicesVector}{\bobsVector} \in (\hamDistThreshold, 2\hamDistThreshold)$; however, implementing this as a precursor to \tHamQueryLite reveals information regarding \hamDist{\alicesVector}{\bobsVector}. We propose a protocol with \textit{communication cost independent of the length of the vectors}, \genericVecLen where the cases $\hamDist{\alicesVector}{\bobsVector} \in (\hamDistThreshold, 2\hamDistThreshold)$ and $\hamDist{\alicesVector}{\bobsVector} \geq 2\hamDistThreshold$ are indistinguishable.

\myparagraph{Hamming Queries Over Restricted Domain}
As the starting point, we present a protocol which distinguishes $\hamDist{\alicesVector}{\bobsVector} \leq \hamDistThreshold$
and $\hamDist{\alicesVector}{\bobsVector} \in (\hamDistThreshold, 2\hamDistThreshold]$. The additional constraint is that for all inputs \alicesVector and \bobsVector, the maximum Hamming distance between them is known apriori to be  $\leq 2\hamDistThreshold$. To generalize over the entire domain, we will subsequently extend this protocol and integrate with \oneSidedSetRecon.

The protocol dubbed \tHamQueryRestricted  is inspired by a result due to \citet{huang2006hamming}. The intuition is as follows: let $S_I$ be the set of indices where $\alicesVector$ and $\bobsVector$ differ. By definition of the problem, $\setSize{S_I} \le 2\hamDistThreshold$. Consider the following balls and bins analysis: let the indices where \alicesVector and \bobsVector differ be represented by balls that are thrown randomly into \numOfBins empty bins, where $\fpr \in (0,0.5)$. Then, the following result shows that \textit{all bins have $\leq 1$ ball with probability at least $1 - \fpr$}. It may be evident that  the number of non-empty bins gives us $\hamDist{\alicesVector}{\bobsVector}$.

\begin{fact}[\cite{huang2006hamming}]
	\label{fact:balls_bins_1}
	If 2\hamDistThreshold balls are randomly thrown into \numOfBins bins, where $\fpr \in (0,0.5)$, then with probability at most $\fpr$, there is one or more bins with more than one ball.
\end{fact} 

\figref{fig:hamQueryRestricted} describes the \tHamQueryRestricted protocol built around this idea.
The protocol comprises two procedures \permAndPart and \hamCompute. \permAndPart
uses a random permutation of the vectors to create $\numOfBinsRef = \numOfBins$ sub-vectors. Specifically, the bits in \alicesVector are partitioned into \numOfBinsRef partitions (each corresponding to a sub-vector)
after permuting with the random permutation. The resulting sub-vectors are denoted \alicesSetOfSubVecs. Similarly, \bobsVector is partitioned into  \bobsSetOfSubVecs (\linesref{step:hamQueryRes:samplePerm}{step:hamQueryRes:createVecs}). Then, \alice and \bob 
create parity vectors \alicesParityVec and \bobsParityVec using the parities of the sub-vectors (\lineref{step:hamQueryRes:computeParityVecs}).

\hamCompute privately computes the Hamming distance between the parity vectors and compares it with the distance threshold. 
\alice sends \alicesParityVec to \bob, with each bit encrypted individually (\lineref{step:hamQueryRes:sendEnc}). 
\bob computes $\encrypt{\hamDist{\alicesParityVec}{\bobsParityVec}}$ by computing the Hamming distance over encrypted bits (\lineref{step:hamQueryRes:computeHamDist}). 
The encryption scheme used is additively homomorphic, which ensures that \bob can compute over the encrypted bits.
This gives $\encrypt{\hamDist{\alicesVector}{\bobsVector}}$ due to following fact: since each pair $\alicesSubVecIdx{i}, \bobsSubVecIdx{i}$ can differ in at most one bit due to Fact \ref{fact:balls_bins_1}, $\hamDist{\alicesSubVecIdx{i}}{\bobsSubVecIdx{i}} = \hamDist{\parity(\alicesSubVecIdx{i})}{\parity(\bobsSubVecIdx{i})}$. Then, $\hamDist{\alicesVector}{\bobsVector} = \sum\limits_{i = 1}^{\numOfBinsRef}\hamDist{\parity(\alicesSubVecIdx{i})}{\parity(\bobsSubVecIdx{i})} = \hamDist{\alicesParityVec}{\bobsParityVec}$ (see \figref{fig:sub_sample_example}).

Finally, \bob samples a key $\prfKey \getsr \finiteField$ and returns a set \keySet containing \prfKey "blinded" by random values, and available to \alice iff. $\hamDist{\vec{X}}{\vec{Y}} \in [0, \hamDistThreshold]$. Specifically, for each $i \in [0,\hamDistThreshold]$, \bob returns $r_i \times (\hamDist{\vec{X}}{\vec{Y}} - i) + \prfKey$ in the \keySet to \alice where $r_i \getsr \finiteField$. \alice obtains \prfKey only when $\hamDist{\vec{X}}{\vec{Y}} \in [0, \hamDistThreshold]$ (\linesref{step:hamQueryRes:computeSet}{step:hamQueryRes:output}). Otherwise, $r_i \times (\hamDist{\vec{X}}{\vec{Y}} - i) + \prfKey \getsr \finiteField$.

\begin{figure}
	\footnotesize
	\begin{mdframed}
		
		\noindent
		\textbf{\underline{Parameters:}}  \alice and \bob have vectors 
		$\alicesVector, \bobsVector \in \{0,1\}^{\genericVecLen}$ respectively. and a Hamming distance threshold $\hamDistThreshold \in [0,\genericVecLen/2)$. 
		
		\smallskip\noindent
		\textbf{\underline{Procedure \permAndPart:}}
		
		\alice and \bob run \tHamQueryRestricted with \alicesVector and \bobsVector as inputs. \alice obtains $\alicesSetFromVecBins := \{\vec{x}_1, \ldots, \vec{x}_{\numOfBinsRef}\}$ and \keySet, while \bob obtains $\bobsSetFromVecBins := \{\vec{y}_1, \ldots, \vec{y}_{\numOfBinsRef}\}$ where $\numOfBinsRef = \numOfBins$.  
			
		\smallskip\noindent
		\textbf{\underline{Procedure \oneSidedSetReconBlind:}}
		\begin{enumerate}[nosep,leftmargin=1.5em,labelwidth=*,align=left,labelsep=0.25em]	
		
				\item \alice and \bob select a set of \numPointsForInterpolAfterBins points in \finiteField, $\setOfXCoords \assign \setOfXCoordsDefnBins$ such that none of the points  are in \alicesSetFromVecBins and \bobsSetFromVecBins. 
		
			\item \alice and \bob compute $\alicesPolyFromVec \assign \polyFromSet{\alicesSetFromVecBins}$ and $\bobsPolyFromVec \assign \polyFromSet{\bobsSetFromVecBins}$ respectively. \bob samples two random polynomials $\alicesRandomPoly, \bobsRandomPoly \getsr \polyField$ of degree \numOfBinsRef.  
				
			\item For each $x_\ptIdx \in \setOfXCoords$, \alice sends \alicesPolyFromVecEval{x_\ptIdx} to \oleFunc, while \bob sends 
			\alicesRandomPolyEval{x_\ptIdx}, and $\bobsRandomPolyEval{x_\ptIdx} \times \bobsPolyFromVecEval{x_\ptIdx} + \finiteFieldPRF(\kappa, \ptIdx)$. \alice obtains $\alicesVoleOutputPolyEval{x_\ptIdx} \assign \tPSIOutputPolyEval{x_\ptIdx} + \finiteFieldPRF(\prfKey, \ptIdx)$
			
			\label{step:ham_query:ole_step_1}
			
			\item For each element $\prfKey_\setIdx \in \keySet$, \alice obtains a candidate key, $\prfKey^{'}_\setIdx \assign \decrypt{\prfKey_\setIdx}$. \alice computes  \label{step:ham_query:comp_pts}
			
			\begin{center}
				$\setOfPtsForInterpolation_\setIdx \assign  \{(x_\ptIdx, y_\ptIdx) : y_\ptIdx \assign \frac{\alicesVoleOutputPolyEval{x_\ptIdx} - \finiteFieldPRF(\prfKey^{'}_{\setIdx}, \ptIdx)}{\alicesPolyFromVecEval{x_\ptIdx}}\}$.
			\end{center}
			
			\item For each $\setIdx \in [0, \hamDistThreshold]$, \alice interpolates $\setOfPtsForInterpolation_\setIdx$ with a rational function $r(x) \assign \frac{\numeratorPoly}{\denominatorPoly}$ and checks that \denominatorPoly is a factor of \alicesPolyFromVec. If so, \alice outputs $\alicesSetFromVecBins \setminus \bobsSetFromVecBins$ which contains the roots of \denominatorPoly, otherwise \alice outputs $\bot$. \label{step:ham_query:interpolation}
		\end{enumerate}

		\smallskip
		\textbf{\underline{Procedure $\mathsf{Recover}$}:} 
		
		If \alice receives $\bot$ from \oneSidedSetReconBlind then output $\bot$. Otherwise, upon receiving $\alicesSetFromVecBins \setminus \bobsSetFromVecBins$, obtain \bobsVector from \bob. Output $(\alicesVector, \bobsVector)$. 
\end{mdframed}
\vspace{-15pt }
\caption{\tHamQuery: Threshold Hamming query protocol \label{fig:hamQuery}}
\end{figure}

\myparagraph{General Hamming Queries}
We are ready to combine \tHamQueryRestricted and \oneSidedSetRecon to achieve a secure threshold Hamming query protocol, \tHamQuery. 
The protocol requires a PRF over a finite field  $\finiteFieldPRF: \finiteField \times \finiteField \rightarrow \finiteField$. The outline of this integration is (see \figref{fig:hamQuery}):

\begin{enumerate}[nosep,leftmargin=1.6em,labelwidth=*,align=left]
	
	\item \alice and \bob send \alicesVector and \bobsVector to \tHamQueryRestricted respectively. \alice obtains \keySet from which she can obtain \prfKey only when $\hamDist{\alicesVector}{\bobsVector} \notin (\hamDistThreshold, 2\hamDistThreshold]$.

	\item \alice builds set $\alicesSetFromVecBins := \{\vec{x}_1, \ldots, \vec{x}_{\numOfBinsRef}\}$ where $\vec{x}_1, \ldots, \vec{x}_{\numOfBinsRef}$ are the sub-vectors created during \permAndPart in \tHamQueryRestricted (see \lineref{step:hamQueryRes:createVecs} of \figref{fig:hamQueryRestricted}). \bob builds $\bobsSetFromVecBins := \{\vec{y}_1, \ldots, \vec{y}_{\numOfBinsRef}\}$.

	\item \alice and \bob run \oneSidedSetRecon with \alicesSetFromVecBins and \bobsSetFromVecBins as inputs and threshold \hamDistThreshold with one change: \bob modifies his inputs to \oleFunc such that \alice obtains a ``blinded'' set of evaluations, i.e., for $k \in [1, \numOfBinsRef + 2\hamDistThreshold +1]$, \alice obtains $\tPSIOutputPolyEval{x_\ptIdx}  + \finiteFieldPRF(\kappa,  \ptIdx)$ (see \figref{fig:hamQueryGenExample}). 
	
\end{enumerate}

\begin{restatable}{thm}{securityHamQuery}
	\label{thm:security_ham_query}
	Assuming that there exists a semantically-secure additively homomorphic encryption scheme that produces \bigO{\secParam}-bit ciphertexts,  and that there is a protocol for \oleFunc that requires \bigO{\secParam} bits of communication, for false positive rate $\fpr \in (0,0.5)$, \tHamQuery realizes \tHamQueryFunc with \bigO{\frac{\hamDistThreshold^2}{\fpr} \cdot \secParam} bits of communication and compute costs polynomial in \hamDistThreshold. 
\end{restatable}

\begin{proofsketch}
	The communication cost of the protocol is straightforward. \alice sends  $|\alicesParityVec| = \numOfBins$ encrypted bits to \bob. \bob sends back the encrypted $\keySet$ with $| \keySet| = \hamDistThreshold$ ciphertexts. Finally, there are $\numOfBins + 2\hamDistThreshold + 1$ calls to \oleFunc, each of which requires \bigO{\secParam} bits of communication. 	
	
	The following arguments show that the protocol is secure. \alice can ``unblind'' and obtain the correct evaluations of \tPSIOutputPoly iff she has obtained \prfKey in Step 1, which happens with high probability (at least $1 - \epsilon$) when $\hamDist{\alicesVector}{\bobsVector} \notin (\hamDistThreshold, 2\hamDistThreshold]$. Otherwise, \alice obtains random points as evaluations of \tPSIOutputPoly which reveals no information regarding \bobsSetFromVecBins. 
	If $\hamDist{\alicesVector}{\bobsVector} > 2\hamDistThreshold$, \alice may still obtain $\prfKey \in \keySet$ since Fact \ref{fact:balls_bins_1} is applicable only when $\hamDist{\alicesVector}{\bobsVector} \le 2\hamDistThreshold$. We show in \appref{app:hamming_query_poly} that $\setDiffCard{\alicesSetFromVecBins}{\bobsSetFromVecBins} \geq 2\hamDistThreshold$ with high probability (at least $1 - \epsilon$), and as \thmref{thm:set_recon_claims} shows, when $\setDiffCard{\alicesSetFromVecBins}{\bobsSetFromVecBins} \geq 2\hamDistThreshold$, \alice learns nothing about $\setDiff{\alicesSetFromVecBins}{\bobsSetFromVecBins}$ from the evaluations of \tPSIOutputPoly.
\end{proofsketch}


\tikzset{every picture/.style={line width=0.75pt}} 
\begin{figure}
	\centering
	\tikzset{every picture/.style={line width=0.75pt}} 
	
	\begin{tikzpicture}[x=0.75pt,y=0.75pt,yscale=-1,xscale=0.5]

		\draw (100,5) node [anchor=north west][inner sep=0.75pt]    {\footnotesize \underline{\alice}};

		\draw (1,20) node [anchor=north west][inner sep=0.75pt]    {\footnotesize Subsample \alicesVector: $\alicesSetFromVecBins := \{\vec{x}_1, \ldots, \vec{x}_{N}\}$};

		\draw (1,35) node [anchor=north west][inner sep=0.75pt]    {\footnotesize $\alicesPolyFromVec := \polyFromSet{\alicesSetFromVecBins}$};

		\draw (500,5) node [anchor=north west][inner sep=0.75pt]    {\footnotesize \underline{\bob}};

		\draw (350,20) node [anchor=north west][inner sep=0.75pt]    {\footnotesize Subsample \bobsVector: $\bobsSetFromVecBins := \{\vec{y}_1, \ldots, \vec{y}_{N}\}$};

		\draw (400,35) node [anchor=north west][inner sep=0.75pt]    {\footnotesize $\bobsPolyFromVec := \polyFromSet{\bobsSetFromVecBins}$};	
		
		\draw (400,51) node [anchor=north west][inner sep=0.75pt]    {\footnotesize $\alicesRandomPoly, \bobsRandomPoly \getsr \polyField$};

		\draw   (200,71) -- (453.5,71) -- (453.5,101) -- (200,101) -- cycle ;
		
		\draw    (7.5,80) -- (200,80) ;
		\draw [shift={(198,80)}, rotate = 180] [color={rgb, 255:red, 0; green, 0; blue, 0 }  ][line width=0.75]    (10.93,-3.29) .. controls (6.95,-1.4) and (3.31,-0.3) .. (0,0) .. controls (3.31,0.3) and (6.95,1.4) .. (10.93,3.29)   ;

		\draw    (585.5,87) -- (453,87) ;
		\draw [shift={(455,87)}, rotate = 359.74] [color={rgb, 255:red, 0; green, 0; blue, 0 }  ][line width=0.75]    (10.93,-3.29) .. controls (6.95,-1.4) and (3.31,-0.3) .. (0,0) .. controls (3.31,0.3) and (6.95,1.4) .. (10.93,3.29)   ;

		\draw    (1,97) -- (200,97) ;
		\draw [shift={(1,97)}, rotate = 0.32] [color={rgb, 255:red, 0; green, 0; blue, 0 }  ][line width=0.75]    (10.93,-3.29) .. controls (6.95,-1.4) and (3.31,-0.3) .. (0,0) .. controls (3.31,0.3) and (6.95,1.4) .. (10.93,3.29)   ;

		\draw (205,80.4) node [anchor=north west][inner sep=0.75pt]    {\small $\tHamQueryRestricted$};
		\draw (90,65) node [anchor=north west][inner sep=0.75pt]    {$\alicesVector$};
		
		\draw (500,70) node [anchor=north west][inner sep=0.75pt]    {$\bobsVector, \prfKey$};
		
		\draw (1,80) node [anchor=north west][inner sep=0.75pt]    {\footnotesize $\hamDist{\alicesVector}{\bobsVector} \leq \hamDistThreshold ? : \prfKey, \bot$};

		\draw   (279,126) -- (373.5,126) -- (373.5,150) -- (279,150) -- cycle ;
		
		\draw    (137.5,130) -- (271,130) ;
		\draw [shift={(273,130)}, rotate = 180] [color={rgb, 255:red, 0; green, 0; blue, 0 }  ][line width=0.75]    (10.93,-3.29) .. controls (6.95,-1.4) and (3.31,-0.3) .. (0,0) .. controls (3.31,0.3) and (6.95,1.4) .. (10.93,3.29)   ;
		
		\draw    (495.5,130) -- (379.5,130) ;
		\draw [shift={(377.5,130)}, rotate = 359.74] [color={rgb, 255:red, 0; green, 0; blue, 0 }  ][line width=0.75]    (10.93,-3.29) .. controls (6.95,-1.4) and (3.31,-0.3) .. (0,0) .. controls (3.31,0.3) and (6.95,1.4) .. (10.93,3.29)   ;

		\draw    (575.5,148) -- (380.5,148) ;
		\draw [shift={(378.5,148)}, rotate = 359.74] [color={rgb, 255:red, 0; green, 0; blue, 0 }  ][line width=0.75]    (10.93,-3.29) .. controls (6.95,-1.4) and (3.31,-0.3) .. (0,0) .. controls (3.31,0.3) and (6.95,1.4) .. (10.93,3.29)   ;
		
		\draw    (275.5,148) -- (1,148) ;
		\draw [shift={(2,148)}, rotate = 0.32] [color={rgb, 255:red, 0; green, 0; blue, 0 }  ][line width=0.75]    (10.93,-3.29) .. controls (6.95,-1.4) and (3.31,-0.3) .. (0,0) .. controls (3.31,0.3) and (6.95,1.4) .. (10.93,3.29)   ;

		\draw (300,130.4) node [anchor=north west][inner sep=0.75pt]    {$\oleFunc$};
		\draw (388,114) node [anchor=north west][inner sep=0.75pt]    {\footnotesize $\alicesRandomPolyEval{x_1} $};
		\draw (388,153.4) node [anchor=north west][inner sep=0.75pt]    {\footnotesize $\bobsRandomPolyEval{x_1} \bobsPolyFromVecEval{x_1} + \finiteFieldPRF(\prfKey, 1)$};
		\draw (184,114) node [anchor=north west][inner sep=0.75pt]    {\footnotesize $\alicesPolyFromVecEval{x_1}$};
		\draw (1,153.4) node [anchor=north west][inner sep=0.75pt]    {\footnotesize $\tPSIOutputPolyEval{x_1} + \finiteFieldPRF(\prfKey, 1)$};

		\draw (316,165) node [anchor=north west][inner sep=0.75pt]   [align=left] {...};
		\draw (316,165) node [anchor=north west][inner sep=0.75pt]   [align=left] {...};

		\draw   (279,186) -- (373.5,186) -- (373.5,210) -- (279,210) -- cycle ;
		
		\draw    (137.5,190) -- (271,190) ;
		\draw [shift={(273,190)}, rotate = 180] [color={rgb, 255:red, 0; green, 0; blue, 0 }  ][line width=0.75]    (10.93,-3.29) .. controls (6.95,-1.4) and (3.31,-0.3) .. (0,0) .. controls (3.31,0.3) and (6.95,1.4) .. (10.93,3.29)   ;
		
		\draw    (495.5,190) -- (379.5,190) ;
		\draw [shift={(377.5,190)}, rotate = 359.74] [color={rgb, 255:red, 0; green, 0; blue, 0 }  ][line width=0.75]    (10.93,-3.29) .. controls (6.95,-1.4) and (3.31,-0.3) .. (0,0) .. controls (3.31,0.3) and (6.95,1.4) .. (10.93,3.29)   ;

		\draw    (575.5,208) -- (380.5,208) ;
		\draw [shift={(378.5,208)}, rotate = 359.74] [color={rgb, 255:red, 0; green, 0; blue, 0 }  ][line width=0.75]    (10.93,-3.29) .. controls (6.95,-1.4) and (3.31,-0.3) .. (0,0) .. controls (3.31,0.3) and (6.95,1.4) .. (10.93,3.29)   ;
		
		\draw    (275.5,208) -- (1,208) ;
		\draw [shift={(2,208)}, rotate = 0.32] [color={rgb, 255:red, 0; green, 0; blue, 0 }  ][line width=0.75]    (10.93,-3.29) .. controls (6.95,-1.4) and (3.31,-0.3) .. (0,0) .. controls (3.31,0.3) and (6.95,1.4) .. (10.93,3.29)   ;

		\draw (300,190.4) node [anchor=north west][inner sep=0.75pt]    {$\oleFunc$};
		\draw (388,174) node [anchor=north west][inner sep=0.75pt]    {\footnotesize $\alicesRandomPolyEval{x_N} $};
		\draw (388,213.4) node [anchor=north west][inner sep=0.75pt]    {\footnotesize $\bobsRandomPolyEval{x_N} \bobsPolyFromVecEval{x_N} + \finiteFieldPRF(\prfKey, N)$};
		\draw (184,174) node [anchor=north west][inner sep=0.75pt]    {\footnotesize $\alicesPolyFromVecEval{x_N}$};
		\draw (1,213.4) node [anchor=north west][inner sep=0.75pt]    {\footnotesize $\tPSIOutputPolyEval{x_N} + \finiteFieldPRF(\prfKey, N)$};

		


	\end{tikzpicture}
	\vspace{-0.4cm}
	\caption{\small Using OLE and \tHamQueryRestricted in \tHamQuery \label{fig:hamQueryGenExample}}
\end{figure}

%

\subsection{\hamAwarePSIProtocol: Hamming DA-PSI from \tHamQuery}
\label{sec:hammingPSI:protocol}
Building a Hamming DA-PSI protocol based on the Hamming query mechanism described so far is straightforward. Let \alicesInputVecSetDefn and \bobsInputVecSetDefn be \alice's and \bob's inputs to the Hamming DA-PSI protocol. Then, for $(i, j) \in \genericSetSize \times \genericSetSize$, \alice and \bob run \tHamQuery with $\alicesVector_i$ and $\bobsVector_j$ 
as inputs. 

We present the full protocol, denoted \hamAwarePSIProtocol, in \appref{app:hamming_psi} with a further optimization using vector OLE (see \secref{sec:background}) to batch the OLEs in \oneSidedSetRecon required across all the instantiations. This optimization 
improves communication costs and compute times without impacting security. 

\begin{restatable}{thm}{securityHamPSI}
	\label{thm:security_ham_psi}
	Assuming that there exists a semantically secure additively homomorphic encryption scheme, and a protocol securely realizing \voleFunc with \bigO{\genericSetSize \secParam} bits of communication, there is a Hamming DA-PSI protocol 
	which securely realizes \hamAwarePSIFunc with \hamPSICommComplexity bits of communication where $\fpr \in (0,0.5)$ is the false positive rate. 
\end{restatable}

\begin{figure}[]
	\begin{mdframed}
		\footnotesize
		
		\noindent
		\textbf{\underline{Parameters}:} \alice and \bob have vectors $\alicesVector, \bobsVector \in \{0,1\}^{\genericVecLen}$, respectively, and a Hamming distance threshold \hamDistThreshold.
		

		\noindent\smallskip
		\textbf{\underline{Procedure $\mathsf{Map}$}:} Follows the steps of \tHamQueryLite.
		

		\smallskip\noindent
		\textbf{\underline{Procedure \oneSidedSetReconExp}:}
		
		\begin{enumerate}[nosep,leftmargin=1.5em,labelwidth=*,align=left,labelsep=0.25em]

			\item \alice and \bob select a set of \numPointsForInterpolPlus points in \finiteField, $\setOfXCoords \assign \setOfXCoordsDefnPlus$ such that none of the points  are in \alicesSetFromVec and \bobsSetFromVec.

			\item \alice encodes \alicesSetFromVec  in the polynomial \alicesPoly = \polyFromSet{\alicesSetFromVec} and \bob encodes \bobsSetFromVec in the polynomial \bobsPoly = \polyFromSet{\bobsSetFromVec}. \bob samples two degree-\genericVecLen random polynomials $\alicesRandomPoly, \bobsRandomPoly \getsr \polyField$. \label{step:set_recon_exp:poly_compute}
			
			\item  For each $x_\ptIdx \in \setOfXCoords$, \alice sends \alicesPolyFromVecEval{x_\ptIdx} and \bob sends \alicesRandomPolyEval{x_\ptIdx} and $\bobsRandomPolyEval{x_\ptIdx}\times \bobsPolyFromVecEval{x_\ptIdx}$ to \oleFunc. \alice learns \alicesVoleOutputPolyEval{x_\ptIdx} \assign \tPSIOutputPolyEval{x_\ptIdx} as the output of \oleFunc. \label{step:set_recon_exp:ole}
		
			\item   \alice computes and interpolates $\setOfPtsForInterpolation_0$ with a rational function of degree $\genericVecLen + \hamDistThreshold$ similar to \lineref{step:set_recon:interpolation}  of \oneSidedSetRecon \label{step:set_recon_exp:interpolate_one}. 
			
			\begin{center}
				$\setOfPtsForInterpolation_0 \assign  \{(x_\ptIdx, y_\ptIdx) : y_\ptIdx \assign \frac{\alicesVoleOutputPolyEval{x_\ptIdx}}{\alicesPolyFromVecEval{x_\ptIdx}}, ~\ptIdx \in [1, \genericVecLen + \hamDistThreshold + 1] \}$.
			\end{center}

			\item If \alice outputs \setDiff{\alicesSetFromVec}{\bobsSetFromVec} from the step above, then output \setDiff{\alicesSetFromVec}{\bobsSetFromVec} and abort. 
			Otherwise, for each $t \in (\hamDistThreshold/2, \hamDistThreshold]$, \alice computes all possible values of \setDiff{\alicesSetFromVec}{\bobsSetFromVec} such that $\setDiffCard{\alicesSetFromVec}{\bobsSetFromVec} = t$. Let $C_\mathsf{sub}$ be the set of all such sets across all $t \in (\hamDistThreshold/2, \hamDistThreshold]$. \label{step:set_recon_exp:set_computation}

			\item For each $C_i \in C_{\mathsf{sub}}$, \alice computes the polynomial $P_i(x) \assign \polyFromSet{C_{i}}$ and computes  the set of points \label{step:set_recon_exp:point_computation}
			
			\begin{center}
				$\setOfPtsForInterpolation_{i} \assign  \{(x_\ptIdx, y_\ptIdx) : y_\ptIdx \assign \frac{\alicesVoleOutputPolyEval{x_\ptIdx}}{P_i(x_\ptIdx)}, ~\ptIdx \in [1, \numPointsForInterpolPlus] \}$.
			\end{center}
			
			\item \alice interpolates $V_i$ for each $C_i \in C_\mathsf{sub}$ with a polynomial. If the degree of the interpolating polynomial is $\leq \genericVecLen + \degree{P_i(x)}$, \alice outputs $\setDiff{\alicesSetFromVec}{\bobsSetFromVec} \assign C_i$. \label{step:set_recon_exp:interpolate_two} 
			
			\end{enumerate}
			\noindent\smallskip
			\textbf{\underline{Procedure $\mathsf{Recover}$}:} Follows the steps of \tHamQueryLite
			

	\end{mdframed}
	\vspace{-15pt}
	\caption{\small \tHamQueryExp: Hamming queries with exponential compute costs. \label{fig:one_sided_set_recon_exp}}
\end{figure}

\subsection{\hamAwarePSIProtocolExp: Sub-Sampling Based Hamming DA-PSI}
\label{sec:hamming_psi:exp_compute}
So far we have discussed a way to build a Hamming DA-PSI protocol using \tHamQuery which fixes \tHamQueryLite by explicitly checking if the inputs, $\hamDist{\alicesVector}{\bobsVector} \in (\hamDistThreshold, 2\hamDistThreshold)$. However, there is an alternate way to fix the problem which leads to a more communication-efficient protocol, but at the cost of additional computation. The protocol uses a Hamming query protocol, denoted \tHamQueryExp, which relies only on OLE.

\myparagraph{\tHamQueryExp} The protocol is based on the findings of \thmref{thm:set_recon_claims}. Specifically, as the proof shows when \alice and \bob interpolate \tPSIOutputPoly at $\setSize{\setOfPtsForInterpolation} > 2\genericVecLen - \degreeOfGCD + 1$ points in \oneSidedSetRecon, \alice can retrieve $\setDiff{\alicesSetFromVec}{\bobsSetFromVec}$ when $\setDiffCard{\alicesSetFromVec}{\bobsSetFromVec} = \genericVecLen - \degreeOfGCD$ with overwhelming probability (Proposition~2). On the other hand, when $\setDiffCard{\alicesSetFromVec}{\bobsSetFromVec} > \genericVecLen - \degreeOfGCD$, the evaluation points reveal nothing to \alice (Proposition~1). 

So, one way to fix \oneSidedSetRecon is by evaluating \tPSIOutputPoly at $(2\genericVecLen - \degreeOfGCD + 1) + 1 = (2\genericVecLen - (\genericVecLen - \hamDistThreshold) + 1) + 1 = \numPointsForInterpolPlus$ points. In this way, \alice learns \setDiff{\alicesSetFromVec}{\bobsSetFromVec} with overwhelming probability when $\setDiffCard{\alicesSetFromVec}{\bobsSetFromVec} \leq \hamDistThreshold$ but nothing otherwise. The modified protocol is called \oneSidedSetReconExp (\figref{fig:one_sided_set_recon_exp}). Similar to \oneSidedSetRecon, \alice and \bob compute polynomials \alicesPolyFromVec and \bobsPolyFromVec from their respective sets. Then, they evaluate \tPSIOutputPoly at \numPointsForInterpolPlus points using calls to \oleFunc. \alice first attempts to interpolate $\genericVecLen + \hamDistThreshold + 1$ points with a rational function of degree $\genericVecLen + \hamDistThreshold$ (\linesref{step:set_recon_exp:poly_compute}{step:set_recon_exp:interpolate_one}). Note if $\setDiffCard{\alicesSetFromVec}{\bobsSetFromVec} \leq \hamDistThreshold/2$, then this step will reveal \setDiff{\alicesSetFromVec}{\bobsSetFromVec} to \alice.

Otherwise, \alice computes for each $t \in (\hamDistThreshold/2, \hamDistThreshold]$, each possible value of $\setDiff{\alicesSetFromVec}{\bobsSetFromVec}$ such that $\setDiffCard{\alicesSetFromVec}{\bobsSetFromVec} = t$ (\lineref{step:set_recon_exp:set_computation}). Let $C_\mathsf{sub}$ be the set of all such sets. Then, for each $C_\setIdx \in C_\mathsf{sub}$, \alice computes $P_i(x) \assign \polyFromSet{C_{i}}$. Finally, \alice checks if a polynomial \numeratorPoly of degree $\leq \genericVecLen + \degree{P_i(x)}$ exists such that $\frac{\numeratorPoly}{P_i(x)}$ is consistent with the points obtained for the rational function $\frac{\tPSIOutputPoly}{P_i(x)}$ (\lineref{step:set_recon_exp:interpolate_two}). Due to Proposition~2, there is a negligible probability of obtaining a false positive, i.e., \alice finds \numeratorPoly when her guess for $\setDiff{\alicesSetFromVec}{\bobsSetFromVec}$ is incorrect. There are no false negatives.



\myparagraph{Reducing Search Space by Sub-Sampling}
The total search space for this process is over $\sum\limits_{t = \hamDistThreshold/2 + 1}^{\hamDistThreshold} \allPossibleComb{\genericVecLen}{t}$ guesses for \setDiff{\alicesSetFromVec}{\bobsSetFromVec}, and is not feasible with large vectors and distance thresholds.  However, as we will show in \secref{sec:eval}, when we replace $\mathsf{Map}$ with a sub-sampling algorithm~\cite{uzun2021cryptographic,uzun2021fuzzy,dodis2004fuzzy,canetti2021reusable} reducing large vectors to a small set of sub-vectors, this method outperforms the existing state of the art \cite{uzun2021fuzzy}. The cost savings come from the fact that our protocol only relies on cheap symmetric-key primitives while the protocol of \citet{uzun2021fuzzy} relies on fully homomorphic encryption. 
Based on this idea, we have built and implemented a DA-PSI protocol combining the sub-sampling algorithm from \cite{uzun2021cryptographic,uzun2021fuzzy,dodis2004fuzzy,canetti2021reusable} with \oneSidedSetReconExp, denoted \hamAwarePSIProtocolExp. More details of the protocol are presented in \appref{app:hamming_query_exp}.

\section{Protocol for Integer Distances}
\label{sec:intergerAwarePSI}

In this section, we present a DA-PSI protocol for $L_1$ distance of order 1
over integers, loosely termed as integer distance-aware PSI. 
The protocol requires \intAwarePSICommComplexity bits of communication for computing the 
intersection of two sets of size \genericSetSize where \intDistThreshold is a user-specified distance threshold. 


 \begin{figure}
	\begin{mdframed}
		\footnotesize
		\noindent
		\textbf{Parameters:~} Parties \alice and \bob.  Integer distance threshold $d \geq 0$. 
		
		\smallskip
		\noindent
		\textbf{Inputs:~} \alice has input $\alicesSet = \{a_i : a_i \in \IntegersPositive \cup \{0\} \}$. \bob has input $\bobsSet = \{b_j : b_j \in \IntegersPositive \cup \{0\} \}$. $\setSize{\alicesSet} = \setSize{\bobsSet} = \genericSetSize$. 
		
		\smallskip
		\noindent
		\textbf{Output:~} \alice and \bob learn $S \subseteq \alicesSet \times \bobsSet$ where if $(\alicesInt, \bobsInt) \in \alicesSet \times \bobsSet, |\alicesInt - \bobsInt| \leq \intDistThreshold$ then $\prob{(\alicesInt, \bobsInt) \in S} \ge \tpr$ and if $|\alicesInt - \bobsInt| > \intDistThreshold$ then $\prob{(\alicesInt, \bobsInt) \not\in S} \ge \tnr$.

	\end{mdframed}
	\vspace{-15pt}
	\caption{\small Ideal functionality \intAwarePSIFunc for Integer DA-PSI. \label{fig:int_aware_PSI_func}}
\end{figure}

%

\myparagraph{Ideal Functionality}
The ideal functionality for an integer distance-aware PSI is defined in \figref{fig:int_aware_PSI_func}. Note that the $(\alicesIthInt, \bobsJthInt) \in S$ only when $\bobsJthInt \in (\alicesIthInt - \intDistThreshold, \alicesIthInt + \intDistThreshold)$. The range excludes the boundary elements, $\alicesIthInt + \intDistThreshold$ and $\alicesIthInt - \intDistThreshold$. This is primarily for ease of description of the protocol and it is trivial to extend the functionality and the protocol to include the boundary elements. Also, while the functionality allows tunable true positive and true negative rates, \textit{the protocol we present is correct with probability $1.0$, i.e., $\tpr = \tnr = 1$}. 

Observe that an inefficient realization of \intAwarePSIFunc  immediately exists: \alice creates an augmented set, \augAlicesSet with all integers $(a - \intDistThreshold, a + \intDistThreshold)$ for each \elemInSet{\alicesInt}{\alicesSet}. Any generic PSI protocol be used for computing the intersection between the augmented set and \bobsSet.  This protocol however requires  \bigO{\genericSetSize \secParam \intDistThreshold} bits of communication, using a PSI protocol with communication cost scaling linearly in the set size.

%
%
%
%
%
%

\myparagraph{Key Idea}
To reduce overall communication, we will reduce the number of items in the augmented set. The key observation behind this reduction is that all integers in the neighborhood of an integer 
\elemInSet{\alicesInt}{\alicesSet}, \integersInNeighborhood  can be succinctly represented by a collection of bit strings corresponding to their binary representations. The total number of 
such strings required is sublinear in \intDistThreshold since multiple integers within a sequence will share prefixes, and the same common prefix can be used to represent multiple consecutive integers. For instance, the binary representation of 42 ($101010$) and 43 ($101011$) share the prefix $10101$. Both these integers can be represented by the string $10101\ast$ where $\ast$ denotes a wildcard bit. Leveraging this fact, the idea is to generate the least number of bit strings to represent all integers \integersInNeighborhood. The problem is reduced to string matching over these bit strings.

The augmenting process is discussed next. The protocol we will present 
allows one of the parties, say \alice, to learn the integer-aware intersection, and then this information can be shared with \bob using an extra round of communication. 

\myparagraph{Augmenting \alice's Set}
The augmented set \augAlicesSet includes fixed-length strings representing \integersInRange for each \elemInSet{\alicesInt}{\alicesSet}. These strings are obtained 
from the prefixes of fixed-length binary representations of all integers in the range. This fixed length, denoted \maxBitLength, may be determined from the universe from which the elements in \alicesSet and \bobsSet are drawn. 

Intuitively, the process is based on two observations. First, the integers in \integersInRange 
can differ only in their $\flooredLog{(2\intDistThreshold - 1)} + 1$ least significant bits if $2^k \leq a - d \leq a + d \leq 2^{k+1}$.  Second, the \maxBitLength-sized binary representations of all integers in $[2^{k}, 2^{k+1} - 1], k \in \IntegersPositive$ have a common prefix of length $\maxBitLength - k -1$.  So, all these integers can be represented by a string formed by 
appending $k + 1$ wildcard bits to the common prefix. Based on these observations, the idea is to recursively partition \integersInRange into smaller ranges of the form $[0, 2^{k} - 1]$ or $[2^{k}, 2^{k+1} - 1]$ and obtain a \textit{representative string} for each such range. 
More formally, to create representative strings for integers in \integersInRange, we identify the \textit{enclosing common} prefixes.

\begin{defn}
	Given any arbitrary set of bit strings, an enclosing common prefix of length \genericPrefLen $\leq$ \maxBitLength satisfies: 
	\begin{enumerate}[nosep,leftmargin=1.6em,labelwidth=*,align=left]
	
	\item There are $2^{(\maxBitLength - \genericPrefLen)}$ bit strings which have this prefix in common. 
	\item The bit strings which share this prefix do not have a common prefix of length $< \genericPrefLen$. 
	\end{enumerate}
\end{defn}

All identified enclosing common prefixes are appended with wildcard bits 
to generate representative strings.  We show later that for each \elemInSet{\alicesInt}{\alicesSet} the number of enclosing common prefixes 
is \bigO{\log \intDistThreshold}. Intuitively this is because the range of integers is recursively halved and each such range has a constant number of 
enclosing common prefixes. Thus, the augmented set contains \bigO{\genericSetSize \cdot \log \intDistThreshold} representative strings.


\smallskip\noindent\textbf{Example:} We are interested in the representative strings for all integers in the range (41, 56) (see \figref{fig:prefixTrie42}). The 8-bit binary representation of 42 is $00101010$ and the 8-bit binary representation of 55 is $00110111$. Integers in the range [42, 43] have a common enclosing prefix $0010101$, 
integers in the range [44, 47] have a common enclosing prefix $001011$ and the integers in the range [48, 55] have a common enclosing prefix $00110$. Thus, 8-bit representative strings for all integers in $(41,56)$ are $0010101$$\ast$, 
$001011$$\ast$$\ast$ and $00110$$\ast$$\ast$$\ast$.

\myparagraph{Augmenting \bob's Set}
To check whether an integer falls in any of the ranges in an augmented set, we need to check 
whether it shares a common prefix with any of the representative strings. The augmented set \augBobsSet
includes these strings. Specifically we need to check prefixes only of length $\genericPrefLen \in [\maxBitLength - \flooredLog{(2d -1)} -1, \maxBitLength]$ (as will be discussed later). Thus, representative strings for each integer \elemInSet{\bobsJthInt}{\bobsSet} is obtained by replacing the required number of least significant bits in the binary representation of \bobsJthInt with wildcard bits. Specifically, the first representative string is generated by replacing the least significant bit, the second string is generated by replacing the last two least significant bits and so on. All the representative strings for all \elemInSet{\bobsJthInt}{\bobsSet} is part of \augBobsSet. Thus, \setSize{\augBobsSet} = \bigO{\genericSetSize \cdot \log \intDistThreshold}.

\smallskip\noindent\textbf{Example:} Consider $b = 49$ ($00110001$). The 8-bit representative strings of 49 are $00110001$, $0011000\ast$, $001100\ast\ast$, $00110\ast\ast\ast$, $0011\ast\ast\ast\ast$. The set of  strings for integers in range $(41,56)$ and the strings for $49$ have the string $00110\ast\ast\ast$ in common which correctly shows that $b \in (41, 56)$.

%


\begin{figure}[t]
\centering
\begin{tikzpicture}[every tree node/.style={draw,circle},sibling
	distance=10pt, level distance=40pt, scale = 0.4]
	\tikzset{edge from parent/.style={draw, edge from parent path=
			{(\tikzparentnode) -- (\tikzchildnode)}}}
	\Tree [.R [.00101 [.0 [.\node [label={[black]right:A}, black] {1}; [ .\node [label={below:42}]{0}; ] [ .\node [label={below:43}]{1}; ]]]  [.\node [label={[black]right:B}, black] {1}; [.0 [ .\node [label={below:44}]{0};  ] [ .\node [label={below:45}]{1}; ] ]  [.1 [ .\node [label={below:46}]{0}; ] [ .\node [label={below:47}]{1}; ] ]]]
	[.\node [label={[black]right:C}, black]{00110}; [.0 [.0 [ .\node [label={below:48}]{0}; ] [.\node [label={below:49}]{1}; ]] [.1 [ .\node [label={below:50}]{0}; ] [ .\node [label={below:51}]{1}; ]]]  [.1 [.0 [ .\node [label={below:52}]{0};  ] [ .\node [label={below:53}]{1}; ] ]  [.1 [ .\node [label={below:54}]{0};  ] [  .\node [label={below:55}]{1}; ] ]]]]
\end{tikzpicture}
\vspace{-10pt}
\caption{\small Prefix trie examples. (a) Trie built over the binary representations of integers [42, 55]. The nodes marked A, B and C are the roots of the \mec subtries.\label{fig:prefixTrie42}} 
\end{figure}
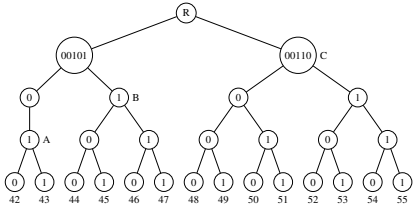

\begin{thm}
	Assuming that there is a secure scheme for computing a private set intersection over sets of size \genericSetSize with \bigO{\genericSetSize \secParam} bits of communication. Then, there is a secure protocol realizing \intAwarePSIFunc with \bigO{\genericSetSize \secParam \log \intDistThreshold} bits of communication where \intDistThreshold is the specified distance threshold. 
\end{thm}

A private set intersection protocol with the augmented sets as inputs provides a integer-distance aware intersection. This is because a non-null intersection implies that a representative string(s) for \elemInSet{\bobsJthInt}{\bobsSet} matched a representative string(s) in \elemInSet{\alicesIthInt}{\alicesSet} which can only happen if \intDistanceLessThanThreshold{\alicesIthInt}{\bobsJthInt}. Any existing private set intersection protocol can be used.  In Section \ref{sec:eval}, we have instantiated an integer DA-PSI protocol with an OT-based PSI protocol due to \citet{pinkas2018PSIOT}. We will omit details of this straightforward integration.  To estimate the communication complexity of the protocol, in \appref{app:intPSIProofs}, we describe an algorithm to augment the input sets using the specified distance threshold. The algorithm builds a prefix trie over the binary representation of \integersInRange, and identifies the \mec subtries  (see \figref{fig:prefixTrie42}).


\begin{defn}
	\label{defn:mec}
	A subtrie in a prefix trie is called a maximal enclosing complete subtrie if it satisfies the following properties: i) it is a complete binary tree, and ii) it is not part of any other complete binary subtree(s) rooted at one of its ancestors. 
\end{defn}

Each \mec subtrie corresponds to an enclosing common prefix of the bit strings for \integersInRange. We count the number of \mec subtries in the prefix trie built over the binary representations and show that the total number of maximal enclosing complete subtries in a prefix trie built over the binary representations of all integers \integersInRange is $\bigO{\log \intDistThreshold}$.


%

\section{Evaluation}
\label{sec:eval}


\begin{figure*}[h!]
	\begin{minipage}{0.5\linewidth}
	\subfloat[Total comm.\ vs.\ vector length]
	{ \label{fig:hamPSI:micro:comm_length}
		\begin{tikzpicture}[scale=0.5]
			\begin{axis}[
				xlabel=Vector Length ($\genericVecLen$),
				ylabel=Total Comm. (in MB. logscale),
				ymax=6
				]

				
				\addplot+[line width=2.5pt, white!80!black, solid] coordinates {(128, 2) (256, 2.3) (512, 2.6) (1024, 2.9) (2048, 3.2) (4196, 3.5) (8192, 3.8) };
				
				\addplot+[line width=2.5pt, white!10!black, solid] coordinates {(128, 2.5) (256, 2.5) (512, 2.5) (1024, 2.5) (2048, 2.5) (4196, 2.5) (8192, 2.5)};
	        	\addplot+[line width=2.5pt, white!30!black, solid] coordinates {(128, 2.8) (256, 2.8) (512, 2.8) (1024, 2.8) (2048, 2.8) (4196, 2.8) (8192, 2.8) };
	        	
				\addplot+[line width=2.5pt, white!50!black, solid] coordinates {(128, 3.5) (256, 3.5) (512, 3.5) (1024, 3.5) (2048, 3.5) (4196, 3.5) (8192, 3.5) };

				\legend{GC Baseline, $\hamAwarePSIProtocol - 0.1$, $\hamAwarePSIProtocol - 0.05$, $\hamAwarePSIProtocol - 0.01$}

			\end{axis}
		\end{tikzpicture}
 	}
	\subfloat[Total comm.\ vs.\ threshold]
	{ \label{fig:hamPSI:micro:comm_threshold}
	\begin{tikzpicture}[scale=0.5]
		\begin{axis}[
			xlabel=Distance Threshold ($\hamDistThreshold$),
			ylabel=Total Comm. (in MB. logscale),
			ymax=8
			]

			\addplot+[line width=2.5pt, white!50!black, solid] coordinates {(1, 3.8) (2, 3.8) (4, 3.8) (6, 3.8) (8, 3.8) (10, 3.8) (12, 3.8) (14, 3.8) (16, 3.8) (18, 3.8) (20, 3.8) (22, 3.8) (24, 3.8) (26, 3.8) (28, 3.8) (30, 3.8)};
			
			\addplot+[line width=2.5pt, white!10!black, solid] coordinates {(1, 0.93) (2, 1.3) (4, 1.8) (6, 2.1) (8, 2.4) (10, 2.5) (12, 2.7) (14, 2.8) (16, 2.9) (18, 3) (20, 3.1) (22, 3.2) (24, 3.30) (26, 3.37) (28, 3.43) (30, 3.49)};
			
			\addplot+[line width=2.5pt, white!30!black, solid] coordinates {(1, 1.07) (2, 1.56) (4, 2.10) (6, 2.4) (8, 2.6) (10, 2.8) (12, 3.0) (14, 3.1) (16, 3.2) (18, 3.3) (20, 3.4) (22, 3.52) (24, 3.59) (26, 3.66) (28, 3.72) (30, 3.78)};
			
			\addplot+[line width=2.5pt, white!50!black, solid] coordinates {(1, 1.5) (2, 2.1) (4, 2.7) (6, 3.0) (8,3.3) (10, 3.5) (12, 3.6) (14, 3.8) (16, 3.9) (18, 4.0) (20,4.1) (22, 4.2) (24, 4.28) (26, 4.35) (28, 4.41) (30, 4.47) (32, 4.53) };
			
		

			
			\legend{GC Baseline, $\hamAwarePSIProtocol - 0.1$, 
			$\hamAwarePSIProtocol - 0.05$, 
			$\hamAwarePSIProtocol - 0.01$}
			
		\end{axis}
	\end{tikzpicture}
	}

\end{minipage}
\begin{minipage}{0.4\linewidth}
\subfloat[Compute time vs.\ vector length]
{\label{fig:hamPSI:micro:comp_size}
	\begin{tikzpicture}[scale=0.5]
\begin{axis}[
ybar,
enlargelimits=0.15,
legend pos = north west,
xlabel= Vector Length ($\genericVecLen$),
ylabel=Time (in sec),
symbolic x coords={256,512, 1048, 2048, 4096, 8192},
nodes near coords,
nodes near coords align={vertical},
cycle list = {black,black!70,black!40,black!10}
]

\addplot+[fill,text=black] coordinates {(256,267) (512,  267) (1048, 267) (2048, 267) (4096, 267) (8192, 267)};
\addplot+[] coordinates {(256,98)(512,221 ) (1048,510) (2048, 977) (4096, 1450) (8192,3000)};

\legend{$\hamAwarePSIProtocol - 0.05$, GC Baseline}
\end{axis}
\end{tikzpicture}}
\subfloat[Compute time vs.\ threshold]
{\label{fig:hamPSI:micro:comp_threshold}
	\begin{tikzpicture}[scale=0.5]
\begin{axis}[
ybar,
enlargelimits=0.15,
legend pos = north west,
xlabel=Distance Threshold ($\hamDistThreshold$),
ylabel=Time (in sec),
ymax=4000,
symbolic x coords={1,2, 4, 8, 16, 32},
nodes near coords,
nodes near coords align={vertical},
cycle list = {black,black!70,black!40,black!10}
]
\addplot+[] coordinates {(1,8)(2,20 ) (4,57) (8, 487) (16,622) (32, 2296)};
\addplot+[fill,text=black] coordinates {(1,3000) (2, 3000) (4, 3000) (8, 3000) (16, 3000) (32, 3000)};
\legend{$\hamAwarePSIProtocol - 0.05$, GC Baseline}
\end{axis}
\end{tikzpicture}}
\end{minipage}
\vspace{-10pt}
\caption{\small Micro-benchmarks for \hamAwarePSIProtocol (\secref{sec:hammingPSI:protocol}). Set size $\genericSetSize = 100$. In \figref{fig:hamPSI:micro:comm_length} and \figref{fig:hamPSI:micro:comm_threshold}, the values on the $y$-axis are in logscale (base 10). \figref{fig:hamPSI:micro:comm_length} shows for \hamDistThreshold = 10, and for FPR = 0.05 and FPR = 0.01, \hamAwarePSIProtocol has around $\mathbf{10\times}$ and $\mathbf{2\times}$ lower communication volume respectively, compared to the GC baseline when the vector dimensions, $\genericVecLen = 8192$. \figref{fig:hamPSI:micro:comm_threshold} shows with $\genericVecLen = 8192$ dimension vector, and for FPR = 0.05, \hamAwarePSIProtocol has $\mathbf{2 - 537\times}$ lower communication volumes up to distance threshold $\hamDistThreshold = 32$ compared to the baseline. \figref{fig:hamPSI:micro:comp_size} shows that for vector lengths, $\genericVecLen \geq 1024$ bits, \hamAwarePSIProtocol is at least $\mathbf{2}\times$ faster than the GC baseline. \figref{fig:hamPSI:micro:comp_threshold} shows that for up to distance threshold $\hamDistThreshold \leq 32$, \hamAwarePSIProtocol is faster than the GC baseline. }
\end{figure*}

We have implemented both \hamAwarePSIProtocol (\secref{sec:hammingPSI:protocol}) and \intAwarePSIProtocol (\secref{sec:intergerAwarePSI}). In the following sections, we benchmark the protocols. The evaluation metrics are communication costs and compute times.  All experiments consist of 5 independent trials and results are collected with a 95\% confidence interval.


\myparagraph{Platform}
We ran our experiments on two different platforms representing low and high resource environments respectively. 

\begin{itemize}[nosep,leftmargin=1em,labelwidth=*,align=left]

	\item \textbf{Low-resource:} Unless stated otherwise, our experiments were run on two t2.xlarge Amazon EC2 instances with 4 vCPUs and 16GB of RAM. To simulate realistic scenarios, these instances were placed in different zones (US East and West). The network bandwidth between them was measured to be around 40--60\megabytes per second using {\em iperf}\footnote{\url{https://iperf.fr/}}.
	
	\item \textbf{High-resource:} We used a Microsoft Azure F72s\_v2 instance, which has 72 virtual cores 
          and 144\gigabytes of RAM, to compare to the results of \citet[Table~10]{uzun2021fuzzy}.

\end{itemize}

\subsection{Hamming Distance Protocol}

\myparagraph{Implementation}
We have implemented \hamAwarePSIProtocol and \hamAwarePSIProtocolExp in C++11. The implementation uses the NTL library\footnote{\url{https://libntl.org/}} for implementing the finite field arithmetic, and the operations are performed over a 128-bit prime order field. We have used the open-source implementation\footnote{\url{https://github.com/emp-toolkit/emp-zk}} of the state of the art VOLE scheme \cite{weng2021vole} for our set reconciliation protocols. 

Finally for \tHamQueryRestricted, we have used the open source implementation\footnote{\url{https://github.com/lubux/ecelgamal}} of the EC-ElGamal encryption scheme on a 256-bit curve as 
the additively homomorphic encryption. The scheme is set up with a 24-bit message space and a precomputed plaintext table of 24-bit messages to speedup the decryption process. The message space is large enough to encrypt the individual bits of the vectors, compute the Hamming distance between them over the ciphertexts and compare with the distance threshold. The 128-bit key returned in the protocol is split into 24-bit chunks to fit into the message space. More details are in \appref{app:el_gamal}.

\subsubsection{Comparison with Generic 2PC \cite{huang2011garbled}}
\label{sec:eval:hamming:twoPC}
We have compared our Hamming DA-PSI protocol (\secref{sec:hammingPSI:protocol}) denoted \hamAwarePSIProtocol with the garbled circuit based construction by \citet{huang2006hamming}. This construction is more efficient than the AHE-based scheme due to \citet{osadchy2010scifi}. We use an open source implementation\footnote{\url{https://mightbeevil.org}}. 


\myparagraph{Micro-Benchmark}
We run micro-benchmarks with sets containing 100 vectors sampled from the space $\{0,1\}^{\genericVecLen}$, and measure the communication volumes relative to the baseline.

\begin{itemize}[nosep,leftmargin=1em,labelwidth=*,align=left]

	\item \textbf{Communication volume:} \figref{fig:hamPSI:micro:comm_length} shows how the communication volume scales with the the vector lengths, \genericVecLen. The distance threshold $\hamDistThreshold = 10$. As expected, the communication volume for \hamAwarePSIProtocol remains constant, while the communication volume for the GC based solution scales linearly in the vector size.   For vectors of length greater than $512$ bits, \hamAwarePSIProtocol outperforms the GC based solution. With 8192-bit vectors, \hamAwarePSIProtocol has 10$\times$ lower communication volume for FPR = 0.05.  \figref{fig:hamPSI:micro:comm_threshold} shows how communication volume scales with the distance threshold, \hamDistThreshold. \hamAwarePSIProtocol has $1.5 - 440\times$ lower communication volumes up to $\hamDistThreshold = 30$ compared to the baseline.

	\item \textbf{Compute time:} \figref{fig:hamPSI:micro:comp_size} and \figref{fig:hamPSI:micro:comp_threshold} shows how the compute time in \hamAwarePSIProtocol scales with vector lengths, and the distance threshold for FPR = 0.05. For vectors of length $\geq 1024$ bits, \hamAwarePSIProtocol is at least $2 \times$ faster than the GC-based system. With $8192$-bit vectors, \hamAwarePSIProtocol is faster than the GC-based system for distance threshold $\hamDistThreshold \leq 32$.
	
\end{itemize}


\myparagraph{Application Benchmark}
We use Iris recognition as an application of \hamAwarePSIProtocol in a setting where the input vectors are long, while the distance threshold is small. In this setting, \alice and \bob have sets comprising 100 images of irises and want to learn if they have common elements. The iris data is collected from the CASIA dataset\footnote{\url{http://www.cbsr.ia.ac.cn/english/IrisDatabase.asp}}. An open source tool is used to extract features from the dataset and compute 6000-bit long binary vectors corresponding to the items\footnote{\url{https://github.com/mvjq/IrisRecognition}}.

%

%
To privately compare a single pair of vectors, the garbled circuit based solution requires around 300KB
of communication. With sets of size 100, the total communication required is around 3GB. The communication cost 
mainly depends on the vector length and is not affected by the distance threshold.  For \hamAwarePSIProtocol, with a threshold $\hamDistThreshold = 20$, we are able to retrieve all the matches with communication cost 2.5$\times$ and $1.3\times$ lower than the garbled circuit baseline with $\tnr = 0.9$ and $\tnr = 0.95$ respectively.


\subsubsection{Comparison with \citet{uzun2021fuzzy}}
\label{sec:eval:hamming:uzun}
We have compared with the Hamming query protocol by \citet{uzun2021fuzzy}. Their protocol has two components: an application-specific sub-sampling procedure that reduces bio-identifiers (e.g., bit vectors derived from facial features) to sets of high-dimensional items, and a $t$-out-of-$T$ matching protocol.
The purpose of our comparison is to show that our set reconciliation protocol is more efficient than the FHE-based $t$-out-of-$T$ protocol of \citet{uzun2021fuzzy}. In this way, the comparison is independent of the sub-sampling procedure, which can change based on the application.

Unfortunately, since we are unable to obtain their code\footnote{The authors declined to provide the code for their implementation and instead recommended that we run our experiments on the same platform and compare our results to those reported in their paper.}, and compute results of their $t$-out-of-$T$ matching protocol in isolation, for a fair comparison we use their sub-sampling procedure on top of \oneSidedSetRecon and compare the overall times. This poses a challenge since the sub-sampling procedure outputs sets of sub-vectors, we cannot directly apply \hamAwarePSIProtocol which takes bit vectors as inputs. To overcome this problem we have compared their protocol with \hamAwarePSIProtocolExp (\secref{sec:hamming_psi:exp_compute}, \appref{app:hamming_query_exp}) with their sub-sampling procedure replacing the mapping procedure. We stress that resorting to \hamAwarePSIProtocolExp (vs.\ \hamAwarePSIProtocol) is simply to enable comparison to \citet{uzun2021fuzzy} without their code.
To implement a client-server containment query (as in \cite{uzun2021fuzzy}), \alice's input (client) is a singleton set $\{\alicesVector\}$ while \bob's input is \bobsInputVecSetDefn. \alice's compute cost is 
$\bigO{\binom{T}{t}}$ where $T, ~t$ are parameters of the sub-sampling procedure. When $T = 64$ and $t = 2$ \cite{uzun2021fuzzy}, this cost is practically feasible.


\pgfplotscreateplotcyclelist{grayscale}{
	very thick,white!10!black,mark=o,mark options={scale=1.5},solid,densely dotted\\%
	very thick,white!20!black,mark=x,mark options={scale=1.5}, solid,densely dotted\\%
	very thick,white!30!black,mark=*,mark options={scale=1.5}, solid,densely dotted\\%
	very thick,white!40!black,mark=+,mark options={scale=1.5}, solid,dashed\\%
	very thick,white!50!black,mark=x,mark options={scale=1.5}, solid,dashed\\%
}

\begin{figure*}[ht!]
\begin{minipage}{0.5\linewidth}
	\subfloat[Comm vs. threshold.]
 { \label{fig:intPSI:micro:comm_threshold}
 	\begin{tikzpicture}[scale=0.5]
 	\begin{axis}[
 		cycle list name=grayscale,
 	legend pos = north west,
 	xlabel=Distance Threshold,
 	ylabel=Total Comm. (in MB.)
 	]
 	
 	\addplot coordinates {(0,0.1) (20,0.8) (40,1) (60,1.1) (80,1.2) (100,1.4)};
 	\addplot coordinates {(0,0.1) (20,3.9) (40,7.9) (60,11.4) (80,15.1) (100,18.6)};
 	\addplot coordinates {(0,1.1) (20,8) (40,11) (60,12) (80,13) (100,14)};
	\addplot coordinates {(0,1.1) (20,40) (40,80) (60,115) (80,153) (100,188)};
 	
 	\legend{\intAwarePSIProtocol 1k, DA Baseline 1k, \intAwarePSIProtocol 10k, DA Baseline 10k}
 	
 	\end{axis}
 	\end{tikzpicture}
 }
\subfloat[Compute time vs. threshold]
{\label{fig:intPSI:micro:comp_threshold}
	\begin{tikzpicture}[scale=0.5]
\begin{axis}[
ybar,
enlargelimits=0.15,
legend pos = north west,
xlabel=Distance Threshold,
ylabel=Time (in sec),
symbolic x coords={0,20, 40, 60, 80, 100},
nodes near coords,
nodes near coords align={vertical},
cycle list = {black,black!70,black!40,black!10}
]
\addplot+[] coordinates {(0,0.1)(20, 0.2) (40,0.2) (60, 0.2) (80, 0.2) (100,0.2)};
\addplot+[fill,text=black] coordinates {(0,0.1) (20,  0.3) (40, 0.4) (60, 0.4) (80, 0.5) (100,0.6)};
\legend{\intAwarePSIProtocol, DA Baseline}
\end{axis}
\end{tikzpicture}}
\end{minipage}
\begin{minipage}{0.5\linewidth}
	\subfloat[Comm vs. threshold]
	{ \label{fig:intPSI:macro:comm_threshold}
		\begin{tikzpicture}[scale=0.5]
			\begin{axis}[
				cycle list name=grayscale,
				legend pos = north west,
				xlabel=Distance Threshold,
				ylabel=Total Comm. (in MB.)
				]
				\addplot coordinates {(0,1.9) (2,3.6) (4,7.2) (8,13.2) (16,24.7) (32,30) (64,39) (128,46) };
				\addplot coordinates {(0,1.9) (2,9.7) (4,17.4) (8,32.4) (16,62.4) (32,130) (64,250) (128,434)};
				\legend{\intAwarePSIProtocol, DA Baseline}
				
			\end{axis}
		\end{tikzpicture}
	}
	\subfloat[Compute time vs. threshold]
	{\label{fig:intPSI:macro:comp_threshold}
		\begin{tikzpicture}[scale=0.5]
			\begin{axis}[
				ybar = 0pt,
				enlargelimits=0.15,
				bar width = 0.2cm,
				xlabel=Distance Threshold,
				ylabel={Time (in ms, log scale)},
				legend pos = north west,
				symbolic x coords={0, 2, 4, 8, 16, 32, 64, 128},
				cycle list = {black,black!70,black!40,black!10}
				]
				\addplot+[] coordinates {(0, 2.4 ) (2, 2.5) (4,2.6) (8, 2.7) (16, 2.9) (32,3) (64,3.11) (128,3.17)};
				\addplot+[fill,text=black] coordinates {(0, 2.4) (2, 2.6) (4,2.7) (8,2.9) (16,3.2) (32,3.5) (64,3.8) (128,4.1)};
				\legend{\intAwarePSIProtocol, DA Baseline}
			\end{axis}
	\end{tikzpicture}}
	\end{minipage}
	\vspace{-15pt}
\caption{\small Benchmarks for \intAwarePSIProtocol. \figref{fig:intPSI:micro:comm_threshold} and \figref{fig:intPSI:micro:comp_threshold} are microbenchmarks with  set sizes = 1k, 10k. The communication volume for \intAwarePSIProtocol is roughly $\mathbf{13\times}$ less than the baseline for threshold = 100. Compute time is $\mathbf{3\times}$ less due to smaller augmented set sizes. 
\figref{fig:intPSI:macro:comm_threshold} and \figref{fig:intPSI:macro:comp_threshold} are benchmaks when \intAwarePSIProtocol is applied to the task of private collaborative blacklisting. Set size = 25k. The communication volume and compute time for \intAwarePSIProtocol are both around $\mathbf{10\times}$ less than the baseline when threshold = 128.}
\end{figure*}

%

%
%
%
%
%
%
%
%
%
%


\myparagraph{Dataset}
The dataset used for bechmarking by \citet{uzun2021fuzzy} is a set of synthetically generated (by a generative network) images of human faces. Since biometric authentication/face recognition is not our focus, we opted to use a dataset with randomly generated bit vectors matching the parameters in \cite{uzun2021cryptographic}. Specifically, we generated sets of varying sizes containing bit-vectors of length \genericVecLen = 256, and then sub-sampled these bit vectors using the algorithm used by \citet{uzun2021fuzzy} to generate corresponding sets of size $T = 64$.


\myparagraph{Results for High-Resource Setup}
 As recommended by the authors of \cite{uzun2021fuzzy}, we have used our high resource setup when comparing our protocol with their protocol. \tblref{tab:comparison_fuzzy_psi} reports results of the comparison. The numbers for the protocol of \citet{uzun2021fuzzy} correspond to a setup without load-balancing the dataset on the server, and so our results report the worst-case performance for both protocols. Load balancing can be applied to our protocol to improve performance; however, since it is unlikely to show improvements with randomly generated vectors, we omit this optimization. 

\begin{table}[th!]
\centering
{\small
\begin{tabular}{@{}l@{\hspace{0.65em}}c@{\hspace{0.35em}}c@{}c@{\hspace{0.65em}}c@{\hspace{0.35em}}c@{}c@{\hspace{0.65em}}c@{\hspace{0.35em}}c@{}}
\toprule
Set size \genericSetSize & \multicolumn{2}{c@{}}{10\kilo}                    && \multicolumn{2}{c@{}}{100\kilo} && \multicolumn{2}{c@{}}{1\mega} \\
Measure & Comm & Comp && Comm & Comp && Comm & Comp \\
\cmidrule{2-3} \cmidrule{5-6} \cmidrule{8-9}
\hamAwarePSIProtocolExp & 25.8MB & 4.00s  && 220MB & 11.0s && 1504MB & 130s \\
\citet{uzun2021fuzzy} & 72.0MB & 2.12s && 528MB & 17.8s && 2124MB & 189s \\
\bottomrule
\end{tabular}}
\vspace{-8pt}
\caption{Comparison of communication and computation costs of \hamAwarePSIProtocolExp (\appref{app:hamming_query_exp}) with \citet{uzun2021fuzzy}. \label{tab:comparison_fuzzy_psi}}
\end{table}

Our protocol outperforms the protocol of \citet{uzun2021fuzzy} for sets containing up to 1 million elements. The improvement is the result of using a communication-efficient VOLE protocol over fully homomorphic encryption. 
For instance, the amortized communication cost of performing a single oblivious linear evaluation as part of the VOLE protocol, i.e., computing a single value $\vec{z}[\vecBitIdx] \assign \vec{u}[\vecBitIdx] x + \vec{v}[\vecBitIdx] \in \finiteField$ (see the definition of \voleFunc) for some $\vecBitIdx \in [1, \genericSetSize]$,  is roughly 3 bits \cite{weng2021vole}. When applied to the set containing 10\kilo elements, we require 1.2 million correlations. The total communication cost is 3.6 million bits, or 450\kilobytes. In addition, we need to transfer a field element in plaintext to convert each pseudorandom VOLE correlation (as provided by \citet{weng2021vole}), to a correlation with the desired parameters. The total cost is $2.8\times$ less than the cost of the protocol of \citet{uzun2021fuzzy}.


Compute costs are also lower due to the use of cheaper symmetric-key primitives in our protocol over fully homomorphic encryption. The only exception is for the set containing 10\kilo elements since the VOLE protocol requires a setup time. For small sets, we still incur the setup time while not fully utilizing all the usable oblivious linear evaluations.

\begin{table}[th!]
\centering
{\small
\begin{tabular}{@{}l@{\hspace{0.65em}}c@{\hspace{0.35em}}c@{}c@{\hspace{0.65em}}c@{\hspace{0.35em}}c@{}c@{\hspace{0.65em}}c@{\hspace{0.35em}}c@{}}
\toprule
Set size \genericSetSize & \multicolumn{2}{c@{}}{10\kilo}                    && \multicolumn{2}{c@{}}{100\kilo} && \multicolumn{2}{c@{}}{1\mega} \\
Party & \alice & \bob && \alice & \bob && \alice & \bob \\
\cmidrule{2-3} \cmidrule{5-6} \cmidrule{8-9}
\hamAwarePSIProtocolExp & 11.4s & 1.20s && 24.3s & 2.80s && 145.3s & 3.50s \\
\bottomrule
\end{tabular}}
\vspace{-8pt}
\caption{Compute costs of \hamAwarePSIProtocolExp (\appref{app:hamming_query_exp}) on an Amazon t2.xlarge instance. \label{tab:low-resource}}
\end{table}

\myparagraph{Results for Low-Resource Setup}
The high-resource setup used by \citet{uzun2021fuzzy} is necessary for the compute-intensive tasks in their FHE-based scheme. In fact, Table 8 of \citet{uzun2021fuzzy} shows that deploying the system with 72 threads utilizing all the available vCPUs leads to a $32.4 \times$ speed up compared to a single-threaded deployment. Unlike their setting, PSI settings are often symmetrically provisioned and have more modest configurations.  We demonstrate feasibility of our protocol even on low-resource platforms. And while we are unable to compute the actual costs of running the protocol of \citet{uzun2021fuzzy} on the same platform, we posit that their compute costs would be significantly higher on low-resource systems due to the inherent cost of FHE. 

\tblref{tab:low-resource} shows the compute times of our protocol on the low-resource platform described before. The communication costs remain the same as the ones presented in \tblref{tab:comparison_fuzzy_psi} and therefore we omit the results. \bob only participates in the VOLE protocol and therefore has no other compute costs.
As is evident, the cost of our protocol on a low resource environment is similar to the performance on the over-provisioned system. The use of a cheap symmetric key primitive, namely VOLE, ensures that \bob's online compute times are low (less than 4s for databases containing up to 1\mega elements). The majority of \alice's time is spent on local computation which can be further optimized by leveraging parallel processing.

\subsection{Integer Distance Protocol}

\myparagraph{Implementation}
The integer distance-aware protocol (or \intAwarePSIProtocol in short) is implemented as a two step process in C++.
First both parties augment their sets using Algorithm \ref{alg:set_augmentation} (in \appref{app:intPSIProofs}). These augmented sets are used as inputs to the OT-based PSI protocol due to \citet{pinkas2018PSIOT}. We rely on an open-source implementation\footnote{\url{https://github.com/encryptogroup/PSI}}.
\textit{We note that \intAwarePSIProtocol can be instantiated with any traditional PSI protocol of choice, and both communication volume and compute times are expected to show similar trends.}


\myparagraph{Micro-benchmarks}
We run micro-benchmarks with sets containing 1k and 10k elements each, randomly sampled from the space of non-negative (32 bit) integers. We evaluate how the communication and overall compute time of the protocol scales with the distance threshold and set size. The baseline is the protocol due to \citet{pinkas2018PSIOT} where we augment the input sets with items in the neighborhood of each item in the set based on the threshold. Here, the augmented set size is expected to scale linearly with the distance threshold. This is called the ``DA Baseline" in our experiments.

\begin{itemize}[nosep,leftmargin=1em,labelwidth=*,align=left]
	\item {\bf Comm. volume vs.\ threshold:} \figref{fig:intPSI:micro:comm_threshold} shows how the communication volume scales with the distance threshold. The communication costs of \intAwarePSIProtocol scale logarithmically, and so with distance threshold $100$, the communication volume is $12.5\times$ less than the baseline for set size 10\kilo.
	
	\item {\bf Compute time vs.\ threshold:} \figref{fig:intPSI:micro:comp_threshold} shows how compute time scales with the distance threshold. Due to smaller set sizes to compute on, \intAwarePSIProtocol compute time is $3\times$ lower than the baseline when the threshold is 100. 	
	
\end{itemize}

\myparagraph{Application Benchmark}
To further explore realistic parameter settings, as a real-world application of \intAwarePSIProtocol, we return to the problem of private collaborative blacklisting where two mutually-untrusting parties compare IP addresses of end-points from where they have observed traffic to their own network. This task is usually performed with a PSI protocol \cite{melis2019collabblacklist}. Replacing this with a DA-PSI protocol enables us to find IP addresses that are common to both sets, as well as find addresses that are ``close'' in the address space. A distance-based comparison is meaningful here because it is well-know that coordinated attacks usually span multiple subnets~\cite{west2010spam,collins2007uncleanliness}.


\myparagraph{Dataset}
To test this application, we have collected data from a public honeypot deployed in a university network. The honeypot logs all incoming and outgoing traffic and stores a wealth of information. From this data, we curate information about traffic observed on two separate days, and build sets with the source IP addresses. There are roughly 25,000 distinct IP addresses in each set. These sets are inputs to \intAwarePSIProtocol and the baseline PSI protocol. 

We observe that using a distance-aware intersection in this context is well-justified based on the results. The number of intersections increases significantly when increasing the search radius and almost doubles by the time we reach the threshold of $128$. In actual numbers, there are around 5\% exact matches between the two sets.
With a distance based search over thresholds of $2, \ldots 128$ we find that the number of items in the intersection increase to more than 10\%.  By searching over larger threshold, we are able to obtain matching IP addresses that fall in the same subnet/adjoining subnets. Note that without a full subnet map, it is not possible to predict these subnet sizes a priori. Therefore, we envision running the protocol multiple times with different thresholds to obtain the most informative intersection.

\myparagraph{Results}
\figsref{fig:intPSI:macro:comm_threshold}{fig:intPSI:macro:comp_threshold} show how the overall communication volume and compute time scale with distance thresholds set to $2, \ldots, 128$. The intuition behind increasing the threshold in powers of two is that we would like to search over entire subnets and find potential overlaps (if any). In terms of the overall runtime we observe that \intAwarePSIProtocol  scales more gracefully with the distance threshold. Note that with threshold set at $128$, the baseline protocol computes over sets of size exceeding 6 million, while \intAwarePSIProtocol computes over sets of size of around 200,000. This difference results in a significant speedup. With threshold = 128, \intAwarePSIProtocol requires only 1.5 seconds to compute the intersection while the baseline requires over 15 seconds to accomplish the same task.

\section{Conclusion}
\label{sec:conclusion}
In this paper, we introduced the distance-aware PSI problem over metric spaces, whereby parties privately compute an intersection of their respective sets with items that are ``close'' in the metric space ending up in the intersection. Closeness is defined based on a user-specified distance threshold in the metric space. As concrete instantiations, we provided distance-aware constructions for two metric spaces: Minkowski distance of order 1 over the integers and Hamming distance.  Both the protocols are communication-efficient. As a practical application of this idea, we evaluated the Minkowski distance protocol in the context of collaborative blacklisting. In addition, the Hamming distance-aware protocol allows constructions for other distances using techniques like locality-sensitive hashing.

\section{Acknowledgments}

This research was supported in part by grant numbers 2040675, Convergence Accelerator award 2040675 and CIF-1705007
from the National Science Foundation, and W911NF-17-1-0370 from the Army Research Office. This work was also made possible by the support of JP Morgan Chase, the Sloan Foundation, Siemens AG, and Cisco. The views and conclusions in this document are those of the authors and should not be interpreted as representing the official policies, either expressed or implied, of the National Science Foundation, Army Research Office, or the U.S. Government.

\bibliographystyle{plainnat}
\bibliography{bib/psi.bib,bib/lsh.bib,bib/crypto.bib,bib/misc.bib, bib/math.bib}

\appendix
\section{Difference Between \tPSI and \oneSidedSetRecon}
\label{app:tpsi_delta}

This section highlights the differences between \tPSI \cite[Fig. 10]{ghosh2019thresholdPSI} and \oneSidedSetRecon (\figref{fig:ham_query_lite}). For this we present the \tPSI set reconciliation protocol in \figref{fig:tpsi}. The protocol requires two sets of calls to \oleFunc in Steps 5 and 6. It also requires two additional rounds of communication in Steps 7 and 8. These steps are required so that both \alice and \bob can obtain the evaluations of \tPSIOutputPoly in Steps 9 and 10. More specifically, $s_{\mathsf{A}}(x_\ptIdx), s'_{\mathsf{A}}(x_\ptIdx)$ and $s_{\mathsf{B}}(x_\ptIdx), s'_{\mathsf{B}}(x_\ptIdx)$ are "blinded" shares of \tPSIOutputPolyEval{x_\ptIdx} with \alice's and \bob's inputs respectively. \alice and \bob exchange these shares to finally obtain \tPSIOutputPolyEval{x_\ptIdx} in Steps 9 and 10.

\oneSidedSetRecon avoids Steps 6 -- 8 by allowing only \alice to obtain the evaluations of \tPSIOutputPoly. while \bob generates all the random polynomials required in the protocol. Therefore, we do not need to generate "blinded" shares as above. As a result, \oneSidedSetRecon is not only significantly simpler but also avoids one set of \numPointsForInterpol of calls to \oleFunc and two rounds of communication.

\begin{figure}
	\footnotesize
	\begin{mdframed}
		
		\noindent
		\textbf{\underline{Parameters:}}  \alice and \bob have sets \alicesSetFromVec and \bobsSetFromVec, $\setSize{\alicesSetFromVec} = \setSize{\bobsSetFromVec} = \genericVecLen$ respectively and a set difference threshold \hamDistThreshold. 
		
			\noindent
		\textbf{\underline{Procedure Noisy Polynomial Addition:}}
		
		\begin{enumerate}
		\item  \alice and \bob select a set of \numPointsForInterpolPlus points in \finiteField $\setOfXCoords \assign \setOfXCoordsDefn$ such that none of the points  are in the ranges of any of the mapping functions. 
		
		\item \alice encodes \alicesSetFromVec  in the polynomial \alicesPoly = \polyFromSet{\alicesSetFromVec} and \bob encodes \bobsSetFromVec in the polynomial \bobsPoly = \polyFromSet{\bobsSetFromVec}. 
		
		\item \alice picks two random polynomials of degree $\genericVecLen$, $R_{1}^{\mathsf{A}}, R_{2}^{\mathsf{A}} \in \polyField$ and a degree $2\genericVecLen$ polynomial $U_{\mathsf{A}} \in \polyField$. 
		
		\item \bob picks two random polynomials of degree $\genericVecLen$, $R_{1}^{\mathsf{B}}, R_{2}^{\mathsf{B}} \in \polyField$ and a degree $2\genericVecLen$ polynomial $U_{\mathsf{B}} \in \polyField$. 

		\item For each $x_\ptIdx \in \setOfXCoords$, \alice sends $\alicesPolyFromVecEval{x_\ptIdx}$ to \oleFunc. \bob sends $R_{1}^{\mathsf{B}}(x_\ptIdx)$ and  $U_{\mathsf{B}}(x_\ptIdx)$ to \oleFunc. \alice receives $s_{\mathsf{A}}(x_\ptIdx) \assign \alicesPolyFromVecEval{x_\ptIdx} R_{1}^{\mathsf{B}}(x_\ptIdx) + U_{\mathsf{B}}(x_\ptIdx)$. 
		
		\item For each $x_\ptIdx \in \setOfXCoords$, \bob sends $\bobsPolyFromVecEval{x_\ptIdx}$ to \oleFunc. \alice sends $R_{1}^{\mathsf{A}}(x_\ptIdx)$ and  $U_{\mathsf{A}}(x_\ptIdx)$ to \oleFunc. \alice receives $s_{\mathsf{B}}(x_\ptIdx) \assign \bobsPolyFromVecEval{x_\ptIdx} R_{1}^{\mathsf{A}}(x_\ptIdx) + U_{\mathsf{A}}(x_\ptIdx)$. 
		
		\item For each $x_\ptIdx \in \setOfXCoords$ \alice sends to \bob $s'_{\mathsf{A}}(x_\ptIdx) \assign s_{\mathsf{A}}(x_\ptIdx) + \alicesPolyFromVecEval{x_\ptIdx} \cdot R_{1}^{\mathsf{A}}(x_\ptIdx) - U_{\mathsf{A}}(x_\ptIdx)$. 
	
		\item For each $x_\ptIdx \in \setOfXCoords$ \bob sends to \bob $s'_{\mathsf{B}}(x_\ptIdx) \assign s_{\mathsf{B}}(x_\ptIdx) + \bobsPolyFromVecEval{x_\ptIdx} \cdot R_{1}^{\mathsf{B}}(x_\ptIdx) - U_{\mathsf{B}}(x_\ptIdx)$. 

		\item \alice outputs the evaluation points of $\alicesVoleOutputPolyEval{x_\ptIdx} \assign s_{\mathsf{A}}(x_\ptIdx) + s'_{\mathsf{B}}(x_\ptIdx) + \alicesPolyFromVecEval{x_\ptIdx} R_{1}^{\mathsf{A}}(x_\ptIdx) -  U_{\mathsf{A}}(x_\ptIdx) = \tPSIOutputPolyEval{x_\ptIdx}$. 

			\begin{center}
			$\setOfPtsForInterpolation_{\mathsf{A}} \assign  \{(x_\ptIdx, y_\ptIdx) : y_\ptIdx \assign \frac{\alicesVoleOutputPolyEval{x_\ptIdx}}{\alicesPolyFromVecEval{x_\ptIdx}} =  \frac{\tPSIOutputPolyEval{x_\ptIdx}}{\alicesPolyFromVecEval{x_\ptIdx}} \}$.
		\end{center}

		\item \bob outputs the evaluation points of $\bobsVoleOutputPolyEval{x_\ptIdx} \assign s_{\mathsf{B}}(x_\ptIdx) + s'_{\mathsf{A}}(x_\ptIdx) +  \bobsPolyFromVecEval{x_\ptIdx} R_{1}^{\mathsf{B}}(x_\ptIdx) -  U_{\mathsf{B}}(x_\ptIdx) = \tPSIOutputPolyEval{x_\ptIdx}$. 
		
		\begin{center}
			$\setOfPtsForInterpolation_{\mathsf{B}} \assign  \{(x_\ptIdx, y_\ptIdx) : y_\ptIdx \assign \frac{\bobsVoleOutputPolyEval{x_\ptIdx}}{\bobsPolyFromVecEval{x_\ptIdx}} =  \frac{\tPSIOutputPolyEval{x_\ptIdx}}{\alicesPolyFromVecEval{x_\ptIdx}} \}$.
		\end{center}

		\end{enumerate}		

			\noindent
		\textbf{\underline{Procedure Interpolation:}}
		
		\item \alice  interpolates $\setOfPtsForInterpolation_{\mathsf{A}}$ with a rational function similar to \oneSidedSetRecon. \bob $\setOfPtsForInterpolation_{\mathsf{B}}$ with a rational function similar to \oneSidedSetRecon. The rest of the steps are the same as \oneSidedSetRecon; \alice and \bob output \setDiff{\alicesSetFromVec}{\bobsSetFromVec} and \setDiff{\bobsSetFromVec}{\alicesSetFromVec} respectively.

	\end{mdframed}	\vspace{-15pt }
\caption{\tPSI set reconciliation protocol \cite{ghosh2019thresholdPSI}. \label{fig:tpsi}}

\end{figure}

\section{Proofs of Propositions}
\label{app:proofs_hamming}

\subsection{Lemma for Polynomials in \polyField}
\label{app:lemmaPolys}

\begin{restatable}{lemma}{polyOnSupSet}
	\label{lemma:polyOnSupSet}
	Let $X$ be a set of distinct arbitrary values $x_i \in \finiteField$.  
	Let $V$ be a set of random points where each point $v_i = (x_i, y_i), y_i \getsr \finiteField$. Then, the probability that there exists polynomial \polyInField{Q} that satisfies all points in $V$ with $\degree{Q}  \leq D < \setSize{V}$ is $1/\fieldOrder^{\setSize{V} - D - 1}$.
	
\end{restatable}

\begin{proof}
	
	Let $V^{*}$ be the combined set of all sets $V = \{\genericPt : \genericPtX \in \setOfXCoords, \genericPtY \in \finiteField\}$. Then, $|V^{*}| = \fieldOrder^{\setSize{V}}$. Each $V_i \in V^{*}$ is consistent with {\em at most} one polynomial of degree $\leq D$. Otherwise, there will be two or more polynomials of degree $\leq D$ passing through the same \setSize{V} points, which contracts the fact that any $D + 1 \leq \setSize{V}$ points uniquely defines a polynomial of degree $\leq D$

	Let $P^{*}$ be the set of all polynomials of degree $\leq D$. Clearly, $| P^{*} | =  \fieldOrder^{D+1}$. Each $P_i \in P^{*}$ is consistent with exactly one $V_j \in V^{*}$. Thus, the mapping from $P^{*}$ to $V^{*}$ is injective. Then the probability that $V \in V^{*}$ is satisfied by a polynomial \polyInField{Q} of degree $\degree{Q} \leq D$ is given by the probability that $V$ has a pre-image in the injective map which is $1/\fieldOrder^{\setSize{V} - D - 1}$.
\end{proof}

\subsection{Proof of \thmref{thm:set_recon_claims}}
\label{app:proof_set_recon_claims}

To prove the propositions in the theorem, we will define an IND-CPA style security game between an adversary \indcpaAdversary and a challenger \indcpaChallenger. 

\begin{oframed}
 	\begin{enumerate}[nosep,leftmargin=1.6em,labelwidth=*,align=left]
 	
 		\item \indcpaAdversary selects a threshold $\hamDistThreshold \in (0, \genericVecLen/2)$,  
 		three sets \indcpaFixedSet, \indcpaChallengeSetZero and \indcpaChallengeSetOne with $\setSize{\indcpaChallengeSetZero} = \setSize{\indcpaChallengeSetOne} = \setSize{\indcpaFixedSet} =\genericVecLen$, such that 
 		
 		$\setIntCard{\indcpaFixedSet}{\indcpaChallengeSetZero} =  \setIntCard{\indcpaFixedSet}{\indcpaChallengeSetOne} = \degreeOfGCD$, and  $\setOfXCoords \assign \{\genericPtX : x_\ptIdx \in \finiteField \}_{\ptIdx = 1}^{\numPointsForInterpol}$ such that $x_\ptIdx \notin \indcpaFixedSet, \indcpaChallengeSetZero, \indcpaChallengeSetOne$.     
 		
 			
 		 


 		\item \indcpaChallenger samples two random polynomials $\alicesRandomPoly, \bobsRandomPoly \getsr \polyField$ of degree $\genericVecLen$, and derives \indcpaFixedPoly \assign  \polyFromSet{\indcpaFixedSet}. 
 	 	
 	 	\item \indcpaChallenger flips a random bit \indcpaBit and based on outcome derives  \indcpaChallengePolyBit \assign  \polyFromSet{\indcpaChallengeSetBit}, and computes $\setOfPtsForInterpolation = \{(x_\ptIdx, y_\ptIdx) : x_\ptIdx \in \setOfXCoords, y_\ptIdx = \indcpaOutputPolyEval{x_\ptIdx} \}$. \indcpaChallenger returns \setOfPtsForInterpolation to \indcpaAdversary.


 		
 		
 		
 			
 			
 			
 		
 		\item \indcpaAdversary 	outputs \indcpaAdversaryBit and wins the game if $\indcpaAdversaryBit = \indcpaBit$.

\end{enumerate}
\end{oframed}

 \begin{restatable}{prop}{indcpaClaimOne}
 \label{claim:indcpaGreaterThanRange}
 	For $\indcpaBit \in [0,1]$, if $\setDiffCard{\indcpaFixedSet}{\indcpaChallengeSetBit} =  \genericVecLen - \degreeOfGCD \geq 2\hamDistThreshold$, then \prob{\indcpaAdversaryBit = \indcpaBit} $\leq 1/2 +  \neglfn{\secParam}$ for any adversary \indcpaAdversary.
 \end{restatable}

\begin{proof}
	We will show that the probability of obtaining an arbitrary set of points \setOfPtsForInterpolation in Step 5 of the security game is the same for the case when \indcpaBit = 0 as the case \indcpaBit = 1. For this, consider the polynomial $ \indcpaChallengePolyBit R_1(x) + \indcpaFixedPoly R_2(x) = \gcd{\indcpaFixedPoly}{\indcpaChallengePolyBit} \times R(x)$ where $R(x)$ is a random polynomial of degree $2\genericVecLen - \degreeOfGCD$ due to \lemmaref{lemma:kissener_uniformly_random}. The evaluations of this polynomial generates \setOfPtsForInterpolation.

	If $\indcpaBit = 0$,  $R(x)$ is a degree $2\genericVecLen - \degreeOfGCD$  polynomial that is consistent with the set of points $\{(\genericPtX, y'_\ptIdx): y'_\ptIdx = \frac{y_\ptIdx}{C_O(x_\ptIdx)}, ~(x_\ptIdx, y_\ptIdx) \in \setOfPtsForInterpolation\}$ where $C_0(x) \assign \gcd{\indcpaFixedPoly}{\indcpaChallengePolyZero}$. A degree $2\genericVecLen - \degreeOfGCD$ polynomial in \polyField is uniquely defined by $2\genericVecLen - \degreeOfGCD + 1$ points. 
	Since, $R(x)$ is consistent with the aforementioned set of points, $\setSize{\setOfPtsForInterpolation} = \numPointsForInterpol$  points required to define $R(x)$ are fixed. Thus, there are $\fieldOrder^{(2\genericVecLen - \degreeOfGCD+1) - (\numPointsForInterpol) = \fieldOrder^{\genericVecLen -(2\hamDistThreshold + \degreeOfGCD)}}$ candidate polynomials for $R(x)$. Let $S_0$ be the set comprising these polynomials. 
	
	Similarly, if $\indcpaBit = 1$, $R(x)$ is one of the $\fieldOrder^{\genericVecLen -(2\hamDistThreshold + \degreeOfGCD)}$ 
	polynomials of degree $\leq 2\genericVecLen - \degreeOfGCD$ consistent with the set of points $\{(\genericPtX, y'_\ptIdx): y'_\ptIdx = \frac{y_\ptIdx}{C_1(x_\ptIdx)}, ~(x_\ptIdx, y_\ptIdx) \in \setOfPtsForInterpolation\}$ where $C_1(x) \assign \gcd{\indcpaFixedPoly}{\indcpaChallengePolyOne}$. Let $S_1$ be the set comprising these polynomials; note that $|S_0| = |S_1|$. 
	
	Since $R(x)$ is a uniformly random polynomial of degree $2\genericVecLen - \degreeOfGCD$, the probability of obtaining \setOfPtsForInterpolation when \indcpaBit = 0 is 
	$\prob{R(x) \in S_0} = \frac{\fieldOrder^{\genericVecLen -(2\hamDistThreshold + \degreeOfGCD)}}{\fieldOrder^{2\genericVecLen - \degreeOfGCD + 1}} = \fieldOrder^{-(\numPointsForInterpol)}$. Similarly, the probability of obtaining \setOfPtsForInterpolation when \indcpaBit = 1 is $\prob{R(x) \in S_1} = \frac{\fieldOrder^{\genericVecLen -(2\hamDistThreshold + \degreeOfGCD)}}{\fieldOrder^{2\genericVecLen - \degreeOfGCD + 1}} = \fieldOrder^{-(\numPointsForInterpol)}$. 
\end{proof}

 \begin{restatable}{prop}{indcpaClaimTwo}
	\label{claim:indcpaInRange}
	 For $\indcpaBit \in [0,1]$, if $\setDiffCard{\indcpaFixedSet}{\indcpaChallengeSetBit} = \genericVecLen - \degreeOfGCD \in (\hamDistThreshold, 2\hamDistThreshold)$, then there is an adversary \indcpaAdversary for which \prob{\indcpaAdversaryBit = \indcpaBit} $\geq 1 - \neglfn{\secParam}$.
\end{restatable}

\begin{proof}

    W.l.o.g assume that $\indcpaChallengePolyZero R_1(x) + \indcpaFixedPoly R_2(x) = \gcd{\indcpaFixedPoly}{\indcpaChallengePolyZero} R^{*}_0(x)$ and $\indcpaChallengePolyZero R_1(x) + \indcpaFixedPoly R_2(x) = \gcd{\indcpaFixedPoly}{\indcpaChallengePolyOne} R^{*}_1(x)$
    are two polynomials consistent with the set of point \setOfPtsForInterpolation in Step 5. Then, we have $R^{*}_0(x)$ is a random polynomial of degree $2\genericVecLen - \degreeOfGCD < \genericVecLen + 2\hamDistThreshold$ consistent with the set of points $\{(\genericPtX, y'_\ptIdx): y'_\ptIdx = \frac{y}{C_0(x_\ptIdx)}, ~\genericPt \in V\}$ and $R^{*}_1(x)$ is a random polynomial of degree $2\genericVecLen - \degreeOfGCD < \genericVecLen + 2\hamDistThreshold$ consistent with the set of points $\{(\genericPtX, y'_\ptIdx): y'_\ptIdx = \frac{y}{C_1(x_\ptIdx)}, ~\genericPt \in V\}$ where $C_0(x) = \gcd{\indcpaFixedPoly}{\indcpaChallengePolyZero}$ and $C_1(x) = \gcd{\indcpaFixedPoly}{\indcpaChallengePolyOne}$. 

    From these facts we get that, for $i \in [1, \numPointsForInterpol], R^{*}_1(x_\ptIdx) = R^{*}_0(x_\ptIdx) \times \frac{C_0(x_\ptIdx)}{C_1(x_\ptIdx)}$. Now, since $R^{*}_0(x)$ is a random polynomial, the set of points $\{(x_1, R^{*}_1(x_1)), \ldots, (x_{\numPointsForInterpol}, R^{*}_1(x_{\numPointsForInterpol}))\}$ is a set of random points. Thus, after fixing $R^{*}_0(x)$ when \indcpaBit = 0, the probability that there exists some $R^{*}_1(x)$ consistent with the aforementioned set of points is $< \frac{1}{\fieldOrder^{(\numPointsForInterpol) - (\degree{R^{*}_1(x)} - 1}} \leq 1/\fieldOrder$ due to \lemmaref{lemma:polyOnSupSet}.
    A similar logic holds when $\indcpaBit = 1$.

	\indcpaAdversary's strategy is to output $\indcpaAdversaryBit = 0$ if there exists a polynomial $R^{*}_0(x)$ of  degree $< \genericVecLen + 2\hamDistThreshold$ consistent with points in \setOfPtsForInterpolation. Similarly, \indcpaAdversary outputs $\indcpaAdversaryBit = 1$ if there exists some polynomial $R^{*}_1(x)$ of degree $< \genericVecLen + 2\hamDistThreshold$ consistent with the set of points \setOfPtsForInterpolation. The probability that  $R^{*}_0(x)$ and  $R^{*}_1(x)$ both exist $< 1/\fieldOrder$ as shown above, which is the probability that $\indcpaAdversaryBit \neq \indcpaBit$. 
\end{proof}

\section{Proofs for \tHamQuery (\secref{sec:hamming_psi:poly_compute})}
\label{app:hamming_query_poly}

\subsection{Using EC-Elgamal in \tHamQuery}
\label{app:el_gamal}
\tHamQueryRestricted uses a additively homomorphic encryption scheme. In our implementation we use the EC-Elgamal encryption scheme tp reduce communication costs. For this, the message space should be large enough to allow the computation in \tHamQueryRestricted, and decrypt the results. We use a 24-bit message space. The computation in \tHamQueryRestricted broadly involves two steps: i) computing the Hamming distance between two vectors as follows, and ii) returning a key \prfKey blinded with the result of the Hamming distance computation. First, consider the mechanism we use to compute the Hamming distances between two vectors. 

\begin{lemma}
	\label{lemma:hamming_dist_compute}
	
	Given two equal length vectors $\alicesVector, \bobsVector \in \{0,1\}^{\genericVecLen}$, the Hamming distance between the vectors is given by:
	
	\begin{center}
		$\hamDist{\alicesVector}{\bobsVector} = \hamWeight{\alicesVector} + \hamWeight{\bobsVector} - 2 \cdot  \alicesVector \cdot \bobsVector$ 
	\end{center}
	\noindent 
	where $\hamWeight{.}$ is the Hamming weight of the input vector.
\end{lemma}

\begin{proof}
	Let $S_{01} = \{m \in [1, \genericVecLen] ~:~ \alicesVector[m] = 0, \bobsVector[m] = 1\}$, $S_{10} = \{m \in [1, \genericVecLen] ~:~ \alicesVector[m] = 1, \bobsVector[m] = 0\}$ and $S_{11} = \{m \in [1, \genericVecLen] ~:~ \alicesVector[m] = 1, \bobsVector[m] = 1\}$. Then, 
	
	\begin{center}
		$\hamWeight{\alicesVector} = \setSize{S_{10}} + \setSize{S_{11}}$ \\
		$\hamWeight{\bobsVector} = \setSize{S_{01}} + \setSize{S_{11}}$ \\
		$\hamDist{\alicesVector}{\bobsVector} = \setSize{S_{10}} + \setSize{S_{01}}$
	\end{center}
	
	The above equations along with the fact $\setSize{S_{11}} = \sum\limits_{m = 1}^{\genericVecLen} \alicesVector[m] \cdot \bobsVector[m] = \alicesVector \cdot \bobsVector$ proves the result. 
\end{proof}

As long as the maximum Hamming distance between the vectors $< 2^{24}$, the 24-bit message space suffices for this computation. After computing $\encrypt{\hamDist{\alicesVector}{\bobsVector}}$, \bob computes $\keySet \assign \{\kappa_i : \kappa_i \assign r_i \times (\encrypt{\hamDist{\alicesVector}{\bobsVector} - i}) + \encrypt{\prfKey}, ~i \in [0, \hamDistThreshold], ~r_i \gets \finiteField \}$. However, the problem here is that $\prfKey \in \finiteField$ and for a sufficiently high statistical security parameter, we require $\fieldOrder$ to be a at least $128$-bit long. Thus, \prfKey does not fit in the 24-bit message space. 

To mitigate this, we split \prfKey into 24-bit chunks. Each individual chunk is encrypted separately, and returned to \alice. In other words, \bob splits \prfKey into $c = \ceil{\frac{\setSize{\fieldOrder}}{24}}$ chunks $\mathsf{Chunk}_1, \ldots, \mathsf{Chunk}_c$, and \alice now receives $\keySet \assign \{\kappa_{ij}: \kappa_{ij} \assign r_{ij} \times (\encrypt{\hamDist{\alicesVector}{\bobsVector} - i}) + \encrypt{\mathsf{Chunk}_j}, ~i \in [0, \hamDistThreshold], ~r_{ij} \gets \finiteField, j \in [1, c]\}$. From this, \alice can obtain $\kappa_i \assign \kappa_{i1} || \ldots || \kappa_{ic}$. Note each 24-bit chunk in encrypted with IND-CPA security, and therefore splitting the key as described has no impact on security. There is a $c$ times blowup in the downstream communication cost i.e., the cost of sending \keySet to \alice. 

\subsection{Proofs}

\begin{lemma}
	\label{lemma:sub_sample_set_recon}
	Let $\alicesSetFromVecBins := \{\vec{x}_1, \ldots, \vec{x}_{\numOfBinsRef}\}$ and $\bobsSetFromVecBins := \{\vec{y}_1, \ldots, \vec{y}_{\numOfBinsRef}\}$ be the set of sub-vectors created after sub-sampling \alicesVector and \bobsVector using \randPerm in \tHamQuery. Then, $\setDiffCard{\alicesSetFromVecBins}{\bobsSetFromVecBins} \geq 2\hamDistThreshold$ if $\hamDist{\alicesVector}{\bobsVector} > 2\hamDistThreshold$ with high probability. 
\end{lemma}

\begin{proof}
	To prove the result, we use the following balls and bins analysis: 	let each index where \alicesVector and \bobsVector differ be represented by a ball. There are $> 2\hamDistThreshold$ such indices and they are uniformly distributed across \numOfBinsRef bins, where the $i$th bin contains the indices which make up the sub-vectors \alicesSubVecIdx{i} and \bobsSubVecIdx{i}. In this framework, the number of \textit{non-empty} bins gives us $| \alicesSetFromVecBins \setminus \bobsSetFromVecBins |$. Then, \lemmaref{lemma:balls_bins_2} shows that $\setDiffCard{\alicesSetFromVecBins}{\bobsSetFromVecBins} \geq 2\hamDistThreshold$ with high probability. 
	
\end{proof}

\begin{lemma}
	\label{lemma:balls_bins_2}
	If 2\hamDistThreshold + 1 balls are randomly thrown into $\numOfBinsRef = \numOfBins$ bins where $\epsilon \in (0,0.5)$, then there are more than $ 2\hamDistThreshold$ non-empty bins with probability at least $1 - \epsilon$. 
\end{lemma}

\begin{proof}
	Let $E_k$ be the event that there are exactly $(2\hamDistThreshold  + 1 - k)$ occupied bins where $k \in [2, 2\hamDistThreshold - 1)$. Then,
	
	\begin{center}
		$\prob{E_k} = \frac{(2\hamDistThreshold + 1 - k)! \stirling{2\hamDistThreshold+1}{2\hamDistThreshold + 1 - k} \binom{\numOfBinsRef}{\numOfBinsRef - (2\hamDistThreshold + 1- k)}}{(\numOfBinsRef)^{2\hamDistThreshold+1}}$ \\
		
		$ = \frac{(2\hamDistThreshold + 1 - k)! \stirling{2\hamDistThreshold+1}{2\hamDistThreshold + 1 - k} \binom{\numOfBinsRef}{2\hamDistThreshold + 1 - k}}{(\numOfBinsRef)^{2\hamDistThreshold+1}}$
		
		$ \leq \stirling{2\hamDistThreshold+1}{2\hamDistThreshold + 1 - k} \times \left(\numOfBinsRef\right)^{-k}$ 
	\end{center}

	where $\stirling{a}{b}$ denotes the Stirling number of the second kind. The last inequality holds since $\binom{\numOfBinsRef}{2\hamDistThreshold + 1 - k} \leq \frac{(\numOfBinsRef)^{2\hamDistThreshold + 1 - k}}{(2\hamDistThreshold + 1 - k)!}$.

	Using the upper bound for Stirling number of second kind, $\stirling{a}{b} \leq \frac{1}{2} \binom{a}{b} b^{a - b}$, we have
	
	\begin{center}
		$\stirling{2\hamDistThreshold+1}{2\hamDistThreshold + 1 - k}  \leq \frac{1}{2} \binom{2\hamDistThreshold+1}{2\hamDistThreshold + 1 - k} (2\hamDistThreshold + 1 - k)^{k} \leq \frac{1}{2} \frac{(2\hamDistThreshold + 1)^{k} (2\hamDistThreshold + 1 - k)^{k}}{k!} $
		
	\end{center}
	
	Now, $\frac{(2\hamDistThreshold + 1) (2\hamDistThreshold + 1 - k)}{\numOfBinsRef} = \frac{(2\hamDistThreshold + 1) (2\hamDistThreshold + 1 - k)}{\numOfBins} < 2\epsilon$. Thus, 
	
	\begin{center}
		$\prob{E_k} < \frac{1}{2} \times \frac{(2\epsilon)^k}{k!}$
		
	\end{center}
	
	Let $E$ be the event that the number of occupied bins is less than $ 2\hamDistThreshold$. Then,

	\begin{center}
		$\prob{E} \leq \sum\limits_{k = 2}^{2\hamDistThreshold - 2} \prob{E_k} = \frac{1}{2} \sum\limits_{k = 2}^{2\hamDistThreshold - 2} \frac{(2\epsilon)^k}{k!} < \frac{1}{2} (e^{2\epsilon} - (1 + 2\epsilon))$
	\end{center}
	
	The last inequality holds since $\sum\limits_{k = 2}^{\infty} \frac{(2e)^k}{k!} = e^{2e} - (1 + 2\epsilon)$. Also, for any $\epsilon \in (0, 0.5)$, 
	\begin{center}
		$\frac{1}{2} (e^{2\epsilon} - (1 + 2\epsilon)) < \epsilon$		
	\end{center}
\end{proof}

%
\securityHamQuery*

\begin{proof}
	
	We show that there is a PPT simulator \simm in the ideal world which indistinguishably simulates the real world execution of \tHamQuery

	\myparagraph{Simulating \bob's view:}
	\bob does not receive any output from the protocol and only observes intermediate results from \tHamQueryRestricted and \oleFunc. Assuming that \oleFunc is realized by a protocol which can be indistinguishably simulated and the AHE scheme used in \tHamQueryRestricted produces IND-CPA secure ciphertexts indistinguishable from random, \bob's view in \hamQuerySample simulating \bob's view is straightforward.

	\myparagraph{Simulating \alice's view}
	\alice's input into the protocol is \alicesVector and \alice receives \keySet and the evaluations of the polynomials \alicesVoleOutputPoly and \bobsVoleOutputPoly. The simulation strategy is as follows.

	\smallskip\noindent\underline{\textit{When $\hamDist{\alicesVector}{\bobsVector} \leq \hamDistThreshold$}:}	
	In this case, \tHamQueryFunc returns \bobsVector. The simulator follows the steps of \tHamQuery 
	generating all outputs using \bobsVector. The simulation is clearly indistinguishable. 
	
	\smallskip\noindent\underline{\textit{When $\hamDist{\alicesVector}{\bobsVector} > \hamDistThreshold$}:}
	In the real world, \alice obtains \keySet, and the set of evaluations $\alicesVoleOutputPolyPtSetReal \assign \{ (x_\ptIdx, y_\ptIdx) : 
	~y_\ptIdx \assign \tPSIOutputPolyEval{x_\ptIdx} + \phi(\prfKey, \ptIdx), ~\ptIdx \in [1, \numOfBinsRef + 2\hamDistThreshold + 1]\}$. 
	where \numOfBinsRef \assign \numOfBins.

	\smallskip\noindent
	\underline{Case when $\hamDist{\alicesVector}{\bobsVector} \in (\hamDistThreshold, 2\hamDistThreshold]$ and \tHamQueryFunc does not output \bobsVector:} In this case, 
	$\prfKey \notin \keySet$ which implies $\hamDist{\alicesParityVec}{\bobsParityVec} \notin [0, \hamDistThreshold)$. Thus, for each $\encrypt{\kappa_i} \in \keySet, ~ \kappa_i \assign r_i(\hamDist{\alicesParityVec}{\bobsParityVec} - i) + \kappa, ~ i \in [0, \hamDistThreshold], ~r_i \getsr \finiteField, \kappa \getsr \finiteField$, $\kappa_i$ is a random element in \finiteField. 
	In the ideal world, \simm sets the key set $\keySetSim \assign \{\encrypt{r_i} : ~r_i \getsr \finiteField\}_{i = 1}^{\setSize{\keySet}}$. From the above we have $\keySet \sim \keySetSim$.

	To simulate the output of \alicesVoleOutputPolyPtSetReal in \lineref{step:ham_query:ole_step_1} of \tHamQuery, for $\ptIdx \in [1, \numOfBinsRef + 2\hamDistThreshold + 1]$, \simm sends \alicesPolyFromVecEval{x_\ptIdx}, \alicesRandomPolyEval{x_\ptIdx}, and $\bobsPolyFromVecSimEval{x_\ptIdx} \bobsRandomPolyEval{x_\ptIdx} + \finiteFieldPRF(\prfKey, \ptIdx)$ to \oleFunc. For $\ptIdx \in [1, \numOfBinsRef + 2\hamDistThreshold + 1]$, \oleFunc returns  $\alicesVoleOutputPolySimEval{x_\ptIdx} \assign \alicesRandomPolyEval{x_\ptIdx} \alicesPolyFromVecEval{x_\ptIdx}  + \bobsRandomPolyEval{x_\ptIdx} \bobsPolyFromVecSimEval{x_\ptIdx} + \finiteFieldPRF(\prfKey, \ptIdx)$. \simm returns $\alicesVoleOutputPolyPtSetSim \assign \{(x_\ptIdx, y_\ptIdx)~:  y_\ptIdx \assign \alicesVoleOutputPolySimEval{x_\ptIdx}, ~\ptIdx \in [1, \numOfBinsRef + 2\hamDistThreshold + 1] \}$ to \alice
	Observe that for $\ptIdx \in [1, \numPointsForInterpolAfterBins], \alicesVoleOutputPolyEval{x_\ptIdx} \assign  \tPSIOutputPolyEval{x_\ptIdx} + \finiteFieldPRF(\prfKey, \ptIdx)$ and $\alicesVoleOutputPolySimEval{x_\ptIdx} \assign \alicesRandomPolyEval{x_\ptIdx} \alicesPolyFromVecEval{x_\ptIdx} + \bobsRandomPolyEval{x_\ptIdx} \bobsPolyFromVecSimEval{x_\ptIdx}  + \finiteFieldPRF(\prfKey, \ptIdx)$. If the PRF \finiteFieldPRF outputs uniformly random values in \finiteField, $\alicesVoleOutputPolyEval{x_\ptIdx} \getsr \finiteField$ and $\alicesVoleOutputPolySimEval{x_\ptIdx} \getsr \finiteField$. Thus, $\alicesVoleOutputPolyPtSetReal \sim \alicesVoleOutputPolyPtSetSim$.

	Since, \prfKey is selected independently of all $\kappa_i \in \keySet$, we have that the points in 
	\alicesVoleOutputPolyPtSetReal are independent of the "keys" in \keySet. 
	Similarly, the points \alicesVoleOutputPolyPtSetSim are independent of the random elements in \keySetSim. Thus, we have 
	$(\keySet, \alicesVoleOutputPolyPtSetReal) \sim (\keySetSim, \alicesVoleOutputPolyPtSetSim)$ which shows that the ideal world simulation is indistinguishable from the real world execution. 
	
	\smallskip\noindent
	\underline{Case when $\hamDist{\alicesVector}{\bobsVector} \in (\hamDistThreshold, 2\hamDistThreshold]$ and \tHamQueryFunc does not output \bobsVector:}
	There are two cases: i) when $\prfKey \notin \keySet$, and ii) when $\prfKey \in \keySet$. In case (i), we can apply the same  arguments as above. Since these cases are not distinguishable, \simm follows the same strategy for both cases. More specifically, as above, \simm sets $\keySetSim \assign \{\encrypt{r_i} : ~r_i \getsr \finiteField\}_{i = 1}^{\setSize{\keySet}}$ and $\alicesVoleOutputPolyPtSetSim \assign \{(x_\ptIdx,y_\ptIdx ) : y_\ptIdx \assign \alicesVoleOutputPolySimEval{x_\ptIdx} \assign \alicesRandomPolyEval{x_\ptIdx} \alicesPolyFromVecEval{x_\ptIdx} + \bobsRandomPolyEval{x_\ptIdx} \bobsPolyFromVecSimEval{x_\ptIdx}  + \finiteFieldPRF(\prfKey, \ptIdx), ~\ptIdx \in [1, \numOfBinsRef + 2\hamDistThreshold + 1]\}$. 
	
	\keySet and \keySetSim are indistinguishable because both sets comprise random field elements. Also, the points in  $\alicesVoleOutputPolyPtSetReal$ are independent of the "keys" in the \keySet. Thus, we only need to show that \alicesVoleOutputPolyPtSetReal and \alicesVoleOutputPolyPtSetSim are indistinguishable.

	Let $\universe_\fieldOrder$ be the set of all sets $\{(x_\ptIdx, y_\ptIdx) : (x_\ptIdx, y_\ptIdx) \in (\setOfXCoords \times \finiteField)\}_{\ptIdx = 1}^{\numOfBinsRef + 2\hamDistThreshold + 1}\}$. Note that $\setSize{\universe_\fieldOrder} = \fieldOrder^{\numOfBinsRef + 2\hamDistThreshold + 1}$. Similar to the proof for \thmref{thm:ham_query_with_sampling}, we show that for any arbitrary pair of polynomials $\alicesPolyFromVec, \bobsPolyFromVec  \in \polyField \times \polyField$ such that $\degree{\gcd{P(x)}{Q(x)}} < \numOfBinsRef - 2\hamDistThreshold$ and $\forall x_\ptIdx \in \setOfXCoords, P(x_\ptIdx), Q(x_\ptIdx) \neq 0$,

	\begin{equation}
		\sum_{\randomPtSetA \in \universe_\fieldOrder} \left|\prob{\alicesVoleOutputPolyPtSetReal = \randomPtSetA} -  \prob{\alicesVoleOutputPolyPtSetSim = \randomPtSetA} \right| < 1/\fieldOrder
	\end{equation}
	
	W.l.o.g for some  $\randomPtSetA \in \universe_\fieldOrder$, $\alicesVoleOutputPolyPtSetReal = \randomPtSetA$, $x_\ptIdx \in X$ and $v_\ptIdx \in \randomPtSetA$, we have 
	
	\begin{center}
		$\tPSIOutputPolyEval{x_\ptIdx} = \gcdPolyEval{x_\ptIdx}R(x_\ptIdx) = v_{\ptIdx} $ 
	\end{center}
	
	where $\gcdPoly = \gcd{\alicesPoly}{\bobsPoly}$, $\degreeOfGCD \assign \degree{\gcdPoly}$ and $R(x)$ is a random polynomial of degree $2\numOfBinsRef - \degreeOfGCD$ due to \lemmaref{lemma:kissener_uniformly_random}. There are $\fieldOrder^{(2\numOfBinsRef - \degreeOfGCD + 1) - (\numOfBinsRef + 2\hamDistThreshold + 1)} = \fieldOrder^{\numOfBinsRef - \degreeOfGCD + 2\hamDistThreshold}$ polynomials that are consistent with $R(x)$. Let $S_R$ be the set of these polynomials. Each $R(x) \in S_R$ is consistent with a unique pair $(\alicesRandomPoly, \bobsRandomPoly) \in \polyField \times \polyField$ of degree-$\numOfBinsRef$ polynomials selected by the protocol. Thus, for any $\alicesPolyFromVec, \bobsPolyFromVec$, we have

	\begin{center}
		$\prob{\alicesVoleOutputPolyPtSetReal = \randomPtSetA} = \frac{\setSize{S_R}}{\text{\# of } (\alicesRandomPoly, \bobsRandomPoly) \text{pairs} } =  \frac{\fieldOrder^{\numOfBinsRef - \degreeOfGCD + 2\hamDistThreshold}}{\fieldOrder^{2\numOfBinsRef + 2}} = \fieldOrder^{-(\numOfBinsRef  - \degreeOfGCD + 2\hamDistThreshold +  2)}$
	\end{center}

	For $\alicesVoleOutputPolyPtSetSim = \randomPtSetA$, we require 
	
	\begin{center}
		$\tPSIOutputPolySimEval{x_\ptIdx} = v_{\ptIdx} - \finiteFieldPRF(\prfKey, \ptIdx)$ 
	\end{center}

	For \bobsPolyFromVecSim, we have with overwhelming probability $\degreeOfGCD = \degree{\gcd{\alicesPolyFromVec}{\bobsPolyFromVecSim}} = 0$. Independent of the value of $\finiteField(\prfKey, \ptIdx)$, we have $\prob{\alicesVoleOutputPolyPtSetSim = \randomPtSetA} = \fieldOrder^{-(\numOfBinsRef  + 2\hamDistThreshold +  2)}$. From the above, we have

	\begin{center}
		$\sum_{\randomPtSetA \in \universe_\fieldOrder} \left|\prob{\alicesVoleOutputPolyPtSetReal = \randomPtSetA} -  \prob{\alicesVoleOutputPolyPtSetSim = \randomPtSetA} \right| < 1/\fieldOrder$ \\
		$\sum_{\randomPtSetA \in \universe_\fieldOrder} \fieldOrder^{-(\numOfBinsRef  + 2\hamDistThreshold +  2)}(1 - 1/\fieldOrder^{\degreeOfGCD})$
		$\setSize{\universe_\fieldOrder} \times \fieldOrder^{-(\numOfBinsRef  	+ 2\hamDistThreshold +  2)} (1 - 1/\fieldOrder^{\degreeOfGCD}) < 1/\fieldOrder$
		
	\end{center}
	
\end{proof}

\tikzset{every picture/.style={line width=0.75pt}} 
\begin{figure}
\centering
\tikzset{every picture/.style={line width=0.75pt}} 

\begin{tikzpicture}[x=0.75pt,y=0.75pt,yscale=-1,xscale=0.5]

\draw (488,1) node [anchor=north west][inner sep=0.75pt]    {\footnotesize \underline{\bob}};
\draw (138,1) node [anchor=north west][inner sep=0.75pt]    {\footnotesize \underline{\alice}};
\draw (388,20) node [anchor=north west][inner sep=0.75pt]    {\footnotesize $R_{11}(x), \ldots, R_{1n}(x) \getsr \polyField $};
\draw (388,34) node [anchor=north west][inner sep=0.75pt]    {\footnotesize $R_{21}(x), \ldots, R_{2n}(x) \getsr \polyField $};

\draw   (279,76) -- (373.5,76) -- (373.5,100) -- (279,100) -- cycle ;

\draw    (137.5,80) -- (271,80) ;
\draw [shift={(273,80)}, rotate = 180] [color={rgb, 255:red, 0; green, 0; blue, 0 }  ][line width=0.75]    (10.93,-3.29) .. controls (6.95,-1.4) and (3.31,-0.3) .. (0,0) .. controls (3.31,0.3) and (6.95,1.4) .. (10.93,3.29)   ;

\draw    (695.5,80) -- (379.5,80) ;
\draw [shift={(377.5,80)}, rotate = 359.74] [color={rgb, 255:red, 0; green, 0; blue, 0 }  ][line width=0.75]    (10.93,-3.29) .. controls (6.95,-1.4) and (3.31,-0.3) .. (0,0) .. controls (3.31,0.3) and (6.95,1.4) .. (10.93,3.29)   ;

\draw    (696.5,98) -- (380.5,98) ;
\draw [shift={(378.5,98)}, rotate = 359.74] [color={rgb, 255:red, 0; green, 0; blue, 0 }  ][line width=0.75]    (10.93,-3.29) .. controls (6.95,-1.4) and (3.31,-0.3) .. (0,0) .. controls (3.31,0.3) and (6.95,1.4) .. (10.93,3.29)   ;

\draw    (275.5,98) -- (98.5,98) ;
\draw [shift={(96.5,98)}, rotate = 0.32] [color={rgb, 255:red, 0; green, 0; blue, 0 }  ][line width=0.75]    (10.93,-3.29) .. controls (6.95,-1.4) and (3.31,-0.3) .. (0,0) .. controls (3.31,0.3) and (6.95,1.4) .. (10.93,3.29)   ;

\draw (300,80.4) node [anchor=north west][inner sep=0.75pt]    {$\voleFunc$};
\draw (388,64) node [anchor=north west][inner sep=0.75pt]    {\footnotesize $< \ R_{11}( x_{1}) ,\ ..\ ,\ R_{1n}( x_{1})  >\ $};

\draw (388,103.4) node [anchor=north west][inner sep=0.75pt]    {\footnotesize $<\bobsRandomPolysIdxEval{1}{x_1} \bobsPolyFromVecIdxEval{1}{x_1}, \ldots, \bobsRandomPolysIdxEval{\genericSetSize}{x_1} \bobsPolyFromVecIdxEval{\genericSetSize}{x_1}> $};

\draw (184,64) node [anchor=north west][inner sep=0.75pt]    {\footnotesize $P( x_{1})$};

\draw (99,103.4) node [anchor=north west][inner sep=0.75pt]    {\footnotesize $\{W_{A1}( x_{1}) ,\ ..\ ,\ W_{An}( x_{1})\}$};

\draw (316,120) node [anchor=north west][inner sep=0.75pt]   [align=left] {...};
\draw (316,130) node [anchor=north west][inner sep=0.75pt]   [align=left] {...};

\draw   (279,146) -- (373.5,146) -- (373.5,170) -- (279,170) -- cycle ;

\draw    (137.5,150) -- (271,150) ;
\draw [shift={(273,150)}, rotate = 180] [color={rgb, 255:red, 0; green, 0; blue, 0 }  ][line width=0.75]    (10.93,-3.29) .. controls (6.95,-1.4) and (3.31,-0.3) .. (0,0) .. controls (3.31,0.3) and (6.95,1.4) .. (10.93,3.29)   ;

\draw    (695.5,150) -- (379.5,150) ;
\draw [shift={(377.5,150)}, rotate = 359.74] [color={rgb, 255:red, 0; green, 0; blue, 0 }  ][line width=0.75]    (10.93,-3.29) .. controls (6.95,-1.4) and (3.31,-0.3) .. (0,0) .. controls (3.31,0.3) and (6.95,1.4) .. (10.93,3.29)   ;

\draw    (696.5,168) -- (380.5,168) ;
\draw [shift={(378.5,168)}, rotate = 359.74] [color={rgb, 255:red, 0; green, 0; blue, 0 }  ][line width=0.75]    (10.93,-3.29) .. controls (6.95,-1.4) and (3.31,-0.3) .. (0,0) .. controls (3.31,0.3) and (6.95,1.4) .. (10.93,3.29)   ;

\draw    (275.5,168) -- (98.5,168) ;
\draw [shift={(96.5,168)}, rotate = 0.32] [color={rgb, 255:red, 0; green, 0; blue, 0 }  ][line width=0.75]    (10.93,-3.29) .. controls (6.95,-1.4) and (3.31,-0.3) .. (0,0) .. controls (3.31,0.3) and (6.95,1.4) .. (10.93,3.29)   ;

\draw (300,150.4) node [anchor=north west][inner sep=0.75pt]    {$\voleFunc$};
\draw (388,134) node [anchor=north west][inner sep=0.75pt]    {\footnotesize $< \ R_{\ptIdx 1} ( x_{\ptIdx}) ,\ ..\ ,\ R_{\ptIdx n}( x_{\ptIdx})  >\ $};
\draw (388,173.4) node [anchor=north west][inner sep=0.75pt]    {\footnotesize $<\bobsRandomPolysIdxEval{1}{x_\ptIdx} \bobsPolyFromVecIdxEval{1}{x_\ptIdx}, \ldots, \bobsRandomPolysIdxEval{\genericSetSize}{x_\ptIdx} \bobsPolyFromVecIdxEval{\genericSetSize}{x_\ptIdx}> $};
\draw (184,134) node [anchor=north west][inner sep=0.75pt]    {\footnotesize $P( x_{\ptIdx})$};
\draw (99,173.4) node [anchor=north west][inner sep=0.75pt]    {\footnotesize $\{W_{A1}( x_{\ptIdx}) ,\ ..\ ,\ W_{An}( x_{\ptIdx})\}$};




\end{tikzpicture}
\vspace{-10pt}
\caption{Using VOLE for \oneSidedSetRecon \label{fig:vole_in_set_recon}}
\end{figure}

\section{Proofs for \hamAwarePSIProtocol (\secref{sec:hammingPSI:protocol})}
\label{app:hamming_psi}

For inputs, $\alicesInputVecSet \assign \alicesInputVecSetDefn$ and $\bobsInputVecSet \assign \bobsInputVecSetDefn$, the straightforward way to realize a Hamming DA-PSI protocol is by running $\genericSetSize^2$ instances on \tHamQuery in parallel over each input pair $(\alicesVector_i, \bobsVector_j)$. However, there is more optimized solution using  vector OLE's (see \secref{sec:background}). The idea is as follows: consider \alice's input is $\alicesVector_i$ and \bob's input 
is the set of vectors \bobsInputVecSetDefn. Then, the goal is to determine 
if there exists $\bobsVector_j \in \bobsInputVecSet$ such that $\hamDist{\alicesVector_i}{\bobsVector_j} < \hamDistThreshold$. This functionality can be considered a \textit{one-sided containment query}.

To realize this functionality, we present \HamContainQuery protocol (see \figref{fig:ham_containment_query}), consisting of two procedures \permAndPart and \oneSidedSetReconMultiBlind. In \permAndPart, \alice and \bob run \tHamQueryRestricted over each pair $(\alicesVector_i, \bobsVector_j)$ (\linesref{step:ham_contain:key_select}{step:ham_contain:ham_restricted}). At the end of this procedure, \alice obtains $\{\keySet_1, \ldots, \keySet_\genericSetSize\}$. \alice obtains the set of sub-vectors 
$\alicesSetFromVecBins := \{\vec{x}_1, \ldots, \vec{x}_{\numOfBinsRef}\}$ derived from \alicesVector. Similarly, for $\setIdx \in [1, \genericSetSize]$, \bob obtains the set of sub-vectors \bobsSetFromVecBinsIdx{\setIdx} derived from \bobsVecIdx{\setIdx} (\lineref{step:ham_contain:compute_sets}).

\oneSidedSetReconMultiBlind takes as input \alicesSetReconInput from \alice while \bob's input is a set of sets \setOfBobsSetReconInputs. \alice derives \alicesPoly from 
\alicesSetReconInput and \bob derives the set of polynomials \bobsPolySet from 
\setOfBobsSetReconInputs (\lineref{step:ham_contain:compute_poly}). \bob selects two sets of degree-\genericVecLen random polynomials \setOfAlicesRandomPolys, \setOfBobsRandomPolys. Then, for $\ptIdx \in [1, \numOfBinsRef + 2\hamDistThreshold + 1]$, \alice and \bob compute \setOfAlicesVoleOutputPolys where $\alicesVoleOutputPolysIdxEval{\setIdx}{x_\ptIdx} \assign \alicesPolyFromVecEval{x_\ptIdx} \alicesRandomPolysIdxEval{\setIdx}{x_\ptIdx} + \bobsRandomPolysIdxEval{\setIdx}{x_\ptIdx} \bobsPolyFromVecIdxEval{\setIdx}{x_\ptIdx} + \finiteFieldPRF(\prfKey_\setIdx, \ptIdx)$. 

In order to compute \setOfAlicesVoleOutputPolys, the idea is to use a VOLE protocol (see \figref{fig:vole_in_set_recon}). Specifically, for $\ptIdx \in [1, \numPointsForInterpol]$ \alice sends \alicesPolyFromVecEval{x_\ptIdx} to \voleFunc, while \bob sends the two vectors $\vec{R} \assign <\alicesRandomPolysIdxEval{1}{x_\ptIdx}, \ldots, \alicesRandomPolysIdxEval{\genericSetSize}{x_\ptIdx}>$, and $\vec{U} \assign <\bobsRandomPolysIdxEval{1}{x_\ptIdx} \bobsPolyFromVecIdxEval{1}{x_\ptIdx} + \finiteFieldPRF(\prfKey_1, \ptIdx), \ldots, \bobsRandomPolysIdxEval{\genericSetSize}{x_\ptIdx} \bobsPolyFromVecIdxEval{\genericSetSize}{x_\ptIdx} + \finiteFieldPRF(\prfKey_\genericSetSize, \ptIdx)>$ (\lineref{step:ham_contain:vole_step_1}). By definition, \voleFunc return to \alice \setOfAlicesVoleOutputPolys. Batching $\genericSetSize$ OLE computations with a single VOLE instance has concrete advantages  both in terms of compute time and communication costs. \figref{fig:vole_in_set_recon} describes this process. Subsequently, for $\prfKey' \in \keySet_\setIdx$ and $\setIdx \in [1, \genericSetSize]$, \alice computes $\setOfPtsForInterpolation_{i\prfKey'} =  \{(x_\ptIdx, y_\ptIdx) : y_\ptIdx = \frac{\alicesVoleOutputPolyIdxEval{\setIdx}{x_\ptIdx} - \finiteFieldPRF(\prfKey', \ptIdx)}{\alicesPolyFromVecEval{x_\ptIdx}}, ~\ptIdx \in [1, \numOfBinsRef + 2\hamDistThreshold + 1]\}$, and interpolates the set of points (\lineref{step:ham_contain:interpolate}). 

Finally, realizing a Hamming PSI protocol is straightforward with a one-sided containment query protocol. Specifically, we run \genericSetSize instances of containment queries corresponding to each element in \alicesSet. Each instance is run independently with independent random coins (see \figref{fig:ham_aware_psi}).

\begin{figure}
	\small
	\begin{mdframed}
		
		\noindent
		\textbf{\underline{Parameters:}}  \alice holds vector \alicesVector and \bob has set of vectors \bobsInputVecSetDefn.
		They jointly select a Hamming distance threshold $\hamDistThreshold \in (0,\genericVecLen/2)$.

		\smallskip\noindent
		\textbf{\underline{Procedure \permAndPart:}}
		
		\begin{enumerate}[nosep,leftmargin=1.5em,labelwidth=*,align=left,labelsep=0.25em]	
			
			\item \bob samples \genericSetSize random keys \prfKeySet where \prfKeyIdx{\setIdx} \getsr \finiteField. \label{step:ham_contain:key_select}
			
			\item For each $\setIdx \in [1, \genericSetSize]$, \alice sends \alicesVector and \bob sends \bobsVecIdx{\setIdx} and \prfKeyIdx{\setIdx} to \tHamQueryRestricted. \alice obtains $\{\keySet_1, \ldots, \keySet_\genericSetSize\}$. 
			\label{step:ham_contain:ham_restricted}
			
			\item Let \numOfBinsRef \assign \numOfBins. \alice computes set $\alicesSetFromVecBins := \{\vec{x}_1, \ldots, \vec{x}_{\numOfBinsRef}\}$ using \randPerm similar to \lineref{step:hamQueryRes:createVecs} of \tHamQueryRestricted. \label{step:ham_query:comp_setst}. Similarly, for each $\setIdx \in [1, \genericSetSize]$, \bob computes set \bobsSetFromVecBinsIdx{\setIdx}. \label{step:ham_contain:compute_sets}
			
		\end{enumerate}

		\smallskip\noindent
		\textbf{\underline{Procedure \oneSidedSetReconMultiBlind:}}

		\begin{enumerate}[resume,nosep,leftmargin=1.5em,labelwidth=*,align=left,labelsep=0.25em]

			\item \alice computes the polynomial \alicesPolyFromVec \assign \polyFromSet{\alicesSetFromVec}. For each
			$\setIdx \in [1, \genericSetSize]$, \bob computes set \bobsPolyIdx{\setIdx} \assign \polyFromSet{\bobsSetFromVecBinsIdx{\setIdx}}. \label{step:ham_contain:compute_poly}
			
			\item \bob samples two sets of \genericSetSize polynomials \setOfAlicesRandomPolys, and \setOfBobsRandomPolys where each \alicesRandomPolysIdx{\setIdx}, \bobsRandomPolysIdx{\setIdx} \getsr \polyField are degree-(\numOfBinsRef) polynomials.
			
			\item For each $\ptIdx \in [1, \numPointsForInterpolAfterBins]$, \alice and \bob engage in a round \voleFunc. Specifically, \alice sends \alicesPolyFromVecEval{x_\ptIdx} to \voleFunc and
			\bob sends the vectors $<\alicesRandomPolysIdxEval{1}{x_\ptIdx}, \ldots, \alicesRandomPolysIdxEval{\genericSetSize}{x_\ptIdx}>$, and 
			$<\bobsRandomPolysIdxEval{1}{x_\ptIdx} \bobsPolyFromVecIdxEval{1}{x_\ptIdx} + \finiteFieldPRF(\prfKey_1, \ptIdx), \ldots, \bobsRandomPolysIdxEval{\genericSetSize}{x_\ptIdx} \bobsPolyFromVecIdxEval{\genericSetSize}{x_\ptIdx} + + \finiteFieldPRF(\prfKey_\genericSetSize, \ptIdx) >$.
			\alice obtains $\{\alicesVoleOutputPolyIdxEval{1}{x_\ptIdx}, \ldots, \alicesVoleOutputPolyIdxEval{\genericSetSize}{x_\ptIdx} \}$ where $\alicesVoleOutputPolyIdxEval{\setIdx}{x_\ptIdx} \assign \alicesRandomPolysIdxEval{\setIdx}{x_\ptIdx} \alicesPolyFromVecEval{x_\ptIdx} + \bobsRandomPolysIdxEval{\setIdx}{x_\ptIdx} \bobsPolyFromVecIdxEval{\setIdx}{x_\ptIdx} + \finiteFieldPRF(\prfKey_\setIdx, \ptIdx)$.  \label{step:ham_contain:vole_step_1}
			
			\item For each $\prfKey' \in \keySet_\setIdx$, \alice computes the set of sets $\{ \setOfPtsForInterpolation_1, \ldots, \setOfPtsForInterpolation_\genericSetSize\}$, where 
			
			\begin{center}
				$\setOfPtsForInterpolation_{\setIdx\prfKey'} =  \{(x_\ptIdx, y_\ptIdx) : y_\ptIdx = \frac{\alicesVoleOutputPolyIdxEval{\setIdx}{x_\ptIdx} - \finiteFieldPRF(\prfKey', \ptIdx)}{\alicesPolyFromVecEval{x_\ptIdx}}\}$.
			\end{center}
			
			\item \alice runs the interpolation step in \lineref{step:set_recon:interpolation} in  \figref{fig:ham_query_lite} over $\setOfPtsForInterpolation_{\setIdx\prfKey'}$ and returns $\{\setIdx, \prfKey'\}$ if the interpolation succeeds. Otherwise, \alice returns $\bot$ \label{step:ham_contain:interpolate}.
			
		\end{enumerate}
	\end{mdframed}
	\vspace{-10pt }
	\caption{\small \HamContainQuery: Hamming containment query protocol \label{fig:ham_containment_query}}
\end{figure}

\begin{figure}
	\small
	\begin{mdframed}
		
		\noindent
		\textbf{Inputs:~}  \alice holds a set of vectors \alicesInputVecSetDefn. and \bob has set of vectors \bobsInputVecSetDefn.
		They jointly select a Hamming distance threshold \hamDistThreshold.

		\textbf{Protocol:~}
		\begin{enumerate}[nosep,leftmargin=1.5em,labelwidth=*,align=left,labelsep=0.25em]	
			
			\item For each $\setIdx \in [1, \genericSetSize]$, \alice and \bob run \HamContainQuery with inputs \alicesVecIdx{\setIdx} and \bobsInputVecSetDefn, and 
			distance threshold \hamDistThreshold. 
			
			\item For $\setIdx \in [1, \genericSetSize]$,  if \HamContainQuery returns $\{j, \prfKey' \}$ for input \alicesVecIdx{\setIdx}, \alice sends to \bob $\{\alicesVecIdx{\setIdx}, j, \prfKey'\}$. \bob verifies i) \prfKey' is one of the chosen keys in \lineref{step:ham_contain:key_select}, and ii) $\hamDist{\alicesVecIdx{\setIdx}}{\bobsVecIdx{j}} < \hamDistThreshold$.
			If so, \bob returns \bobsVecIdx{j}. \alice add $(\alicesVecIdx{\setIdx}, \bobsVecIdx{j})$ to the set $S_{\mathsf{out}}$. 
			
			\item \alice outputs $S_{\mathsf{out}}$. 
			
		\end{enumerate}
	\end{mdframed}
	\vspace{-10pt }
	\caption{\small \hamAwarePSIProtocol: Hamming distance-aware PSI protocol \label{fig:ham_aware_psi}}
\end{figure}

\securityHamPSI*

\begin{proofsketch}
	The complexity of the protocol is straightforward. In \lineref{step:ham_contain:ham_restricted} of \figref{fig:ham_containment_query}, running \tHamQueryRestricted over a set of \genericSetSize vectors requires \bigO{\genericSetSize\cdot \frac{\hamDistThreshold^2}{\fpr}\cdot \lambda} bits of communication.  In \lineref{step:ham_contain:vole_step_1} of \figref{fig:ham_containment_query}, there are \numPointsForInterpolAfterBins calls to \voleFunc where \numOfBinsRef \assign \numOfBins with vectors sizes of \genericSetSize. The total communication cost of this step is \bigO{\genericSetSize \cdot \frac{\hamDistThreshold^2}{\fpr} \cdot \lambda} bits. In \lineref{step:ham_contain:vole_step_1} of \figref{fig:ham_containment_query}, the communication cost of this step is also \bigO{\genericSetSize \cdot \frac{\hamDistThreshold^2}{\fpr} \cdot \lambda} bits of communication. Thus, the overall communication cost of \HamContainQuery is \bigO{\genericSetSize \cdot \frac{\hamDistThreshold^2}{\fpr} \cdot \lambda}. Finally, \hamAwarePSIProtocol executes \genericSetSize instances of \HamContainQuery and thus the overall communication cost of \hamAwarePSIProtocol is \hamPSICommComplexity.

	To see why this construction is secure, consider that in \hamAwarePSIProtocol (\figref{fig:ham_aware_psi}), \genericSetSize instances of \HamContainQuery are run independently, with independent random coins etc. Thus,  it is enough to show that a single instance of \HamContainQuery is secure. Next note that \HamContainQuery primarily aggregates \genericSetSize instances of \tHamQuery. As shown in \thmref{thm:security_ham_query}, \tHamQuery securely realizes \tHamQueryFunc, and when \genericSetSize instances are run independently with independent random coins, we can reduce security to breaking the security of a single instance of \tHamQuery. \HamContainQuery also similarly runs \genericSetSize instances of \tHamQuery with random coins but instead of having $\genericSetSize  \times (\numPointsForInterpolAfterBins)$ OLE instances to generate the values of generating $\alicesVoleOutputPolyIdxEval{\setIdx}{x_\ptIdx}, \setIdx \in \genericSetSize, \ptIdx \in [1, \numPointsForInterpolAfterBins]$, it batches the $\genericSetSize$ OLEs $\{\alicesVoleOutputPolyIdxEval{1}{x_\ptIdx}, \ldots, \alicesVoleOutputPolyIdxEval{\genericSetSize}{x_\ptIdx} \}$ into a single VOLE instance, thereby running \numPointsForInterpolAfterBins VOLEs in total. These VOLE instances are run with independent random coins, and are therefore independent of each other. By definition, \voleFunc securely batches a set of OLE instances. Therefore, the existence of a protocol securely realizing \voleFunc ensures that the security of \HamContainQuery reduces to showing that \tHamQuery securely realizes \tHamQueryFunc.

\end{proofsketch}

\section{Proofs for \hamAwarePSIProtocolExp (\secref{sec:hamming_psi:exp_compute})}
\label{app:hamming_query_exp}

\begin{figure}[ht!]
	\small
	\begin{mdframed}
		
		\noindent
		\textbf{\underline{Parameters:}} \alice and \bob have vectors $\alicesVector, \bobsVector \in \{0,1\}^{\genericVecLen}$ respectively and a hamming distance threshold \hamDistThreshold.

		\smallskip\noindent
		\textbf{\underline{Procedure \subSample:}}
		
		\begin{enumerate}[nosep,leftmargin=1.5em,labelwidth=*,align=left,labelsep=0.25em]	
			
		\item \bob samples $T$ random masking functions $\{mask_1, \ldots, mask_T\}$ \cite{uzun2021fuzzy}, and keyed PRF, $PRF_{\kappa_s}$ e.g., AES with key $\kappa_s \assign \{0,1\}^{\secParam}$.

		\item \alice and \bob run a 2PC circuit where \alice's input is \alicesVector and \bob's input is $\{mask_1, \ldots, mask_T\}$, and $PRF_{\kappa_s}$.
		The circuit returns to \alice $\alicesSetFromVec \assign \{ PRF_{\kappa_s}(\alicesVector \bigwedge mask_1), \ldots,  PRF_{\kappa_s}(\alicesVector \bigwedge mask_T) \}$. \label{step:ham_query_sample:masking}
		
		\item \bob locally computes \bobsSetFromVec \assign $\{  PRF_{\kappa_s}(\bobsVector \bigwedge mask_1), \ldots,   PRF_{\kappa_s}(\bobsVector \bigwedge mask_T) \}$.
		\label{step:ham_query_sample:set_compute}
		
		\end{enumerate}

		\smallskip\noindent
		\textbf{\underline{Procedure \oneSidedSetReconExp:}} \alice and \bob run \oneSidedSetReconExp with \alicesSetFromVec and \bobsSetFromVec as inputs.

				\begin{enumerate}[resume,nosep,leftmargin=1.5em,labelwidth=*,align=left,labelsep=0.25em]

					\item  \alice and \bob select set of points $\setOfXCoords = \setDefn{2T - t + 2}$.
					
					\item \alice encodes \alicesSetFromVec in roots of the polynomial \alicesPoly = \polyFromSet{\alicesSetFromVec} and \bob encodes \bobsSetFromVec in roots of the polynomial \bobsPoly = \polyFromSet{\bobsSetFromVec}. \label{step:ham_query_sample:poly_compute}

					\item \bob samples two random polynomials $\alicesRandomPoly, \bobsRandomPoly \getsr \polyField$ of degree $T$. 
					
					\item For $\ptIdx \in [1, 2T - t + 2]$, \alice sends \alicesPolyFromVecEval{x_\ptIdx} to \oleFunc and \bob sends \alicesRandomPolyEval{x_\ptIdx} and \bobsRandomPolyEval{x_\ptIdx} \bobsPolyFromVecEval{x_\ptIdx}. \alice learns \alicesVoleOutputPolyEval{x_\ptIdx} \assign \tPSIOutputPolyEval{x_\ptIdx} as the output of \oleFunc. \label{step:ham_query_sample:ole_step_1}

					\item \alice computes set $C_{\mathsf{sub}}$ where each $C_{i} \in C_{\mathsf{sub}}$ is a subset of \alicesSetFromVec having exactly $t$ elements. \label{step:ham_query_sample:set_computation}
					
					\item For each $C_i \in C_{\mathsf{sub}}$, \alice computes the polynomial $P_i \assign \polyFromSet{C_i}$ and computes the set of points \label{step:ham_query_sample:point_computation}

					\begin{center}
						$\setOfPtsForInterpolation_i =  \{(x_\ptIdx, y_\ptIdx) : y_\ptIdx = \frac{\alicesVoleOutputPolyEval{x_\ptIdx}}{P_i(x_\ptIdx)}, ~\ptIdx \in [1, 2T - t + 2] \}$.
					\end{center}
					
					\item \alice interpolates $V_i$ for all $i \in |C_{\mathsf{sub}}|$ with a polynomial. If the degree of the interpolating polynomial is $< 2T - t$, \alice learns that \alicesVector and \bobsVector are close. Otherwise, \alice outputs $\bot$.	\label{step:ham_query_sample:interpolation}  
					
				\end{enumerate}

	\noindent\smallskip
	\textbf{\underline{Procedure $\mathsf{Recover}$}:}
	
	If \oneSidedSetReconExp outputs \setDiffCard{\alicesSetFromVec}{\bobsSetFromVec} then obtain \bobsVector from \bob. Output $(\alicesVector, \bobsVector)$. Otherwise, output $\bot$.

%
	\end{mdframed}
	\vspace{-10pt}
	\caption{\hamQuerySample: Threshold Hamming queries with sampling \label{fig:ham_query_sampling}}
\end{figure}

\begin{figure}
	\small
	\begin{mdframed}
		
		\noindent
		\textbf{Inputs:~}  \alice holds a set of vectors \alicesInputVecSetDefn. and \bob has set of vectors \bobsInputVecSetDefn.
		They jointly select a Hamming distance threshold \hamDistThreshold.

		\textbf{Protocol:~}
		\begin{enumerate}[nosep,leftmargin=1.5em,labelwidth=*,align=left,labelsep=0.25em]	
			
			\item For each $\setIdx \in [1, \genericSetSize], j \in [1, \genericSetSize]$, \alice and \bob run \hamQuerySample with inputs \alicesVecIdx{\setIdx} and \bobsVecIdx{j}, and 
			distance threshold \hamDistThreshold. 
			
			\item For $\setIdx \in [1, \genericSetSize], ~j \in [1, \genericSetSize]$,  if \hamQuerySample returns $(\alicesVecIdx{\setIdx}, \bobsVecIdx{j})$, then add it $S_\mathsf{out}$.

			\item \alice outputs $S_{\mathsf{out}}$. 
			
		\end{enumerate}
	\end{mdframed}
	\vspace{-10pt }
	\caption{\small \hamAwarePSIProtocolExp: Hamming distance-aware PSI protocol from \hamQuerySample \label{fig:ham_aware_psi_exp}}
	\end{figure}

This section details a construction for Hamming queries which combines \oneSidedSetReconExp with a sub-sampling algorithm which makes the computation feasible. 

\myparagraph{Sub-Sampling}
There is extensive work on reducing bit vectors to sets of small sizes for Hamming distance comparisons in specific application settings e.g., in biometric authentication \cite{uzun2021cryptographic,uzun2021fuzzy}. The idea is to sub-sample the 
bit vectors of length \genericVecLen into $T$ sub-vectors, and then compare the sets of the sub-vectors. If the input bit vectors are close in Hamming space, then $t$ out of the $T$ sub-vectors will match across the sets. The sub-sampling scheme is parameterized such that $t \ll T < \genericVecLen$. For example, in the context of biometric data represented by bit vectors of length $\genericVecLen = 256$ after a transformation with a locality-sensitive hash, the sub-sampling algorithm can yield sets of size $T = 64$ as inputs, with $t = 2$. 

In the typical setting where \alice and \bob hold bit vectors, \alicesVector and \bobsVector of length \genericVecLen, \bob samples $T$ "masking" functions, $\{mask_1, \ldots, mask_T\}$. These masking functions are applied to \alicesVector using a 2PC circuit (e.g., with a garbled circuit or an OPRF) to create a set of sub-vectors \alicesSetFromVec \assign $\{PRF_{\kappa_s}(\alicesVector \bigwedge mask_1), \ldots,  PRF_{\kappa_s}(\alicesVector \bigwedge mask_T)\}$. Here, $PRF$ is a keyed PRF e.g., AES with $\kappa_s$ as the key, uniformly sampled by \bob. In this way, \alice does not learn the masking functions but learns \alicesSetFromVec as output of the 2PC circuit. \bob similarly applies the masking functions to his own inputs, \bobsSetFromVec \assign $\{PRF_{\kappa_s}(\bobsVector \bigwedge mask_1), \ldots,  PRF_{\kappa_s}(\bobsVector \bigwedge mask_T) \}$. Next, \alice and \bob run a threshold PSI (t-out-of-T matching) algorithm over \alicesSetFromVec and \bobsSetFromVec. If $t$ elements match in \alicesSetFromVec and \bobsSetFromVec, \alice learns that \alicesVector and \bobsVector are close in context of the application. We refer to existing work \cite{uzun2021cryptographic,dodis2004fuzzy,canetti2021reusable,uzun2021fuzzy} for further details, and focus on the threshold PSI part of this process. \figref{fig:ham_query_sampling} describes the protocol. After creating the sets by sub-sampling the input vectors in \lineref{step:ham_query_sample:masking} of \figref{fig:ham_query_sampling}, \alice and \bob run \oneSidedSetReconExp with $\genericVecLen = T, \hamDistThreshold = (T - t)$.

\myparagraph{Optimizing the Interpolation}
In \lineref{step:ham_query_sample:interpolation} of \figref{fig:ham_query_sampling}, the interpolation step can be optimized based on two insights. The first insight is based on the fact that we are only interested in the degree of the interpolating polynomial \textit{and not the polynomial itself}. Therefore, instead of computing the interpolating polynomial completely we can compute the coefficients of the constituent monomials of Newton's interpolating polynomial. More specifically, the interpolating polynomial is of the form 

\begin{center}
	$R(x) := f[x_1] + f[x_1, x_2] (x - x_1) + \ldots + f[x_1, \ldots, x_{2T - t + 2}](x - x_1)(x - x_2)\ldots(x - x_{2T-t + 2})$
\end{center}

Here, $f[x_1], f[x_1, x_2] \ldots$ are the Newton's divided differences computed from the points in $\setOfPtsForInterpolation_i$. Since, we only need the coefficient of $x^{2T -t + 1}$ in this polynomial, it suffices to compute the value of $f[x_1, \ldots, x_{2T - t + 2}]$. If $f[x_1, \ldots, x_{2T - t + 2}] = 0$, the degree of $R(x)$ is $< 2T - t + 1$. 

Consider the divided differences:

\begin{enumerate}
	\item $f'[x_1], f[x_1, x_2], \ldots, f'[x_1, \ldots, x_{2T - t + 2}]$, over the points  $\{(x_\ptIdx, y_\ptIdx) : y_\ptIdx \assign \alicesVoleOutputPolyEval{x_\ptIdx},  ~\ptIdx \in [1, 2T - t + 2]\}$
	\item $g[x_1], g[x_1, x_2], \ldots, g[x_1, \ldots, x_{2T - t + 2}]$ over the points $\{(x_\ptIdx, y_\ptIdx) : y_\ptIdx \assign \frac{1}{P_i(x_\ptIdx)},  ~\ptIdx \in [1, 2T - t + 2]\}$
\end{enumerate}

Since $R(x) = \alicesVoleOutputPoly \times \frac{1}{P_i(x)}$, the divided differences for $R(x)$ can be computed using the divided 
differences for $\alicesVoleOutputPoly$ and $\frac{1}{P_i(x)}$ due to Leibniz's rule \cite{leibniz_rule}. That is, 

\begin{multline}
	\label{eqn:divided_diff}
	f[x_1, \ldots, x_{2T - t + 2}] = (f'.g)[x_1, \ldots, x_{2T - t + 2}] = \\ 
	\sum\limits_{k = 1}^{2T-t+2} f'[x_1, \ldots, x_k] g[x_k, \ldots, x_{2T-t+2}]
\end{multline}

The second insight is based on the fact that the divided differences for all the polynomials $P_i(x)$'s tested in \linesref{step:ham_query_sample:point_computation}{step:ham_query_sample:interpolation} can be pre-computed offline by \alice and stored for speeding up the computation in \eqref{eqn:divided_diff}. Specifically, \alice computes for each polynomial $P_i(x)$ (generated from $C_i \in C_{\mathsf{sub}}$, the set comprising the divided differences $\{g[x_1, \ldots, x_{2T-t+2}], g[x_2, \ldots, x_{2T-t+2}], \ldots, g[x_{2T-t+2}]\}$. This is done offline.  Subsequently, when \alice obtains the points for \alicesVoleOutputPoly after executing \lineref{step:ham_query_sample:ole_step_1}, she computes the divided differences for 
$\alicesVoleOutputPoly \times \frac{1}{P_i(x)}$ using the \eqref{eqn:divided_diff}. This significantly speeds up the compute times.

\begin{thm}
	\label{thm:ham_query_with_sampling}
	Assuming that there is a protocol securely realizing \oleFunc with communication cost scaling linearly with the vector size and a sub-sampling algorithm which derives sets of size $T$ after sub-sampling binary vectors of length \genericVecLen such that sets corresponding to the vectors close in Hamming space have at least $t$ common elements, \hamQuerySample securely realizes \tHamQueryFunc with \bigO{T} communication cost and \bigO{\allPossibleComb{T}{t}} compute costs. 
	
\end{thm}

\begin{proof}
	To prove this, we will show that there is a PPT simulator \simm in the ideal world which indistinguishable simulates the real world execution of \hamQuerySample.

	\myparagraph{Simulating \bob's view:}
	\bob does not receive any output from the protocol and only observes intermediate results from \oleFunc. Assuming that \oleFunc is realized by a protocol which can be indistinguishably simulated, \bob's view in \hamQuerySample 
	may be simulated by \simm. 
	
	\myparagraph{Simulating \alice's view:}
	\alice's input into the protocol is \alicesVector and \alice receives the evaluations of the polynomials $\alicesVoleOutputPoly \assign \tPSIOutputPoly$. 
	The simulation strategy is as follows. 
	
	\underline{\textit{When $\hamDist{\alicesVector}{\bobsVector} < T - t$}:} In this case, \tHamQueryFunc returns \bobsVector. So, \simm may indistinguishably simulate \hamQuerySample by running the steps of the protocol with \bobsVector as input. 
	
	\underline{\textit{When $\hamDist{\alicesVector}{\bobsVector} \geq T - t$}:} In this case, \tHamQueryFunc returns $\bot$. To simulate \alice's view, \simm sets $\bobsSetFromVec \assign \{b_1, \ldots, b_T\}$ where $b_i \getsr \{0,1\}^{\secParam}$. In effect, $\bobsPolyFromVecSim = \polyFromSet{\bobsSetFromVec} = \prod\limits_{x_i \getsr \finiteField} (x - x_i)$ is a degree-$T$ polynomial with random roots in $\finiteField$. 
	
	\simm follows the rest of the steps \hamQuerySample by sampling two degree-$T$ random polynomials $\alicesRandomPoly, \bobsRandomPoly \getsr \polyField$. In the real world, \alice receives evaluations of $\alicesVoleOutputPoly \assign \tPSIOutputPoly$ from \oleFunc. In the ideal world, \simm sends evaluations of \alicesPolyFromVec, \alicesRandomPoly and $\bobsRandomPoly \bobsPolyFromVecSim$ to \oleFunc. \simm sends the output of \oleFunc, $\alicesVoleOutputPolySim \assign \alicesRandomPoly \alicesPolyFromVec + \bobsRandomPoly \bobsPolyFromVecSim$ to \alice.

	Let $\alicesVoleOutputPolyPtSetReal \assign \{(x_\ptIdx, y_\ptIdx) : y_\ptIdx \assign \alicesVoleOutputPolyEval{x_\ptIdx}, ~ x_\ptIdx \in \setOfXCoords\}$,  and 
	$\alicesVoleOutputPolyPtSetSim \assign \{(x_\ptIdx, y_\ptIdx) : y_\ptIdx \assign \alicesVoleOutputPolySimEval{x_\ptIdx}, ~ x_\ptIdx \in \setOfXCoords\}$. Let $\universe_\fieldOrder$ be the set of all sets $\{(x_\ptIdx, y_\ptIdx) : (x_\ptIdx, y_\ptIdx) \in (\setOfXCoords \times \finiteField)\}_{\ptIdx = 1}^{2T -t +2}\}$. Note that $\setSize{\universe_\fieldOrder} = \fieldOrder^{2T - t + 2}$
	
	We show that for any arbitrary pair of polynomials $\alicesPolyFromVec, \bobsPolyFromVec  \in \polyField \times \polyField$ such that $\degree{\gcd{P(x)}{Q(x)}} < 2T - t - (T-t) =  t$ and $\forall x_\ptIdx \in \setOfXCoords, P(x_\ptIdx), Q(x_\ptIdx) \neq 0$,

	\begin{equation}
		\sum_{\randomPtSetA \in \universe_\fieldOrder} \left|\prob{\alicesVoleOutputPolyPtSetReal = \randomPtSetA} -  \prob{\alicesVoleOutputPolyPtSetSim = \randomPtSetA} \right| < 1/\fieldOrder
	\end{equation}
	
	W.l.o.g for some  $\randomPtSetA \in \universe_\fieldOrder$, $\alicesVoleOutputPolyPtSetReal = \randomPtSetA$, $x_\ptIdx \in X$ and $v_\ptIdx \in \randomPtSetA$, we have 
	
	\begin{center}
		$\tPSIOutputPolyEval{x_\ptIdx} = \gcdPolyEval{x_\ptIdx}R(x_\ptIdx) = v_{\ptIdx} $ 
	\end{center}
	
	where $\gcdPoly = \gcd{\alicesPoly}{\bobsPoly}$, $\degreeOfGCD \assign \degree{\gcdPoly}$ and $R(x)$ is a random polynomial of degree $2T - \degreeOfGCD$ due to \lemmaref{lemma:kissener_uniformly_random}. There are $\fieldOrder^{(2T -  \degreeOfGCD + 1) - (2T - t + 2)} = \fieldOrder^{t - \degreeOfGCD - 1}$ polynomials that are consistent with $R(x)$. Let $S_R$ be the set of these polynomials. Each $R(x) \in S_R$ is consistent with a unique pair $(\alicesRandomPoly, \bobsRandomPoly) \in \polyField \times \polyField$ of degree-$T$ polynomials selected by the protocol. Thus, for any $\alicesPolyFromVec, \bobsPolyFromVec$, we have

	\begin{center}
		$\prob{\alicesVoleOutputPolyPtSetReal = \randomPtSetA} = \frac{\setSize{S_R}}{\text{\# of } (\alicesRandomPoly, \bobsRandomPoly) \text{pairs} } =  \frac{\fieldOrder^{(t - \degreeOfGCD - 1)}}{\fieldOrder^{2T + 2}} = \fieldOrder^{-(2T - t + \degreeOfGCD + 3)}$
	\end{center}

	For \bobsPolyFromVecSim, we have with overwhelming probability $\degreeOfGCD = \degree{\gcd{\alicesPolyFromVec}{\bobsPolyFromVecSim}} = 0$. Thus, $\prob{\alicesVoleOutputPolyPtSetSim = \randomPtSetA} = \fieldOrder^{-(2T - t + 3)}$. From the above, we have

	\begin{center}
		$\sum_{\randomPtSetA \in \universe_\fieldOrder} \left|\prob{\alicesVoleOutputPolyPtSetReal = \randomPtSetA} -  \prob{\alicesVoleOutputPolyPtSetSim = \randomPtSetA} \right| < 1/\fieldOrder$ \\
		$\sum_{\randomPtSetA \in \universe_\fieldOrder} \fieldOrder^{-(2T -t + 3)} (1 - 1/\fieldOrder^{\degreeOfGCD})$
		$\setSize{\universe_\fieldOrder} \times \fieldOrder^{-(2T - t + 3)} (1 - 1/\fieldOrder^{\degreeOfGCD}) < 1/\fieldOrder$
			$ \frac{1}{\fieldOrder} \times (1 - 1/\fieldOrder^{\degreeOfGCD}) < 1/\fieldOrder$
	\end{center}
\end{proof}

\section{Proofs for Integer Distance-Aware PSI}
\label{app:intPSIProofs}

\begin{algorithm}
	\caption{Distance-Aware Set Augmentation}\label{alg:set_augmentation}
	\begin{algorithmic}[1]
		\footnotesize
		
		\Input
		\Desc{\alicesSet}{\alice's input set}
		\Desc{\bobsSet}{\bob's input set}
		\Desc{\intDistThreshold}{distance threshold}
		\EndInput
		\Output
		\Desc{Augmented sets \augAlicesSet, \augBobsSet}
		\EndOutput
		\Procedure{Generate Representative Strings for \alicesSet}{}
		\For{each \elemInSet{\alicesIthInt}{\alicesSet}} 
		\For {each integer, $a'_i \in (a_i - d, a_i + d)$}
		\State $s'_i :=$ Bit string of \maxBitLength representing $a'_i$
		\EndFor
		\State $P :=$ prefix trie with $s'_1, \ldots, s'_{2\intDistThreshold - 1}$ 
		\For{each {\em maximal enclosing complete subtrie} $T$ in $P$}
		\State $\phi :=$ prefix of $T$ in $P$
		\State Append ``don't care'' bits ($\ast$) to $\phi$ up to \maxBitLength
		\State Add $\phi$ to \augAlicesSet
		\EndFor
		\EndFor
		\State \Return \augAlicesSet
		\EndProcedure
		\Procedure{Generate Representative Strings for \bobsSet}{}
		\For{each \elemInSet{\bobsJthInt}{\bobsSet}}
		\State $s_j :=$ Bit string of \maxBitLength representing $b_j$ 
		\item // Let $s_j =  s_j[\maxBitLength - 1] \ldots s_j[0]$, $s_j[i]$ is the $i$th bit of $s_j$
		\For{$i = 0, 1, \ldots, \flooredLog{(2d - 1)}$}
		\State $s'_j := s_j[\maxBitLength - 1]\ldots s_j[i]  \ast \ldots \ast$   
		\State Add $s'_j$ to \augBobsSet  
		\EndFor
		\EndFor
		\State \Return \augBobsSet
		\EndProcedure
	\end{algorithmic}
\end{algorithm}

\myparagraph{Algorithm for Augmenting Sets}
The algorithm (Algorithm \ref{alg:set_augmentation}) has two procedures corresponding to the processes of augmenting \alice's input \alicesSet and \bob's input \bobsSet. 
To generate representative strings for each \elemInSet{\alicesIthInt}{\alicesSet}, the algorithm first builds a prefix trie over the bit strings corresponding to the binary representations (of length \maxBitLength) of integers in $(a_i - \intDistThreshold, a_i + \intDistThreshold)$ (Steps 8 - 11). As usual, each bit string corresponds to a path in the trie and the strings that share a prefix intersect at some level of the trie. The leaf nodes contain the least significant bits of the bit strings (see \figref{fig:prefixTrie42}).


Then, the algorithm determines the \mec subtries in the prefix trie (see Definition \ref{defn:mec}). A \mec subtrie essentially contains leaves (and its ancestors up to the root of the subtrie) corresponding to the integers that share an enclosing common prefix. \figref{fig:prefixTrie42} shows a prefix trie built over integers in the range [42,55]. Note that a leaf node in itself can be a \mec subtrie when the node is not part of any complete subtree. The algorithm 
identifies all the \mec subtries in the prefix trie corresponding to each \elemInSet{\alicesIthInt}{\alicesSet}. The prefix of each \mec subtrie is  is appended with wildcard bits up to the maximum bit length to form a representative string (Steps 12 - 15). The  representative strings are added to the augmented set, \augAlicesSet. 

For each \elemInSet{b_j}{\bobsSet}, the algorithm generates representative strings by progressively replacing the least significant bits in the binary representation of $b$ with wildcard bits. Specifically, the first representative string is generated by replacing the least significant bit, the second string is generated by replacing the last two least signficant bits and so on. The process is repeated \flooredLog{2d-1}+1 times until a bit string is obtained by replacing the last \flooredLog{2d-1}+1 least significant bits . These strings are added to the augmented set \augBobsSet (Steps 17 - 23).

\myparagraph{Correctness}
As discussed earlier, a non-null intersection between the augmented sets implies that there is some $(\alicesIthInt, \bobsJthInt)$ pair such that \intDistanceLessThanThreshold{\alicesIthInt}{\bobsJthInt}. To see why, consider two integer inputs $a, b$ and a distance threshold \intDistThreshold.
The following facts are ensured by design: 

\begin{enumerate}[nosep,leftmargin=1.6em,labelwidth=*,align=left]
	\item {\bf Fact 1:} If \intDistanceLessThanThreshold{\alicesInt}{\bobsInt} then the binary representation of $b$ is in the prefix trie built over all integers  \integersInRange.
	\item {\bf Fact 2:} The set of all \mec subtries together spans all the leaf nodes in the prefix trie.
	\item {\bf Fact 3:} The height of the largest \mec subtrie in a prefix trie built over integers in range \integersInRange is $\leq \flooredLog{(2d-1)}$.
\end{enumerate}

Observe that Fact 3 is true because if there existed a \mec subtrie with height $\flooredLog{(2d - 1)}$ then this would span over $2^{(1 + \floor{\log (2d - 1)})}  > 2d - 1$ leaf nodes which contradicts the fact that the prefix trie is built over $2d - 1$ strings corresponding to the $2d - 1$ integers in \integersInRange.

Fact 1 and Fact 2 together ensure that the the bit string corresponding to $b$ will share a prefix with a \mec subtrie in the prefix trie built over integers. \integersInRange. Fact 3 determines the prefix lengths to be checked i.e., since the height of the largest \mec subtrie is $\flooredLog{(2d - 1)}$, it is enough to check if a prefix of the bit string for $b$ matches the prefix of any \mec subtrie of height $[0, \flooredLog{(2d - 1)}]$.

Note that since the bit string corresponding to an integer $b', |b' - a| > d$ is not in the prefix trie, it will not share a prefix with any of the \mec subtries. This is because the range of integers used to build the trie also determines is \mec subtries. 

\myparagraph{Number of Representative Strings}
The following result shows that the number of representative strings  for integers in the range \integersInRange as a function of the distance parameter, \intDistThreshold is \bigO{\log \intDistThreshold}. This is analyzed by counting the number of \mec subtries in the prefix trie built over the binary representations of integers \integersInRange since each such subtrie corresponds to an enclosing common prefix.

\begin{thm}
	The total number of maximal enclosing complete subtries in a prefix trie built over the binary representations of all integers \integersInRange is $\bigO{\log \intDistThreshold}$.
\end{thm}

\begin{proof}
	We prove the result for the range of integers $[a, a+d)$. Due to symmetry, the exact same arguments holds for the range $[a-d + 1, a)$. The main idea behind the proof is to partition integers in the range $(a, a+d)$ into two ranges based on their higher order prefix. Specifically, let the bit-string $a_1 a_2 \ldots a_k$ be the binary representation of $a$.   The string is divided into two parts based on a {\em pivot} = \flooredLog{d}.
	The higher order bits together constitute the bit string $a_1 a_2 \ldots a_{pivot}$. Let \integerFromHigherOrderString{a} denote the integer whose binary representation matches this string. Similarly, the lower order bits constitute the bit string $a_{pivot+1} \ldots a_{0}$ and let \integerFromLowerOrderString{a} denote the integer whose binary representation matches this string. 
	
	We first observe that for any \integersInNeighborhood, \integerFromHigherOrderString{a'} = \integerFromHigherOrderString{a} or \integerFromHigherOrderString{a'} = \integerFromHigherOrderString{a} + 1. For any such $a'$, with \integerFromHigherOrderString{a'} = \integerFromHigherOrderString{a}, it must be that $\integerFromLowerOrderString{a} \leq \integerFromLowerOrderString{a'} \leq \poweredfloorlog{d}$ - 1 to ensure that $a' - a \leq d$. Similarly, for any $a'$, with \integerFromHigherOrderString{a'} = \integerFromHigherOrderString{a} + 1, it must be that $0 \leq \integerFromLowerOrderString{a'} \leq d - \poweredfloorlog{d}$.  Thus, the prefix trie \trie{a}{a+\intDistThreshold} built over the binary representations of all $a', a' - a \leq d$ contains two subtries with prefixes \integerFromHigherOrderString{a} and \integerFromHigherOrderString{a}+1 respectively. These subtries are \trie{\integerFromLowerOrderString{a}}{\poweredfloorlog{d} - 1} and \trie{0}{\intDistThreshold - \poweredfloorlog{d}} (see \figref{fig:prefixTrieSplit}).

	\begin{center}
		\numOfMaximalSubtriesinTree{ \trie{a}{a+\intDistThreshold}} = \numOfMaximalSubtriesinTree{\trie{\integerFromLowerOrderString{a}}{\poweredfloorlog{\intDistThreshold}} } + \numOfMaximalSubtriesinTree{\trie{0}{\intDistThreshold - \poweredfloorlog{\intDistThreshold}}}
	\end{center}
	
	From Lemmas \ref{lemma:number_enclosing_trees_more_than} and \ref{lemma:number_enclosing_trees_less_than}, we get  \numOfMaximalSubtriesinTree{\trie{\integerFromLowerOrderString{a}}{\poweredfloorlog{\intDistThreshold}}} = \bigO{\log \intDistThreshold} and \numOfMaximalSubtriesinTree{\trie{0}{\intDistThreshold - \poweredfloorlog{\intDistThreshold}}} = \bigO{\log \intDistThreshold}. This completes the proof.
\end{proof}

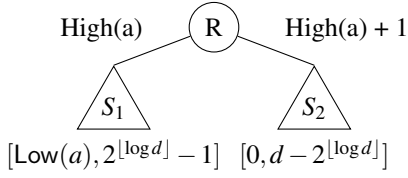
\begin{figure}
	\centering
	 \begin{tikzpicture}[level distance = 30pt, sibling distance = 0.1pt, scale =1.0,
   edge from parent/.style = {
   	draw, edge from parent path = {(\tikzparentnode) -- (\tikzchildnode.north)}}]
   \tikzset{every internal node/.style = {draw, circle, black, scale=1}}
   \tikzset{every leaf node/.style = {draw, black, regular polygon, regular polygon sides = 3, inner sep = 1pt, scale = 1}}
   	
   \Tree [.R 
      \edge node[auto=right] {High(a)}; \node [label={below:$[\mathsf{Low}(a), \poweredfloorlog{d} - 1]$}]{$S_1$}; 
  \edge node[auto=left] {High(a) + 1};  \node [label={below:$[0,d-2^{\floor{\log d}}]$}]{$S_2$};
    ]
   \end{tikzpicture}
\caption{\small Structure of a prefix trie built over integers in $(a, a+d) $ \label{fig:prefixTrieSplit}}
\end{figure}

\begin{lemma}
	\label{lemma:trie_similarity}
	Let \trie{2^k}{2^k + d} be a prefix trie built over the binary representations of integers in the range $[2^k, 2^k + d), d < 2^k, k \in \IntegersPositive$. The total number of \mec subtries in \trie{2^k}{2^k + d} is equal to the number of \mec subtries in the prefix trie built over integers in the range $[0, d)$, denoted by \trie{0}{d}. 
\end{lemma}

\begin{proof}
	
	Let the bit string $a_k a_{k-1}, \ldots a_0$ be the binary representation of an integer $a \in [2^k, 2^k + d), d < 2^k, k \in \IntegersPositive$. For all $a \in [2^k, 2^k + d)$, $a_k = 1$.
	Now consider the bit string $0 a_{k-1}, \ldots a_0$. The integer corresponding to this bit string $b = a - 2^k$. 
	
	For all $a \in [2^k, 2^k+d)$, we generate a corresponding integer $b$ by similarly replacing the value of the most significant bit in the binary representation of $a$ with 0. The resulting integers are in the range $[0, d)$.  The prefix trie built over the binary representations of these integers, \trie{0}{d} only differ in the value of the root from \trie{2^k}{2^k + d}. Otherwise the two tries are exactly the same both in node values as well as structure. Since the value of the root has no impact on the overall structure of the trie, the total number of \mec subtries in \trie{2^k}{2^k + d} is the same as the total number of \mec subtries in \trie{0}{d}.
\end{proof}

\begin{lemma}
	\label{lemma:number_enclosing_trees_more_than}
	Let \trie{0}{d} be a prefix trie built over the binary representations of consecutive non-negative integers $[0,d)$. Let 
	\numOfMaximalSubtriesinTree{\trie{0}{d}}  denote the number of maximal enclosing complete subtries in the trie as a function of $d$. Then, \numOfMaximalSubtriesinTree{\trie{0}{d}} = \bigTheta{\log d}.
\end{lemma}

\begin{proof}
	Observe first that \numOfMaximalSubtriesinTree{\trie{0}{2^k - 1}} = 1 for all $k \in \NaturalNumbers$. Next observe that \trie{0}{d} contains two non-overlapping subtries \trie{0}{\poweredfloorlog{d} - 1} and \trie{\poweredfloorlog{d}}{d}. Consequently, we can write 
	
	\begin{center}
		\small
		\numOfMaximalSubtriesinTree{\trie{0}{d}} = \numOfMaximalSubtriesinTree{\trie{0}{\poweredfloorlog{d}} - 1} +  \numOfMaximalSubtriesinTree{\trie{\poweredfloorlog{d}}{d}} = \numOfMaximalSubtriesinTree{\trie{\poweredfloorlog{d}}{d}} + 1
	\end{center}
	
	From Lemma \ref{lemma:trie_similarity}, \numOfMaximalSubtriesinTree{\trie{\poweredfloorlog{d}}{d}} = \numOfMaximalSubtriesinTree{\trie{0}{d - \poweredfloorlog{d}}}.
	Thus, 
	
	\begin{center}
		\numOfMaximalSubtriesinTree{\trie{0}{d}} = \numOfMaximalSubtriesinTree{\trie{0}{d - \poweredfloorlog{d}}} + 1
	\end{center}
	
	For all $x \geq 0$, we have $x - 2^{\floor{\log x}} \leq x/2$. Therefore, this recursion terminates in \bigTheta{\log d} calls and \numOfMaximalSubtriesinTree{\trie{0}{d}} = \bigTheta{\log d}.
\end{proof}

\begin{lemma}
	\label{lemma:number_enclosing_trees_less_than}
	Let \trie{a}{2^{\flooredLog{d}}} be a prefix trie over the binary representations of non-negative integers in the range $[a,2^{\flooredLog{d}})$. Let \numOfMaximalSubtriesinTree{\trie{a}{2^{\flooredLog{d}}}} denote the total number of maximal enclosing complete subtries in the trie as a function of $d$. Then,  \numOfMaximalSubtriesinTree{\trie{a}{2^{\flooredLog{d}}}} = \bigO{\log d}.
\end{lemma}

\begin{proof}
	
	Observe that if $a = 2^{\flooredLog{d}} - 1$, then \trie{a}{a} is in itself a maximal enclosing complete subtree. Otherwise, for some $1 \leq k < \flooredLog{\intDistThreshold}$, $\poweredfloorlog{\intDistThreshold} - 1 - 2^{k} \geq a$  such that there does not exists any $k' > k$ and $\poweredfloorlog{d} - 1 - 2{k'} \geq a$. Then $[a, \poweredfloorlog{d})$ can be partitioned into the non-overlapping ranges $[a, \poweredfloorlog{d} -  2^{k})$ and $[\poweredfloorlog{d}-  2^{k}, 2^{\flooredLog{d}})$. Thus,

	\begin{center}
		\numOfMaximalSubtriesinTree{\trie{a}{2^{\flooredLog{d}}}} = \numOfMaximalSubtriesinTree{\trie{a}{\poweredfloorlog{d} -  2^{k}}} + \numOfMaximalSubtriesinTree{\trie{\poweredfloorlog{d}-  2^{k}}{2^{\flooredLog{d}}}}
	\end{center}
	
	We observe that \numOfMaximalSubtriesinTree{\trie{\poweredfloorlog{d}-  2^{k}}{2^{\flooredLog{d}}}} = 1. Then, 
	
	\begin{center}
		\numOfMaximalSubtriesinTree{\trie{a}{2^{\flooredLog{d}}}} = 1 + \numOfMaximalSubtriesinTree{\trie{a}{\poweredfloorlog{d} -  2^{k}}}
	\end{center} 
	
	Let $j = \poweredfloorlog{d} -a$. Then we observe that $\poweredfloorlog{d} - 2^{k} - a \leq j/2$ because otherwise there exists $k' = \flooredLog{j}> k$ such that $\poweredfloorlog{d} - 2^{\flooredLog{j}} -  a \leq j/2$. Every step of the recursion halves the range of integers and terminates in $\log (\poweredfloorlog{d})$ calls. Thus,  \numOfMaximalSubtriesinTree{\trie{a}{2^{\flooredLog{d}}}} = \bigO{\log d}. 
\end{proof}

\end{document}